\makeatletter
\@ifundefined{ifdebugmode}{\let\ifdebugmode\iffalse}{}%
\@ifundefined{ifblind}{\let\ifblind\iftrue}{}%
\@ifundefined{ifconfstuffs}{\let\ifconfstuffs\iftrue}{}%
\@ifundefined{ifappendix}{\let\ifappendix\iftrue}{}%
\makeatother

\let\ifdebug\iffalse\let\ifblind\iffalse\let\ifconfstuffs\iffalse\let\ifappendix\iftrue

\ifconfstuffs
  \ifblind
    \documentclass[sigconf,9pt,anonymous]{acmart}
  \else
    \documentclass[sigconf,9pt]{acmart}
  \fi
\else
  \documentclass[sigconf,9pt,nonacm]{acmart}
\fi

\usepackage{cmll}
\usepackage{tikz}
\usepackage{graphics}
\usepackage[all]{xy}
\usepackage{tikz-cd}
\usepackage{cleveref}
\usepackage{csquotes}
\usepackage{thm-restate}

\usepackage{mathtools}
\usepackage{xspace}
\usepackage{mathpartir}
\usepackage{bm}
\usepackage{enumitem}
\usepackage{subcaption}

\newcommand{\PCF}{\mathsf{PCF}}

\newcommand{\finv}[1]{#1^{{-}1}}
\newcommand{\set}[1]{\{#1\}}

\newcommand\Sgroup{\operatorname{\mathcal S}}
\newcommand\permcat{\mathbb{P}}
\newcommand\permplus{\oplus}
\newcommand\nset[1]{[#1]}
\newcommand\plact[2]{#1 \circledless #2}

\newcommand{\intr}[1]{\llbracket #1 \rrbracket}
\newcommand{\lin}{\multimap}
\let\linto\lin
\newcommand{\Mf}{\mathcal{M}}
\newcommand{\iso}{\cong}
\newcommand{\bij}{\simeq}

\let\C\undefined
\newcommand{\C}{\mathcal{C}}
\newcommand{\catC}{\mathcal{C}}
\newcommand{\catD}{\mathcal{D}}
\renewcommand{\see}{\mathsf{see}}
\newcommand{\display}{\partial}
\newcommand\stdisp[2][]{\display^{#2}\ifempty{#1}{}{_{#1}}}
\newcommand\stdispoc[2][]{\display^{#2,\oc}\ifempty{#1}{}{_{#1}}}
\newcommand\altidisp[2]{#2}
\newcommand{\wit}{\mathsf{wit}}

\newcommand{\pwit}{\mathsf{wit}^{+}}
\newcommand{\spwit}{\text{$\sim$-$\mathsf{wit}$}^{+}}

\newcommand{\seq}[1]{\langle #1 \rangle}
\newcommand{\itmvec}[1]{\vec{#1}}
\newcommand{\op}{\mathrm{op}}
\newcommand{\cocat}{{\mathrm{co}}}
\newcommand{\ot}{\leftarrow}
\newcommand{\id}{\mathrm{id}}

\newcommand{\Fam}{\mathbf{Fam}}
\newcommand{\N}{\mathbb{N}}
\newcommand{\PreStrat}{\mathsf{PreStrat}}
\newcommand{\bS}{\mathbf{S}}

\let\U\undefined
\newcommand{\U}{\mathbf{U}}

\newcommand{\T}{\mathbf{T}}

\newcommand{\tensor}{\otimes}

\newcommand{\Thin}{\mathbf{Thin}}
\newcommand{\Thinb}{\mathbf{Thin}_{\oc}}

\newcommand{\tuple}[1]{\langle #1 \rangle}
\newcommand{\sym}{\cong}

\newcommand{\tto}{\Rightarrow}

\newbox\privcpl
\setbox\privcpl\hbox{$\bar l$}

\newbox\privcpr
\setbox\privcpr\hbox{$\bar r$}

\DeclareFontEncoding{LS1}{}{}
\DeclareFontSubstitution{LS1}{stix}{m}{n}

\newcommand{\evm}{\mathbf{ev}}

\newcommand\pl{{l}} 
\newcommand\pr{{r}} 

\newcommand\termcat{\mathrm{1}}
\newcommand\unit[1]{\id_{#1}}
\newcommand\seely{s}

\renewcommand{\v}{\mathsf{v}}
\newcommand{\w}{\mathsf{w}}
\newcommand{\m}{\mathsf{m}}

\newcommand{\q}{\mathsf{q}}

\newcommand{\e}{\mathsf{e}}
\newcommand{\pid}{\mathsf{pid}}
\newcommand{\pcomp}{\mathsf{pcomp}}
\newcommand{\Id}{\mathsf{Id}}
\newcommand{\syms}[1]{|#1|}
\newcommand{\ca}{\mathbf{a}}
\newcommand{\cb}{\mathbf{b}}
\newcommand{\cc}{\mathbf{c}}
\newcommand{\cs}{\mathbf{s}}

\newcommand{\class}[1]{\overline{#1}}
\newcommand{\rep}[1]{\underline{#1}}
\newcommand{\Ni}{\mathbb{N}_{\infty}}
\newcommand{\bt}{\mathbf{w}}

\newcommand\lstconsr{::}
\newcommand\lstconsrp[2]{#1 :: #2}

\newcommand\synto{\rightarrow}

\newcommand\judge{\vdash}
\newcommand\judgem{\vdash}
\newcommand\judgel{\judge}

\newcommand{\itype}[1]{\textcolor{brown}{#1}}
\newcommand{\itmor}[1]{\textcolor{violet}{#1}}
\newcommand{\rest}[1]{\textcolor{teal}{#1}}
\newcommand\simplefj[3]{#1 \judge #2 : #3} 
\newcommand\itrefinfj[3][]{\ifempty{#1}{}{#1 :} \textcolor{brown}{#2} \refinl #3} 
\newcommand\itmrefinfj[5][]{\ifempty{#1}{}{#1 :} \textcolor{violet}{#2} \itsco%
  \textcolor{brown}{#3} \textcolor{violet}{\;\itsto\;}%
  \textcolor{brown}{#4} \ifempty{#5}{}{\refinl #5}} 
\newcommand\rtrefinshort[2]{\textcolor{teal}{#1} \refin #2}
\newcommand\itfj[5]{\textcolor{brown}{#1} \refin #2 \judge #3 : \textcolor{brown}{#4} \refin #5}
\makeatletter
\newcommand\itmfj[7]{\textcolor{violet}{#1} \ifempty{#2}{}{\refin #2} \judge #3 : \textcolor{violet}{#4} \itsco \textcolor{brown}{#5} \textcolor{violet}{\;\itsto\;} \textcolor{brown}{#6} \refin #7}
\newcommand\rtfj[6]{\textcolor{brown}{#1} \ifempty{#2}{}{\refin #2} \judge \textcolor{teal}{#3} \refin #4 : \textcolor{brown}{#5} \refin #6}
\newcommand\rtcolor[1]{\textcolor{teal}{#1}}
\makeatother
\makeatletter
\define@key{itfj}{rctxt}{\def\itfj@rctxt{#1}}
\define@key{itfj}{ctxt}{\def\itfj@ctxt{#1}}
\define@key{itfj}{term}{\def\itfj@term{#1}}
\define@key{itfj}{it}{\def\itfj@it{#1}}
\define@key{itfj}{type}{\def\itfj@type{#1}}
\newcommand{\itfjnamed}[1]{%
  \begingroup%
  \setkeys{itmfj}{#1}
  \itfj{\itfj@rctxt}{\itfj@ctxt}{\itfj@term}{\itfj@it}{\itfj@type}
  \endgroup%
}
\define@key{itmfj}{ctxt}{\def\itmfj@ctxt{#1}}
\define@key{itmfj}{rctxt}{\def\itmfj@rctxt{#1}}
\define@key{itmfj}{term}{\def\itmfj@term{#1}}
\define@key{itmfj}{itm}{\def\itmfj@itm{#1}}
\define@key{itmfj}{itsrc}{\def\itmfj@itsrc{#1}}
\define@key{itmfj}{ittgt}{\def\itmfj@ittgt{#1}}
\define@key{itmfj}{type}{\def\itmfj@type{#1}}
\newcommand{\itmfjnamed}[1]{%
  \begingroup%
  \setkeys{itmfj}{#1}
  \@ifundefined{itmfj@ctxt}{%
    \itmfj{\itmfj@rctxt}{}{\itmfj@term}{\itmfj@itm}{\itmfj@itsrc}{\itmfj@ittgt}{\itmfj@type}%
  }{%
    \itmfj{\itmfj@rctxt}{\itmfj@ctxt}{\itmfj@term}{\itmfj@itm}{\itmfj@itsrc}{\itmfj@ittgt}{\itmfj@type}%
  }%
  \endgroup%
}
\define@key{rtfj}{rctxt}{\def\rtfj@rctxt{#1}}
\define@key{rtfj}{ctxt}{\def\rtfj@ctxt{#1}}
\define@key{rtfj}{rt}{\def\rtfj@rt{#1}}
\define@key{rtfj}{term}{\def\rtfj@term{#1}}
\define@key{rtfj}{it}{\def\rtfj@it{#1}}
\define@key{rtfj}{type}{\def\rtfj@type{#1}}
\newcommand{\rtfjnamed}[1]{%
  \begingroup%
  \setkeys{rtfj}{#1}
  \@ifundefined{rtfj@ctxt}{%
    \rtfj{\rtfj@rctxt}{}{\rtfj@rt}{\rtfj@term}{\rtfj@it}{\rtfj@type}%
  }{%
    \rtfj{\rtfj@rctxt}{\rtfj@ctxt}{\rtfj@rt}{\rtfj@term}{\rtfj@it}{\rtfj@type}%
  }%
  \endgroup%
}
\makeatother

\let\synlinto\itslinto
\newcommand\itsstar{\star}
\newcommand\itcl[1]{\class{#1}}

\newcommand\itsstarid{\id_\star}
\newcommand\itsco{::}
\newcommand\itsto{\To}
\newcommand\itsid[1]{\unit{#1}}
\newcommand\itscirc{\circ}

\newcommand\refin{\triangleleft}
\newcommand\refinl{\refin}
 
\newcommand{\itgpd}[1]{\mathbf{IT}(#1)}
\newcommand{\itgpdn}[1]{\mathbf{IT}_-(#1)}
\newcommand{\itgpdp}[1]{\mathbf{IT}_+(#1)}
\newcommand{\itgpdoc}[1]{\mathbf{IT}_{\oc}(#1)}

\newcommand\ctxtwf{\vdash_{\mathrm {wf}}}

\newcommand\rctxtplus{\bullet}
\newcommand\rmctxtplus{\bullet}

\newcommand\rmcplact[2]{#1 \circledless #2}
\newcommand\ctxtcirc{\circ^{\mathrm{ctxt}}}
\newcommand\ctxtid[1]{\unit{#1}^{\mathrm{ctxt}}}
\newcommand\itsctxtmu{\tilde \mu^\mathrm{it}}

\newcommand\seqplus{\bullet}

\newcommand\ctxtplus{\bullet}
\newcommand\ctxtmu{\tilde \mu}

\newcommand\rtmterm[1]{\vec{#1}}

\newcommand\rtseq[1]{\langle #1 \rangle}

\newcommand\rtcl[1]{\class{#1}}

\newcommand\stdit[1]{\boldsymbol{#1}}
\newcommand\stdmit[1]{\boldsymbol{#1}^*}
\newcommand\stdrt[1]{\mathbf{#1}}
\newcommand\stdmrt[1]{\mathbf{#1}^*}
\newcommand\mset[1]{[ #1 ]}

\newcommand\Var{\mathrm{Var}}

\newcommand\terminterp[1]{\llbracket #1 \rrbracket}
\newcommand\terminterpoc[1]{\llbracket #1 \rrbracket^{\oc}}
\newcommand\typeinterp[1]{\llbracket #1 \rrbracket}
\newcommand\ctxtinterp[1]{\llbracket #1 \rrbracket}
\newcommand\terminterpb[1]{(\!| #1 |\!)}
\newcommand\typeinterpb[1]{(\!| #1 |\!)}
\newcommand\ctxtinterpb[1]{(\!| #1 |\!)}

\newcommand{\Rel}{\mathbf{Rel}}
\newcommand{\WRel}{\mathbf{WRel}}
\newcommand{\Set}{\mathbf{Set}}
\newcommand{\Ob}{\mathbf{Ob}}
\newcommand{\Ar}{\mathbf{Ar}}

\newcommand{\Cat}{\mathbf{Cat}}
\newcommand{\Gpd}{\mathbf{Gpd}}
\newcommand{\Span}{\mathbf{Span}}
\newcommand{\Sym}{\mathbf{Sym}}
\newcommand{\Symp}{\mathbf{Sym}^+}

\newcommand{\Dist}{\mathbf{Dist}}
\newcommand{\Esp}{\mathbf{Esp}}
\newcommand{\Ren}{\mathbf{Ren}}

\newcommand\Bicat{\mathbf{Bicat}}
\mathchardef\mhyphen="2D
\newcommand\pmFunct{\pm\mhyphen\mathbf{Funct}}
\newcommand\src{\partial^-}
\newcommand\tgt{\partial^+}
\newcommand{\dom}{\mathsf{dom}}
\newcommand{\cod}{\mathsf{cod}}
\newcommand{\Res}{\mathbf{Res}}
\newcommand{\bgamma}{\boldsymbol{\gamma}}
\newcommand{\bmu}{\boldsymbol{\mu}}
\newcommand{\bnu}{\boldsymbol{\nu}}

\newcommand\stsupp[1]{#1}

\tikzcdset{%
   myname/.code args={#1}{%
     \tikzcdset{""{auto=false,name=#1}}%
   }
}
\makeatletter
\pgfarrowsdeclare{tripleimplies}{tripleimplies}
{
 \pgfmathsetlength{\pgfutil@tempdima}{.25\pgflinewidth+.25*\pgfinnerlinewidth}%
 \pgfmathsetlength{\pgfutil@tempdimb}{.5\pgflinewidth-.5*\pgfinnerlinewidth}%
 \pgfarrowsrightextend{2.06\pgfutil@tempdima+.5\pgfutil@tempdimb}
 \pgfarrowsleftextend{-1.36\pgfutil@tempdima-.5\pgfutil@tempdimb}
}
{
 \pgfmathsetlength{\pgfutil@tempdima}{.25\pgflinewidth+.25*\pgfinnerlinewidth}%
 \pgfmathsetlength{\pgfutil@tempdimb}{.5\pgflinewidth-.5*\pgfinnerlinewidth}%
 \pgftransformxshift{.06\pgfutil@tempdima}
 \pgfsetlinewidth{\pgfutil@tempdimb}
 \pgfsetdash{}{+0pt}
 \pgfsetroundcap
 \pgfsetroundjoin
 \pgfpathmoveto{\pgfpoint{-1.4\pgfutil@tempdima}{2.65\pgfutil@tempdima}}
 \pgfpathcurveto
 {\pgfpoint{-0.75\pgfutil@tempdima}{1.25\pgfutil@tempdima}}
 {\pgfpoint{1\pgfutil@tempdima}{0.05\pgfutil@tempdima}}
 {\pgfpoint{2\pgfutil@tempdima}{0pt}}
 \pgfpathcurveto
 {\pgfpoint{1\pgfutil@tempdima}{-0.05\pgfutil@tempdima}}
 {\pgfpoint{-.75\pgfutil@tempdima}{-1.25\pgfutil@tempdima}}
 {\pgfpoint{-1.4\pgfutil@tempdima}{-2.65\pgfutil@tempdima}}
 \pgfusepathqstroke
 \pgfpathmoveto{\pgfpoint{0\pgfutil@tempdima}{0\pgfutil@tempdima}}
 \pgfpathlineto{\pgfpoint{2\pgfutil@tempdima}{0\pgfutil@tempdima}}
 \pgfusepathqstroke
}%
\makeatother
\tikzset{tikztriple/.style={preaction={draw,tikzcd
cap-tripleimplies,shorten
     >=0pt,double distance=3pt},tikzcd cap-round cap,shorten >=3.5pt}}
\tikzcdset{%
 Rrightarrow/.code={\tikzset{tikztriple}}
}
\tikzcdset{%
 cs/.code args={#1}{%
   \tikzcdset{column sep={#1}}%
 }
}
\tikzcdset{%
 rs/.code args={#1}{%
   \tikzcdset{row sep={#1}}%
 }
}
\makeatletter
\def\tikzcd@sep@bo#1#2{%
 \pgfkeysifdefined{/tikz/commutative diagrams/#1 sep/#2}%
 {\pgfkeysgetvalue{/tikz/commutative diagrams/#1 sep/#2}\tikzcd@temp%
   \edef\tikzcd@temp{{\tikzcd@temp,between origins}}%
   \pgfkeysalso{/tikz/#1 sep/.expand once=\tikzcd@temp}}%
 {\pgfkeysalso{/tikz/#1 sep={#2,between origins}}}}
\tikzcdset{%
 rsbo/.code={\tikzcd@sep@bo{row}{#1}}}
\tikzcdset{%
 csbo/.code={\tikzcd@sep@bo{column}{#1}}}
\tikzcdset{%
 sbo/.code={\tikzcd@sep@bo{row}{#1}\tikzcd@sep@bo{column}{#1}}}
\makeatother

\newcommand\phar[1][]{\ar[phantom,#1]}

\newcommand\tikzcdin[2][column sep=small]{
 \begin{tikzcd}[cramped,ampersand replacement=\&,#1]
   #2 
 \end{tikzcd}%
}
\newcommand\drcorner{\rotatebox[origin=c]{180}{$\ulcorner$}}

\newcommand{\xto}{\xrightarrow}

\newcommand{\To}{\Rightarrow}

\newcommand{\TO}{\Rrightarrow}

\makeatletter
\newcommand{\xRrightarrow}[2][]{\ext@arrow 0359\Rrightarrowfill@{#1}{#2}}
\newcommand{\Rrightarrowfill@}{\arrowfill@\equiv\equiv\Rrightarrow}
\makeatother

\makeatletter
\newsavebox\my@boxdcorner
\savebox\my@boxdcorner{\begin{tikzpicture}[rotate=45,x=3.85pt,y=3.85pt]
  \draw[line cap=round,line join=round] (0,1) -- (0,0) -- (1,0);
\end{tikzpicture}}
\newcommand\dcorner{\usebox\my@boxdcorner}
\newsavebox\my@boxucorner
\savebox\my@boxucorner{\begin{tikzpicture}[rotate=-135,x=3.85pt,y=3.85pt]
    \draw[line cap=round,line join=round] (0,1) -- (0,0) -- (1,0);
  \end{tikzpicture}}
\newcommand\ucorner{\usebox\my@boxucorner}
\newbox\my@privdorthon
\setbox\my@privdorthon\hbox to 9pt{$\perp$\hss$\perp$}
\newbox\my@privdorthos
\setbox\my@privdorthos\hbox to 7pt{$\scriptstyle\perp$\hss$\scriptstyle\perp$}
\newbox\my@privdorthoss
\setbox\my@privdorthoss\hbox to 6pt{$\scriptscriptstyle\perp$\hss$\scriptscriptstyle\perp$}
\newbox\my@privdorthod
\setbox\my@privdorthod\hbox to 9pt{$\displaystyle\perp$\hss$\displaystyle\perp$}
\newcommand\storthob{\ensuremath{\mathbin{\mathchoice{\copy\my@privdorthod}{\copy\my@privdorthon}{\copy\my@privdorthos}{\copy\my@privdorthoss}}}}
\newcommand{\pperp}{\Bot}
\makeatother

\newcommand\co\colon
\newcommand\cA{A}

\newcommand{\resp}{resp.\xspace}

\crefname{enumi}{}{}
\Crefname{enumi}{}{}
\crefname{lemma}{Lemma}{Lemmas}
\Crefname{lemma}{Lemma}{Lemmas}

\newcommand\nbd{\protect\nobreakdash}
\newcommand\sethyp[2]{%
  \expandafter\def\csname nbd#1\endcsname{\nbd-#2\xspace}%
  \expandafter\ifx\csname #1\endcsname\relax
  \expandafter\def\csname #1\endcsname{\nbd-#2\xspace}%
  \fi
}
\sethyp{cell}{cell}
\sethyp{cells}{cells}
\sethyp{category}{ca\-te\-gory}
\sethyp{categories}{ca\-te\-gories}
\sethyp{categorical}{ca\-te\-go\-rical}
\sethyp{pullback}{pull\-back}
\sethyp{pullbacks}{pull\-backs}
\sethyp{groupoid}{grou\-poid}
\sethyp{groupoids}{grou\-poids}
\sethyp{composition}{com\-po\-si\-tion}
\sethyp{globular}{glo\-bu\-lar}
\sethyp{functor}{fun\-ctor}
\sethyp{functors}{fun\-ctors}
\sethyp{composable}{com\-po\-sa\-ble}
\sethyp{isomorphism}{iso\-mor\-phism}
\sethyp{term}{term}
\sethyp{morphism}{mor\-phism}
\sethyp{morphisms}{mor\-phisms}
\newcommand\trms{\nbd-terms\xspace}

\sethyp{natural}{na\-tu\-ral}
\sethyp{naturality}{na\-tu\-ra\-li\-ty}
\sethyp{dimensional}{di\-men\-sio\-nal}
\sethyp{transformation}{trans\-for\-ma\-tion}
\sethyp{transformations}{trans\-for\-ma\-tions}
\sethyp{modification}{mo\-di\-fi\-ca\-tion}
\sethyp{modifications}{mo\-di\-fi\-ca\-tions}
\DeclareRobustCommand\pmfunctor{$\pm$\functor}
\DeclareRobustCommand\pmfunctors{$\pm$\functors}
\DeclareRobustCommand\pmtransformation{$\pm$\transformation}
\DeclareRobustCommand\pmtransformations{$\pm$\transformations}
\DeclareRobustCommand\pmmodification{$\pm$\modification}
\DeclareRobustCommand\pmmodifications{$\pm$\modifications}
\DeclareRobustCommand\lterm{$\lambda$\term}
\DeclareRobustCommand\lterms{$\lambda$\trms}

\makeatletter
\@ifundefined{ifdebugmode}{\let\ifdebugmode\iffalse}{}
\makeatother
\def\ifparam#1#2#3{\csname if#1\endcsname #2\else #3\fi}
\newcommand{\activatecomments}{\let\ifdebugmode\iftrue}
\newcommand{\TODO}[1]{\ifparam{debugmode}{\marginpar{\tiny #1}}{}}
\newcommand{\todo}[1]{\ifparam{debugmode}{\marginpar{\tiny #1}}{}}

\newcommand{\simon}[1]{\ifparam{debugmode}{\TODO{\color{blue} SF: #1}}{}}
\newcommand{\pierre}[1]{\ifparam{debugmode}{\TODO{\color{red} PC: #1}}{}}

\newcommand{\incomment}[1]{}
\newcommand{\syndef}{\mathrel{\vcenter{\hbox{$::$}}{=}}}
\newcommand\eq{\enquote}

\newcommand{\qp}[1]{\quad #1\quad}
\newcommand{\qqp}[1]{\qquad #1\qquad}

\newcommand{\qand}{\quad\text{and}\quad}
\newcommand{\qqand}{\qquad\text{and}\qquad}
\newcommand\zbox[1]{\makebox[0pt][l]{#1}}
\makeatletter
\DeclareRobustCommand\ifdash[3]{%
  \def\ifdash@o{-}%
  \def\ifdash@t{#1}%
  \def\ifdash@yes{#2}%
  \def\ifdash@no{#3}%
  \ifx\ifdash@o\ifdash@t
  \ifdash@yes
  \else
  \ifdash@no
  \fi
}
\makeatother
\newcommand\ifempty[3]{%
  \def\ifemptytemp{#1}%
  \ifx\ifemptytemp\empty%
  #2%
  \else%
  #3%
  \fi%
}

\newcommand{\D}[1]{\|#1\|}

\AtBeginDocument{%
  \providecommand\BibTeX{{%
      Bib\TeX}}}

\ifconfstuffs
  \setcopyright{acmlicensed}
  \copyrightyear{2018}
  \acmYear{2018}
  \acmDOI{XXXXXXX.XXXXXXX}
\else
\fi

\begin{document}

\title{An Analysis of Symmetry in Quantitative Semantics}

\author{Pierre Clairambault}
\affiliation{%
  \institution{Aix Marseille Univ, Université de Toulon, CNRS, LIS}
  \city{Marseille}
  \country{France}
}
\email{Pierre.Clairambault@cnrs.fr}

\author{Simon Forest}
\affiliation{%
  \institution{Aix Marseille Univ, CNRS, LIS}
  \city{Marseille}
  \country{France}
}
\email{Simon.Forest@univ-amu.fr}

\renewcommand{\shortauthors}{Clairambault\and Forest}

\begin{abstract}
  Drawing inspiration from linear logic, \emph{quantitative semantics}
aim at representing quantitative information about programs and their
executions: they include the relational model and its numerous
extensions, game semantics, and syntactic approaches such as
non-idempotent intersection types and the Taylor expansion of
$\lambda$-terms. The crucial feature of these models is that 
programs are interpreted as witnesses which consume \eq{bags} of
resources. 

\eq{Bags} are often taken to be finite multisets, \emph{i.e.} quotiented
structures. Another approach typically seen in categorifications of the
relational model is to work with unquotiented structures (\emph{e.g.}
sequences) related with explicit morphisms referred to here as
\emph{symmetries}, which express the exchange of resources.
Symmetries are obviously at the core of these categorified models, but
we argue their interest reaches beyond those -- notably, symmetry
\emph{leaks} in some non-categorified quantitative models (such as the
weighted relational model, or Taylor expansion) under the form of
numbers whose combinatorial interpretation is not always clear.

In this paper, we build on a recent bicategorical model called \emph{thin
spans of groupoids}, introduced by Clairambault and Forest.
Notably, thin spans feature a decomposition of symmetry into two
sub-groupoids of \emph{polarized} -- \emph{positive} and
\emph{negative} -- symmetries. We first construct a variation of the
original exponential of thin spans, based on sequences rather than
families. Then we give a syntactic characterisation of the
interpretation of simply-typed $\lambda$-terms in thin spans, in terms
of rigid intersection types and rigid resource terms. Finally, we
formally relate thin spans with the weighted relational model and
generalized species of structure. This allows us to show how some
quantities in those models reflect polarized symmetries: in particular
we show that the weighted relational model counts witnesses from
generalized species of structure, \emph{divided} by the cardinal of a
group of positive symmetries.

%
%

\end{abstract}


\ifconfstuffs
\else
\fi

\ifconfstuffs
\else
\fi


\maketitle

\section{Introduction}

\emph{Denotational semantics} is an approach to the semantics of
programming languages that consists in associating to every program a
denotation in an adequate mathematical universe; crucially this is done
compositionally, by induction on syntax. Most denotational models are
\emph{qualitative}: a term $\vdash M : A \to B$ is typically
represented by a function from the denotation of $A$ to the denotation
of $B$, giving us the input/output behaviour of $M$, but omitting
\emph{quantitative} information, such as resources, time,
probabilities\ldots

Within denotational semantics, \emph{quantitative semantics} is a
family of models whose distinguishing feature is to record quantitative
information -- first and foremost, displaying \emph{how many times} a
function $\vdash M : A \to B$ must evaluate its argument in
order to produce a given result.  Originally prompted by Girard's
linear logic \cite{DBLP:journals/tcs/Girard87}, quantitative semantics
has developed into a wide research topic with numerous models and
approaches, including the relational model
\cite{DBLP:journals/tcs/Girard87} and its weighted
\cite{DBLP:journals/tcs/Lamarche92,DBLP:conf/lics/LairdMMP13,DBLP:journals/mscs/Ehrhard02}
or categorical \cite{fiore2008cartesian,DBLP:conf/lics/ClairambaultF23}
extensions, resource terms and the Taylor expansion of $\lambda$-terms
\cite{DBLP:journals/tcs/EhrhardR03}, non-idempotent intersection types
\cite{de2007semantiques,DBLP:conf/tacs/Gardner94}, game semantics
\cite{DBLP:journals/iandc/HylandO00,DBLP:journals/iandc/AbramskyJM00},
and others. This is not merely a subjective methodological difference:
quantitative models are well-suited to model quantitative features
such as probabilistic \cite{DBLP:journals/jacm/EhrhardPT18} or quantum
\cite{DBLP:conf/popl/PaganiSV14} primitives, reflecting quantitative
property such as execution time \cite{DBLP:journals/tcs/CarvalhoPF11},
or the number of non-deterministic branches
\cite{DBLP:conf/lics/LairdMMP13}, and many others. 

To keep track of quantitative information, quantitative models must
represent all individual resource accesses, but this is trickier than
it might seem. Linear logic decomposes the intuitionistic arrow $A \to
B$ as $\oc A \lin B$ where $\lin$ is the \emph{linear arrow} (for
functions calling their argument exactly once), and $\oc$ is the
\emph{exponential modality}, allowing arbitrary duplications of
resources. Typically, the difficulty in designing a quantitative model
arises with handling the exponential: how to keep track of all
individual resource accesses while ensuring the laws required for a
$\oc$ in models of linear logic? 

\paragraph{Quotients.}
If resource accesses in $\oc A$ are ordered in a sequence
\[
\seq{\alpha_1, \dots, \alpha_n}\,,
\]
then this will generally fail the commutations laws for the
exponential, which require a \emph{commutative comonoid}
\cite{panorama}\footnote{Though some games models, notably simple games
with the Hyland exponential \cite{hyland1997game}, get away with that
exploiting that copy accesses are totally chronologically ordered.}. So
sequences are often quotiented out by commutativity, as in the
relational model~\cite{DBLP:journals/tcs/Girard87} (and in general the
so-called web-based models of linear logic), where $\oc A =
\Mf(A)$ the set of finite multisets. This quotient is also
found in quantitative notions of program approximation: for instance,
the \emph{Taylor expansion} of $\lambda$-terms
\cite{DBLP:journals/tcs/EhrhardR03} approximates $\lambda$-terms via
the \emph{resource calculus}, a strongly finitary calculus
where an application $M\,N$ from the $\lambda$-calculus is approximated
with
\[
  m\,[n_1, \dots, n_k]
\]
the application of a resource term $m$, approximating $M$, to a
\emph{finite multiset} of resource terms $n_1, \dots, n_k$, all
approximating $M$. This expresses one of the possible behaviours of
$M\,N$, where $M$ will call its argument \emph{exactly $k$ times}, each
call associated to one of the $n_i$'s.

This quotient, at the heart of quantitative semantics, is by no means
innocent: in situations when quantitative semantics manipulate
numerical coefficients, the underlying symmetries on multisets
\emph{leak}, yielding scalars which are not clearly related to the
computational situation, but instead reflect some aspect of its
underlying symmetries. For instance, the relational model weighted by
(completed) natural numbers \cite{DBLP:conf/lics/LairdMMP13}, which in
this paper we refer to as $\WRel_\oc$, counts distinct execution
branches for non-deterministic programs when applied at ground type.
But at higher-order type it yields non-trivial coefficients, even for
plain simply-typed $\lambda$-terms: what do these numbers mean? Are
those numbers related to the coefficients appearing in the Taylor
expansion of $\lambda$-terms?

\paragraph{Rigid structures} It is tempting to avoid these quotients:
in the quantitative semantics literature, the corresponding structures
are often called \emph{rigid}. Developping rigid models is subtle; for
instance naively replacing finite multisets with sequences in the
resource calculus yields a non-confluent reduction
\cite{DBLP:journals/lmcs/OlimpieriA22}; while naive rigid
non-idempotent intersection types fail subject reduction.

Proper treatments of rigid structures may be found in 
\emph{categorifications} of the relational model, the prime example
being the cartesian closed bicategory $\Esp$ of generalized species of
structure \cite{fiore2008cartesian}. There, types are interpreted as
\emph{categories} (or \emph{groupoids}) and the exponential $\oc A$ is
the \emph{free strict symmetric monoidal category} $\Sym(A)$ on $A$,
where objects are sequences $\seq{a_1, \dots, a_n}$ of objects of $A$,
and where a morphism from $\seq{a_1, \dots, a_n}$ to $\seq{a'_1, \dots,
a'_m}$ is a bijection $\sigma : n \bij m$ along with $f_i :
a_i \to a'_{\sigma(i)}$ in $A$ for all $1\leq i \leq n$. 
A term $\Gamma \vdash M : A$ is interpreted as a \emph{distributor}
from $\Sym(\intr{\Gamma})$ to $\intr{A}$, \emph{i.e.}
\[
\intr{M}_{\Esp} : \Sym(\intr{\Gamma})^{\op} \times \intr{A} \to \Set\,,
\]
a functor which to $\vec{\gamma} \in \Ob(\Sym(\intr{\Gamma}))$ and
$a \in \Ob(\intr{A})$ associates a set $\intr{M}_{\Esp}(\gamma, a)$ of
\emph{witnesses} -- crucially, $\intr{M}_{\Esp}$ also has a functorial
action, making the symmetries (morphisms) of $\Sym(\intr{\Gamma})$ and
$\intr{A}$ \emph{act} on witnesses. Tsukada \emph{et al.}
\cite{DBLP:conf/lics/TsukadaAO17} and Olimpieri \cite{ol:intdist} have
studied the nature of these witnesses, showing that they can be
regarded as terms of a rigid resource calculus. Their calculi are not
the naive rigid resource calculus mentioned above: they refine it by
letting resource terms carry \emph{morphisms}/\emph{symmetries} from
the types -- but the precise location of these symmetries in the term
is irrelevant, and it must be forgotten by yet another quotient!

Nevertheless, as $\Esp$ is a generalization of $\Rel$ properly
accounting for symmetries, it looks a natural candidate to illuminate
the scalars arising from the weighted relational model: we may
expect
\begin{eqnarray}
(\intr{M}_{\WRel_\oc})_{\gamma, a} &=& \#\,
(\intr{M}_{\Esp})(\gamma, a)
\label{eq:esp_wrel}
\end{eqnarray}
(conflating for now objects and symmetry classes). But this fails, and
we shall see that the link between the two involves data that is
missing from the theory of $\Esp$: \emph{polarized symmetries}.

\paragraph{Contributions} Recently, Clairambault and Forest have
introduced a new bicategorical model $\Thin$, called \emph{thin spans
of groupoids} \cite{DBLP:conf/lics/ClairambaultF23}, also a
categorification of the relational model, inspired concurrent game
semantics \cite{cg2} -- our first contribution is to show that it
supports an exponential based on sequences rather than families.

We then delve deeper into the interpretation of the simply-typed
$\lambda$-calculus in the Kleisli bicategory $\Thinb$.
Just like for $\Esp$ \cite{DBLP:conf/lics/TsukadaAO17,ol:intdist}, we
show that an intersection type system (and matching resource terms) is
implicit in thin spans. Perhaps surprisingly, it turns out to be the
naive rigid intersection type system discussed above, obtained by
merely replacing finite multisets with sequences (or the similarly
naive rigid resource calculus), not carrying any symmetries, and
without any quotient. Though subject reduction fails on the nose, our
results entail that it does hold in a relaxed sense, \emph{up to
symmetry}. Beyond just characterising the witnesses as in
\cite{DBLP:conf/lics/TsukadaAO17,ol:intdist}, we go further and also
give a syntactic description of \emph{symmetries} between derivations,
obtaining a syntactic description of the full groupoid obtained as the
interpretation of a term.

A central feature of $\Thin$ is that objects are certain groupoids $A$
admitting two sub-groupoids $A_-$ and $A_+$, respectively of
\emph{negative} and \emph{positive} symmetries.  Those are reminiscent
from ideas in game semantics: \emph{negative} symmetries exchange
resources controlled by the environment, while \emph{positive}
symmetries exchange resources controlled by the program. Not every
symmetry is negative or positive, but every symmetry factors uniquely
as a negative composed with a positive. Far from being a technicality
of the model construction, we argue that these polarized sub-symmetries
are \emph{fundamental}. In particular, they are the key to illuminate
some of the questions mentioned earlier: in this paper, we
characterise the coefficients obtained by $\WRel_\oc$ as counting
witnesses in $\Thinb$ -- \emph{i.e.} rigid resource terms -- \emph{up
to positive symmetry}, or symmetry classes of witnesses -- \emph{i.e.}
standard resource terms -- with a correcting coefficient involving
\emph{negative symmetries}. Drawing inspiration from recent work
linking thin concurrent games with generalized species of
structure~\cite{DBLP:conf/lics/ClairambaultOP23}, we also construct an
interpretation-preserving pseudofunctor from $\Thinb$ to $\Esp$,
allowing us overall to express the coefficients obtained through
$\WRel_\oc$ directly in terms of $\Esp$, correcting \eqref{eq:esp_wrel}
-- again, the correct equation involves polarized symmetries.

\paragraph{Related work} \emph{Polarized symmetries} are central to the
construction of thin spans of groupoids (and before that, thin
concurrent games \cite{cg2}), but they predate those models: to our
knowledge, they first appear in Melliès' approach to uniformity by
bi-invariance, in the setting of asynchronous games
\cite{mellies2003asynchronous}. They also make an appearance in Tsukada
\emph{et al.}'s study of \emph{weighted generalized
species}~\cite{DBLP:conf/lics/TsukadaAO18}, though they are not part of
the general theory but computed \emph{a posteriori} for groupoids
arising from simple types.  

This work is part of an ongoing effort from the community to refine our
understanding of resources in quantitative models, replacing quotients
with rigid structures related with explicit morphisms and explore the
corresponding categorical structures. Aside from work on generalized
species of structure, a work complementary to ours is Melliès'
\emph{homotopy template games} \cite{DBLP:conf/lics/Mellies19}, also
based on categorical spans, focusing on links with
homotopy theory. 

\paragraph{Outline}
In Section \ref{sec:thin_sym} we recall the definition of $\Thin$
from~\cite{DBLP:conf/lics/ClairambaultF23}, replacing their exponential
with a new one based on $\Sym$. In Section
\ref{sec:intersection-resource}, we give our syntactic characterisation
of the interpretation of the simply-typed $\lambda$-terms in $\Thinb$.
Finally, in Section \ref{sec:relational} we explore the link between
$\Thinb$ and relational models: first the plain relational model
$\Rel$, then the weighted (by completed natural numbers) relational
model $\WRel$, and finally generalized species $\Esp$. 

\section{Thin Spans on Sequences}
\label{sec:thin_sym}

We start with a brief reminder on
$\Thin$~\cite{DBLP:conf/lics/ClairambaultF23}, along with the
definition of the new exponential based on sequences.
%
%
%
In the following, we write $\Gpd$ for the $2$-category of
groupoids, functors between groupoids and natural transformations
between such functors. We will also often call \textbf{symmetries} the
morphisms of a groupoid.

\subsection{Reminder on Thin Spans of Groupoids}


A \textbf{span} from $A$ to $B$ in a category $\C$ is simply a diagram
like
\[
  \begin{tikzcd}[cramped]
    A
    &
    S
    \ar[l,"\display^S_{\altidisp A l}"']
    \ar[r,"\display^S_{\altidisp B r}"]
    &
    B
  \end{tikzcd}
\]
which in $\Set$ (or $\Cat$, of $\Gpd$) is regarded as a
generalized relation: a pair $(a, b)$ may be related via a number of
distinct \textbf{witnesses}, \emph{i.e.} elements $s \in S$ \emph{s.t.}
$\display^S_{\altidisp A l}(s) = a$ and $\display^S_{\altidisp B r}(s) = b$ -- in this
paper, we often write $\display^S_{\altidisp A l}(s) = s_A$ and
$\display^S_{\altidisp B r}(s) =
s_B$, keeping $\display^S_{\altidisp A l}$ and $\display^S_{\altidisp B r}$ implicit.
Here we focus on spans over groupoids: those form a bicategory $\Span$
where objects are groupoids, and a morphism from $A$ to $B$ is a span
$A \ot S \to B$. The identity span $\Id_A$ is $A \ot A \to A$ with two
identity functors, and spans are composed by pullback.  

In $\Span$, the $2$-cells from a span $A \ot S \to B$ to $A \ot T \to
B$ are functors $S \to T$ making the two triangles commute, and their
horizontal composition is given by the universal property of 
pullbacks. Unfortunately, these $2$-cells are too strict for many
purposes; in particular they are incompatible with the laws required
for the exponential modality of linear logic. Alternative $2$-cells
relax the hypothesis that the two triangles commute, asking instead for
\[
\xymatrix@R=5pt@C=20pt{
&S	\ar[dl]
	\ar[dr]
	\ar[dd]\\
A	\ar@{}[r]|{~~~~~~~~\Downarrow\!\!\!\!\!\!}&&
B	\ar@{}[l]|{\!\!\!\!\!\!\Downarrow~~~~~~~~}\\
&T	\ar[ul]
	\ar[ur]
}
\]
two natural isomorphisms. This allows us to relate more spans and
indeed supports the laws for the exponential modality. However, the
universal property of pullbacks then fails to provide a definition
of horizontal composition for those. This mismatch has different
solutions, either replacing the pullbacks with adequate notions of
homotopy pullbacks, or requiring additional fibrational conditions on
spans -- in almost all cases this concretely means importing the
morphisms of groupoids inside witnesses, as in generalized species of
structure or in template games \cite{DBLP:conf/lics/Mellies19}.

In~\cite{DBLP:conf/lics/ClairambaultF23}, an alternative idea was
introduced. In $\Span$, some pullbacks happen to behave well
\emph{w.r.t.} homotopy (they are \emph{bipullbacks}, see below). The
key observation is that as it turns out, the pullbacks arising from the
denotational interpretation of programs actually \emph{always are}
bipullbacks! The bicategory $\Thin$ of \emph{thin spans} captures this
via a biorthogonality construction, morally cutting $\Span$ down and
keeping only certain spans -- those deemed ``uniform'' -- ensuring that
their composition pullbacks are always bipullbacks.

\subsubsection{Uniformity}

Given a groupoid $A$, a \textbf{prestrategy}\footnote{Some
terminology in \cite{DBLP:conf/lics/ClairambaultF23} is
game-theoretic, reflecting the game semantics inspirations.} on $A$ is a pair
$(S,\display^S)$ of a groupoid $S$ and a functor $\display^S \co S \to
A$, the \textbf{display map}. We write $\PreStrat(A)$ for the
class of prestrategies on $A$.

Given two prestrategies $(S,\display^S)$ and $(T,\display^T)$, we write
$(S,\display^S) \perp (T,\display^T)$ (or, more simply, $S \perp T$), when the
following pullback
\begin{equation}
  \label{eq:pb-of-prestrategies}
  \begin{tikzcd}[cramped,rsbo=1.7em,csbo=2.5em]
    &
    P
    \ar[rd,"\pr",dashed] 
    \ar[ld,"\pl"',dashed] 
    \phar[dd,"\dcorner",very near start]
    &
    \\
    S
    \ar[rd,"\display^S"']
    &&
    T
    \ar[ld,"\display^T"]
    \\
    &
    A
  \end{tikzcd}
\end{equation}
is a \textbf{bipullback}. In $\Gpd$, this means that
for every $s \in S, t \in T$ and $\theta \co s_A \to t_B$, there is $u
\co s \to s' \in S$ and $v \co t' \to t \in T$ such that $\theta = v_A
\circ u_A$ in $A$: when two states can synchronize up to symmetry, we
can find symmetric states that can synchronize on the nose, coherently.
Given a set, or even a
class $\bS$ of prestrategies on $A$, we write $\bS^\perp$ for the class
$\{T \in \PreStrat(A) \mid \forall S \in \bS,~ S \perp T\}$.

A \textbf{uniform groupoid} is a pair $A = (A,\U_A)$ where $A$ is a
groupoid and $\U_A \subseteq\PreStrat(A)$ is a class of prestrategies
such that $\bS^{\perp \perp} = \bS$. One can define several
constructions on uniform
groupoids~\cite{DBLP:conf/lics/ClairambaultF23}. The \textbf{dual}
$A^\perp$ of the uniform groupoid $A$ has $(A,\U_A^\perp)$. Given
another uniform groupoid $B = (B,\U_B)$, one can define binary
constructions like the \textbf{tensor} $A \otimes B$ and its de Morgan
dual the \textbf{par} $A \parr B$, both having underlying groupoid $A
\times B$. From these two constructions, one then defines the
\textbf{linear arrow} $A \linto B$ as $A^\perp \parr B$. Finally, the
\textbf{with} $A \with B$ has underlying groupoid $A + B$.

\subsubsection{Spans}
The underlying groupoid of $A \linto B$ is $A \times B$ so that $S \in
\U_{A \linto B}$ is a prestrategy on $A \times B$, equivalently seen as
a span 
\[
A \ot S \to B
\]
in $\Gpd$. In the following, we call such $S$ a \textbf{uniform span} to emphasize
that it is a prestrategy of $\U_{A \linto B}$.
Notably, the identity span on a uniform groupoid $A$, is uniform.
%
Given uniform groupoids $A,B,C$, $S \in \U_{A \linto B}$ and $T\in
\U_{B\linto C}$, the composition via the pullback
\[
  \begin{tikzcd}[cramped,csbo=large,rsbo=normal]
    &&
    T \odot S
    \ar[dl,"\pl"',dashed] 
    \ar[dr,"\pr",dashed] 
    \phar[dd,"\dcorner",very near start]
    &&
    \\
    &
    S	\ar[dl,"\display^S_{\altidisp A l}"']
    \ar[dr,"\display^S_{\altidisp B r}"]
    &&
    T
    \ar[dl,"\display^T_{\altidisp B l}"']
    \ar[dr,"\display^T_{\altidisp C r}"]
    \\
    A&&B&&C \zbox.
  \end{tikzcd}
\]
is uniform (\emph{i.e.} in $U_{A \linto C}$) by
\cite[Lem.~2]{DBLP:conf/lics/ClairambaultF23} -- and the composition
pullback is a bipullback, as stated in our motivation for $\Thin$.

\subsubsection{Morphisms of spans}
As introduced above, uniform spans must be related via adequate notions
of morphisms between spans:
\begin{definition}[{\cite[Def.~1]{DBLP:conf/lics/ClairambaultF23}}]\label{def:weakmor}
  A \textbf{weak morphism} from $A \leftarrow S \rightarrow B$ to $A
  \leftarrow S' \rightarrow B$ is $(F, F^A, F^B)$, with $F^A$
and $F^B$ natural isos, and
  \[
    \begin{tikzcd}[cramped,rsbo=normal]
      &&
      S	\ar[dll,"\display^S_{\altidisp A l}"']
      \ar[drr,"\display^S_{\altidisp B r}"]
      \ar[dd,"F"{description}]
      &&
      \\
      A
      &
      F^A\!\Downarrow\hspace{-15pt} && \hspace{-15pt}\Downarrow \! F^B&B\\
      &&
      S'
      \ar[ull,"\display^{S'}_{\altidisp A l}"]
      \ar[urr,"\display^{S'}_{\altidisp B r}"']
    \end{tikzcd}
  \]

We call this a \textbf{strong morphism} if $F^A$ and $F^B$ are
identities.
\end{definition}

The bipullback property, for the composition pullback, ensures the
existence of candidates for the horizontal composition of weak
morphisms. However, it is not uniquely defined, and the bipullback
property is insufficient to guarantee a canonical choice satisfying the
laws of a bicategory (see~\cite[Par.
III-B4]{DBLP:conf/lics/ClairambaultF23}). We thus need additional
structure in order to ensure the existence of a canonical choice.

%

\subsubsection{Thinness} 
For this we must capture a more subtle property observed in the
denotational interpretation of programs: non-trivial symmetries between
states always originate from the environment -- in a closed world
interaction, no non-trivial symmetry is left. This is called
\emph{thinness}, and again is captured by orthogonality. 

Given a uniform groupoid $A$, $S \in \U_A$ and $T \in \U_{A}^\perp$, we write $S
\pperp T$ when the pullback vertex of~\eqref{eq:pb-of-prestrategies} is a
discrete groupoid. Given a class $\bS \subseteq \U_A$, we write
$\bS^{\pperp}$ for the class $\{T \in \U_{A}^\perp \mid \forall S \in \bS,~ S
\pperp T\}$.

\begin{definition}[{\cite[Def.~10]{DBLP:conf/lics/ClairambaultF23}}]\label{def:thin_gpd}
  A \textbf{thin groupoid} is a tuple $A = (A,A_-,A_+,\U_A,\T_A)$ where
$(A, \U_A)$ is a uniform groupoid, and
\begin{itemize}
  \item $A_-$ and $A_+$ are subgroupoids of $A$ with the same objects,
with embedding functors $\id^-_A \co A_- \to A$ and $\id^+_A
\co A_+ \to A$; 
  \item $\T_A \subseteq \U_A$ is a class of prestrategies such that
    $\T_A^{\pperp\pperp} = \T_A$, satisfying that
$(A_-,\id^-_A) \in \T_A$ and $(A_+,\id^+_A) \in \T_A^{\pperp}$.
  \end{itemize}
\end{definition}

In a groupoid $G$ with $x, y \in G$, we often write $\theta : x \sym_G
y$ to mean that $\theta \in G[x, y]$. For $A$ a thin
groupoid, $\theta : a \sym_A^+ a'$ indicates that $\theta \in A_+[a,
a']$ -- we say that $\theta$ is a \textbf{positive} symmetry --
likewise, $\theta : a \sym_A^- a'$ indicates that $\theta \in A_-[a,
a']$, and we say that $\theta$ is \textbf{negative}. Intuitively, this
polarity tells us who, among the program or the environment, is
responsible for a permutation. If it is a permutation among resources
called upon by the environement (\emph{e.g.}, coming from an occurrence
of $\oc$ in covariant position), then the symmetry is \emph{negative}.
If it permutes resources controlled by the program (\emph{e.g.} with a
$\oc$ in contravariant position), then the symmetry is \emph{positive}.
In general a symmetry may mix the two and can be neither negative nor
positive, but from Defininition \ref{def:thin_gpd} we get:

\begin{lemma}\label{lem:factor}
For any $\theta : a \sym_A a'$ in a thin groupoid $A$, there are unique
$a'' \in A$ and $\theta^+ : a \sym_A^+ a''$, $\theta^- : a'' \sym_A^-
a'$ \emph{s.t.} $\theta = \theta^- \circ \theta^+$. 
\end{lemma}

See~\cite[Lem.~3]{DBLP:conf/lics/ClairambaultF23}.
The constructions introduced before on uniform groupoids
($(-)^\perp, \tensor, \parr, \with$) extend to thin
groupoids~\cite{DBLP:conf/lics/ClairambaultF23}.
%

\subsubsection{Thin spans}
Given thin groupoids $A$ and $B$, a \textbf{thin span} is a prestrategy $S \in
\T_{A \linto B}$. As above the underlying groupoid of $A \linto B$ is $A \times
B$, so $S$ can be seen as a span between $A$ and $B$. Given a thin groupoid $A$,
we have $\Id_A \in \T_{A \linto A}$; and for thin spans $A \ot S \to B$ and $B
\ot T \to C$, we have $T \odot S \in \T_{A \linto C}$
(see~\cite[Prop.~2]{DBLP:conf/lics/ClairambaultF23}).

Together, uniformity and thinness guarantee strong properties for the
composition of thin spans. For thin spans $A \ot S \to B$ and $B \ot T
\to C$, recall that (following the obvious pullback construction in
$\Gpd$) elements of $T \odot S$ are simply pairs $(s, t)$ such that
$s_B = t_B$. However, it is central in the construction of $\Thin$ (in
particular for the horinzontal composition of $2$-cells that we shall
not detail here) that thin spans may synchronize \emph{up to symmetry}:

\begin{lemma}\label{lem:comp_upto_sym}
Consider $A \ot S \to B$ and $B \ot T \to C$ thin spans, $s \in S, t
\in T$, linked with a symmetry $\theta : s_B \sym_B t_B$.

Then there are \emph{unique} $s'\in S, t' \in T$ and $\varphi : s
\sym_S s'$, $\psi : t' \sym_T t$ such that $\varphi_A$ negative,
$\psi_C$ positive, and $\theta = \psi_B \circ \varphi_B$.
\end{lemma}

See~\cite[Lem.~2]{DBLP:conf/lics/ClairambaultF23}.  Another important
consequence of the definition of thin spans is that symmetries
\emph{act} on thin spans:

\begin{restatable}{lemma}{actsym}\label{lem:act_sym}
Consider $A \ot S \to B$ a thin span, $s \in S$, with $\theta_A : a
\sym_A s_A$ and $\theta_B : s_B \sym_B b$. Then, there are unique $s'
\in S$, $\varphi : s \sym_S s'$, $\vartheta_A^-$ and $\vartheta_B^+$
such that the two triangles commute:
\[
\xymatrix@R=5pt{   
&s_A	\ar[dd]^{\varphi_A}\\
a	\ar[ur]^{\theta_A}
	\ar[dr]_{\vartheta_A^-}\\
&s'_A
}
\qquad
\xymatrix@R=5pt{
s_B	\ar[dd]_{\varphi_B}
	\ar[dr]^{\theta_B}\\
&b\\
s'_B	\ar[ur]_{\vartheta_B^+}
}
\]  
\end{restatable}

See Appendix~\ref{app:reindexing}. So $s \in S$ may be reindexed by
symmetries $\theta_A$ and $\theta_B$, though we will not exactly hit
the targets $a$ and $b$: only up to positive (or negative, depending on
the variance) symmetry.

\subsubsection{Positive weak morphisms} This additional structure may
be leveraged to get the canonicity of horizontal composition of
$2$-cells -- modulo a final fine-tuning of their definition:

\begin{definition}
Given two thin groupoids $A$ and $B$, a weak morphism $(F, F^A, F^B)$
between $A$ and $B$ as in \Cref{def:weakmor} is \textbf{positive} when,
for every $s \in S$, $F^B_s \co s_B \sym_B^+ F(s)_B$ and $F^A_s \co s_A
\sym_A^- F(s)_A$. 
\end{definition}

We call it \emph{positive} since it is positive on $A \lin B$.
Positivity lets us use the uniqueness property of Lemma
\ref{lem:comp_upto_sym} to give a \emph{unique} choice for
horizontal composition of positive weak morphisms, and:


\begin{theorem}[{\cite[Thm~2]{DBLP:conf/lics/ClairambaultF23}}]
  There is a bicategory $\Thin$ of thin groupoids, thin spans, and positive weak
  morphisms. The identity on $A$ is $\Id_A$, and the composition of
thin spans is given by plain pullbacks. 
\end{theorem}

\subsection{The $\Sym$ Exponential on $\Thin$}

$\Thin$ was originally developped
using the $\Fam$ functor as exponential, mapping a groupoid $A$ to 
$\Fam(A)$ with objects families $(a_i)_{i\in I}$ indexed
by finite sets of integers $I$. Instead, we consider here the $\Sym$
functor (used as exponential modality on distributors to construct
generalized species of structure), which extends to groupoids the list
functor of $\Set$. This seems a minor difference since
$\Fam$ and $\Sym$ are equivalent as endofunctors of $\Gpd$, but it is
actually a non-trivial shift since thin spans do not respect the
principle of equivalence, by relying on strict pullbacks in a
$2$-categorical setting.

\subsubsection{The $\Sym$ monad on $\Gpd$}

We start by considering the functor
\[
  \Sym \co \Gpd \to \Gpd
\]
mapping $A$ to the free strict symmetric monoidal groupoid $\Sym(A)$.
Concretely, the objects of $\Sym(A)$ are sequences $\seq{a_i}_{i \in
\set{1,\ldots,n}} = \seq{a_1,\ldots,a_n}$ of objects of $A$, and its
morphisms from $\seq{a_1,\ldots,a_n}$ to $\seq{b_1,\ldots,b_m}$ are
pairs $(\pi,\seq{f_i}_{i \in \set{1,\ldots,n}})$ where $\pi$ is a
bijection between $\set{1,\ldots,n}$ and $\set{1,\ldots,m}$, and
$\seq{f_i}_i$ is a sequence of morphisms $f_i \co a_i \to b_{\pi(i)}$
for $i \in \set{1,\ldots,n}$. $\Sym$ can be extended to a monad
$(\Sym,\eta,\mu)$ on $\Gpd$: on objects, the unit $\eta_A \co A \to
\Sym(A)$ maps $a \in A$ to $\seq{a}$, and $\mu_A \co \Sym(\Sym(A)) \to
\Sym(A)$ concatenates sequences -- this extends to symmetries as
expected. 

\subsubsection{The pseudocomonad}

The definition of a pseudocomonad $\oc$ for $\Thin$ based on $\Sym$ is done as
in~\cite[Sec.~IV-A]{DBLP:conf/lics/ClairambaultF23}, we recall the salient
elements here.
Given $A = (A,A_-,A_+,\U_A,\T_A)$, we set
\[
  \oc A = (\Sym(A), \Sym(A_-),\Symp(A_+), (\Sym\,\U_A)^{\perp\perp},
(\Sym\,\T_A)^{\pperp\pperp})
\]
where $\Symp(A_+)$ is a subgroupoid of $\Sym(A_+)$ with the same
objects but morphisms only the $(\id,\seq{f_i}_i)$; where
$\Sym\,\U_A$ has all $(\Sym(S), \Sym(\display^S))$ for all $(S,
\display^S) \in \U_A$, and likewise for $\Sym\,\T_A$.

$\Sym$ lifts to a pseudofunctor $\oc$ on $\Thin$ via the functorial
action
\[
\oc \left(
\raisebox{10pt}{$
\xymatrix@R=10pt@C=10pt{
&S	\ar[dl]_{\display^S_{\altidisp A l}}
	\ar[dr]^{\display^S_{\altidisp B r}}\\
A&&B
}$}
\right)
\qquad=\qquad
\raisebox{10pt}{$
\xymatrix@R=10pt@C=0pt{
&\Sym(S)
	\ar[dl]_{\Sym(\display^S_{\altidisp A l})}
	\ar[dr]^{\Sym(\display^S_{\altidisp B r})}\\
\Sym(A)&&\Sym(B)
}$}
\]
on thin spans, defining similarly the image of $2$-cells
as the image by $\Sym$ of their underlying components.

When instantiated on the underlying groupoid of a thin groupoid $A$, the natural
transformations $\eta_A$ and $\mu_A$ are not only functors, but
\textbf{renamings} in the sense of~\cite{DBLP:conf/lics/ClairambaultF23}. Recall
from there the pseudofunctor $\check{-} : \Ren^{\op} \to \Thin$ from the
(dualized) $2$-category of renamings to the bicategory of thin spans, mapping a
renaming $f \co A \to B$ to 
\[
B
\quad \stackrel{f}{\ot} \quad
A
\quad \stackrel{\id_A}{\to} \quad
A
\]
a thin span, yielding a counit $\check{\eta_A} \in \Thin[\oc A, A]$
and a comultiplication $\check{\mu_A} \in \Thin[\oc A, \oc \oc A]$ for
$\oc$. We have (see Appendix \ref{app:sym}):

\begin{theorem}
  \label{thm:sym-pseudocomonad}
  We have a pseudocomonad $\oc$ on $\Thin$ based on $\Sym$.
\end{theorem}

\subsubsection{The exponential}
$\Sym$ enjoys a Seely equivalence in $\Thin$, 
derived from an equivalence already existing in $\Gpd$:
\begin{equation}
  \label{eq:seely-gpd}
  \begin{tikzcd}[cramped]
    \Sym(A+B)
    \ar[rr,shift left,"s_{A, B}"]&&
    \Sym(A) \times \Sym(B)
    \ar[ll,shift left, "\bar{s}_{A, B}"]
  \end{tikzcd}
  \quad\in \Gpd
\end{equation}
for groupoids $A,B$, with $s_{A,B}$ mapping the sequence $\seq{a_1,b_1,b_2,a_2}$
to $(\seq{a_1,a_2},\seq{b_1,b_2})$, and with $\bar s_{A,B}$ mapping
$(\seq{a_1,a_2},\seq{b_1,b_2})$ to $\seq{a_1,a_2,b_1,b_2}$ for instance. When
$A$ and $B$ are thin groupoids, $s_{A,B}$ and $\bar{s}_{A,B}$ are moreover
renamings, so that we can take the image of the above equivalence by 
$\check{-}$ to obtain the Seely equivalence
\[
  \begin{tikzcd}[cramped]
    \oc A \otimes \oc B
    \ar[rr,shift left, "\check s_{A, B}"]
    &&
    \oc(A\with B)
    \ar[ll,shift left,"\check{\bar{s}}_{A, B}"]
  \end{tikzcd}
  \quad\in \Thin
  \zbox.
\]

\subsubsection{The cartesian closed bicategory}

Equipped with the pseudocomonad $\oc$, we derive a Kleisli bicategory
$\Thinb$, whose $1$-morphisms are thus thin spans of the form $\oc A
\ot S \to B$,
composed using the comonadic structure.
By following the proofs in~\cite{DBLP:conf/lics/ClairambaultF23}, which
were mostly non-specific to the $\Fam$ pseudomonad used there, we get:
\begin{theorem}\label{th:cc_bicat}
  $\Thinb$ is a cartesian closed bicategory.
\end{theorem}


\todo{fusionner le sous-tex une fois fini}%

\section{Intersections and Resource Terms}
\label{sec:intersection-resource}

\subsection{Interpreting programs as spans}

\Cref{th:cc_bicat} automatically provides an interpretation of
simply-typed \lterms.
Suppose fixed a countable set $\Var$ of \textbf{variables}. 

The \lterms are defined by the inductive grammar
\[
  M,N,\ldots \qp\syndef x \in \Var \qp\mid M\,N \qp\mid \lambda x. M
  \zbox,
\]
and the \textbf{simple types} are $A, B,\ldots ::= o \mid A \synto B$. A
\textbf{context} is a sequence of bindings $x_1:A_1,\ldots,x_n:A_n$ where the
$x_i$ are (distinct) elements of $\Var$ and the $A_i$ are simple types.
We write
$x \in \Gamma$ when there is a binding $x:B$, for some $B$,
appearing in the sequence of $\Gamma$. 
We consider the standard typing relation $\simplefj \Gamma M A$ for the
simply-typed $\lambda$-calculus.


\subsubsection{Kleisli interpretation}
Given a simple type $A$ we define inductively its
\textbf{interpretation}~$\terminterpb A$, by $\terminterpb{o} = 1$ the
unique thin groupoid based on the terminal (singleton) groupoid, and 
$\terminterpb{A \synto B} = \oc \typeinterpb{A} \linto
\typeinterpb{B}$.
Given a
context $\Gamma = x_1:A_1,\ldots,x_n:A_n$, we define its \textbf{Kleisli
  interpretation} $\ctxtinterpb \Gamma$ as $\typeinterpb {A_1} \with \cdots
\with \typeinterpb {A_n}$. The underlying
groupoid of $\oc\ctxtinterpb \Gamma$ has a monoid structure
in the cartesian category $\Gpd$ giving resource management
operations:
the \eq{multiplication} $\gamma \seqplus \gamma'$ of $\gamma$ and
$\gamma'$ in $\oc G$ is simply their concatenation as sequences; the neutral
element of $\oc G$ is the empty sequence $\seq{}$.

A simply-typed \lterm $\simplefj{\Gamma}{M}{A}$ then admits an
interpretation
%
\begin{gather*}
  \terminterpb{M} \qqp= \tikzcdin[csbo=4em]{\oc \ctxtinterpb \Gamma \&
\terminterpb{M} \ar[l] \ar[r] \& \typeinterpb{A}}
\end{gather*}
in $\Thinb$
via the standard clauses of the interpretation of the
simply-typed $\lambda$-calculus into a cartesian closed category -- we
call this the \textbf{Kleisli interpretation}.
The soundness theorem of cartesian closed categories ensures that
$\beta\eta$-equivalent terms map to positively isomorphic thin
spans; the results of Fiore and Saville
\cite{DBLP:journals/mscs/FioreS21} even yield a coherent interpretation
of reduction sequences as positive isos.

We now set to show that this interpretation is a rigid intersection
type system in disguise; but this will be more visible after we
cope with two aspects of the Kleisli interpretation:
\emph{(1)} elements of $\oc \ctxtinterpb{\Gamma}$ are
sequences \emph{over the whole context}, interleaving accesses to all
variables -- whereas in intersection type systems it is more natural to
have a distinct sequence for each variable; and \emph{(2)} unfolding
the categorical interpretation of $\lambda$-terms in a cartesian closed
category itself constructed as a Kleisli category yields some heavy
bureaucracy, involving compositions with many structural maps, blurring
out the connection with syntax. To mitigate these, we first give a more
syntax-directed characterisation of the interpretation.

\subsubsection{Direct interpretation}
%

We first change the
interpretation of contexts: the interpretation of $\Gamma$ as above is
the thin groupoid $\ctxtinterp \Gamma = \oc \typeinterp {A_1} \otimes
\cdots \otimes \oc \typeinterp {A_n}$ -- for $A$ a type, we 
write $\typeinterp{A}$ as a synonym for $\typeinterpb{A}$. Note
$\ctxtinterp \Gamma$ still has a monoid structure:
the multiplication of 
$\gamma = (\alpha_1,\ldots,\alpha_n)$ and $\gamma' =
(\alpha'_1,\ldots,\alpha'_n)$, two 
elements of $\ctxtinterp{\Gamma}$, is
\[
  \gamma \ctxtplus \gamma'
  \qp=
  (\alpha_1 \seqplus \alpha'_1,\ldots,\alpha_n \seqplus \alpha'_n)
  \mathmakebox[0pt][l]{\qquad
  \in \ctxtinterp \Gamma}
\]
and the neutral element is the $n$-tuple of empty sequences.

Given a typed \lterm $\simplefj{\Gamma}{M}{A}$, 
we now describe its \textbf{direct interpretation} in $\Thinb$ as a span
$\ctxtinterp\Gamma \ot \terminterp M \to \typeinterp A$
given by induction on the typing derivation. In the case of a variable $x_i$ typed in a
context $\Gamma = x_1:A_1,\ldots,x_n:A_n$, we define $\terminterp{x_i}$ as
\[
  \begin{tikzcd}[cramped]
    \oc \typeinterp{A_1} \times \cdots \times \oc \typeinterp{A_n}
    &[5em]
    \typeinterp{A_i}
    \ar[l,"{(\seq{},\ldots,\eta_{\typeinterp{A_i}},\ldots,\seq{})}"']
    \ar[r,"\id_{\typeinterp{A_i}}"]
    &
    \typeinterp{A_i}
    \zbox.
  \end{tikzcd}
\]

For $\simplefj{\Gamma}{M~N}{B}$ where $\simplefj{\Gamma}{M}{A
  \synto B}$ and $\simplefj{\Gamma}{N}{A}$, we set:
\[
  \begin{tikzcd}[csbo=5em,rsbo=3.4em]
    &&
    \terminterp {M~N}
    \ar[dl,dotted,"\pl'"']
    \ar[dr,dotted,"\pr'"]
    \phar[dd,very near start,"\dcorner"]
    &&
    \\
    &
    \terminterp {M} \times \terminterpoc {N}
    \ar[dl,"{(\ctxtplus) \circ (\stdisp[l]{\terminterp M} \times
      \stdisp[l]{\terminterpoc N})}"']
    \ar[dr,"{\stdisp[r]{\terminterp M} \times \stdisp[r]{\terminterpoc N}}"{description,yshift=2pt}]
    &&
    \oc \typeinterp A \times \typeinterp B
    \ar[dl,"{((\pl,\pr),\pl)}"{description,yshift=2pt}]
    \ar[dr,"\pr"]
    \\
    \ctxtinterp{\Gamma}
    &&
    (\oc \typeinterp A \times \typeinterp B) \times \oc \typeinterp A
    &&
    \typeinterp B
  \end{tikzcd}
\]
where we used $\terminterp{M}^{\oc}$, the \textbf{promotion} of
$\terminterp{M}$, defined as the span
\[
  \begin{tikzcd}
    \ctxtinterp\Gamma
    &
    \oc\ctxtinterp\Gamma
    \ar[l,"\ctxtmu_\Gamma"']
    &
    \oc\terminterp{M}
    \ar[l,"\oc \display^{\terminterp{M}}_{\altidisp{\ctxtinterp\Gamma}{l}}"']
    \ar[r,"\oc \display^{\terminterp{M}}_{\altidisp{\typeinterp A}{r}}"]
    &
    \oc \typeinterp{A}
  \end{tikzcd}
\]
where $\ctxtmu_\Gamma : \oc \ctxtinterp{\Gamma} \to \ctxtinterp{\Gamma}$
is the obvious functor sending a sequence of tuples of sequences into
the tuple of concatenated sequences.


Finally, for $\simplefj{\Gamma}{\lambda x.\,M}{A\synto B}$, we set 
$\terminterp{\lambda x.\,M}$ to be the span
\[
  \begin{tikzcd}[cramped]
    \ctxtinterp\Gamma
    &
    \terminterp{M}
    \ar[l,"\display^{\terminterp{M}}_{ll}"']
    \ar[r,"{(\display^{\terminterp{M}}_{lr},\display^{\terminterp{M}}_{\altidisp{\typeinterp B}{r}})}"]
    &[4em]
    \oc \typeinterp{A} \times \typeinterp{B}
  \end{tikzcd}
\]
where $\display^{\terminterp{M}}_{ll}$ and $\display^{\terminterp{M}}_{lr}$ are
obtained from
$\display^{\terminterp{M}}_{\altidisp{\ctxtinterp{\Gamma,x:A}}{l}}$ by
adequately projecting from $\ctxtinterp{\Gamma,x:A} \cong \ctxtinterp\Gamma
\times \oc\typeinterp A$.

We relate the two interpretations: given a context $\Gamma$, we write
\[
  \seely_\Gamma \co \oc (\typeinterp{A_1} + \cdots + \typeinterp{A_n})
\to \oc \typeinterp{A_1} \times \cdots \times   \oc \typeinterp{A_n}
\]
for the evident generalization of the Seely functor
from~\eqref{eq:seely-gpd}. Then:
\begin{theorem}\label{th:direct_carac}
  Given a simply-typed term $\simplefj{\Gamma}{M}{A}$, the span
  \[
    \begin{tikzcd}[cs=5em]
      \oc \typeinterp{A_1} \times \cdots \times \oc\typeinterp{A_n}
      &
      \terminterpb{M}
      \ar[l,"{\seely_\Gamma \circ \display^{\terminterpb{M}}_{\altidisp{\ctxtinterpb \Gamma}{l}}}"']
      \ar[r,"{\display^{\terminterpb{M}}_{\altidisp{\typeinterpb A}{r}}}"]
      &
      \typeinterp{A}
    \end{tikzcd}
  \]
  is thin
  and moreover strongly isomorphic to the span $\terminterp M$.
\end{theorem}

\subsection{Intersection types for spans}
\label{subsec:int_span}

As the direct interpretation is syntax-directed, it is fairly easy to
represent it purely syntactically as an intersection type system.

\subsubsection{Rigid intersection types}
The \textbf{rigid intersection types} are:
\begin{align*}
  \itype{\alpha,\beta,\ldots}
  \qqp\syndef
  &
    \itype{\itsstar}
    \qp\mid
    \itype{\vec{\alpha} \synlinto \beta}
  \\
  \itype{\vec{\alpha},\vec{\beta},\ldots}
  \qqp\syndef
  &
    \itype{\seq{\alpha_1,\ldots,\alpha_n}}
    \quad (n \in \N)
    \zbox.
\end{align*}

As we study the simply-typed $\lambda$-calculus, we shall not
consider these intersection types as standalone objects but only as
refinements of simple types -- we now move to the refinement relation.

\subsubsection{Refinement}
The refinement relation is defined with
\begin{gather*}
  \inferrule{\zbox{}}{\itrefinfj{\itsstar}{o}}
  \qquad
  \inferrule{\itrefinfj{\vec{\alpha}}{A} \qquad \itrefinfj{\beta}{B}}%
  {\itrefinfj{\vec{\alpha} \synlinto \beta}{A \synto B}}
  \qquad
  \inferrule{\forall i \in \set{1,\ldots,n}\qquad \itrefinfj{\alpha_i}{A}}{\itrefinfj{\seq{\alpha_1,\ldots,\alpha_n}}{A}}
  \zbox,
\end{gather*}
noting that both intersection and sequence types may refine simple
types. This refinement judgement correctly captures the objects in the
groupoid interpreting a type $A$, as expressed by the following: 
\begin{proposition}
  \label{prop:corres-typeinterp-it}
  For every simple type $A$, there are bijections
  \[
    K_A \co \Ob(\typeinterp A) \bij \set{\itype{\alpha} \mid \itrefinfj{\alpha}{A}}\,,
\quad
    K^\oc_{A}\co \Ob(\oc \typeinterp A) \bij \set{\itype{\vec{\alpha}} \mid \itrefinfj{\vec{\alpha}}{A}}
    \zbox.
  \]
\end{proposition}


\subsubsection{Resource contexts}

To extend this to contexts, it is convenient to introduce \emph{resource
  contexts}. A \textbf{resource context} for $\Gamma = x_1 : A_1, \dots, x_n :
A_n$ is a sequence of bindings $\Theta = (\itrefinfj[x_1]{\vec{\alpha}_1}{A_1},
\ldots, \itrefinfj[x_n]{\vec{\alpha}_n}{A_n})$ -- we then write
$\itrefinfj{\Theta}{\Gamma}$. Clearly, the bijections above extend to $K_\Gamma
: \Ob(\ctxtinterp{\Gamma}) \bij \{\Theta \mid \itrefinfj{\Theta}{\Gamma}\}$. Given
resource contexts for $\Gamma$
\[
  \Sigma = (\itrefinfj[x_i]{\vec{\alpha}_i}{A_i})_{1 \le i \le n}
  \qqand
  \Theta = (\itrefinfj[x_i]{\vec{\beta}_i}{A_i})_{1 \le i \le n}
  \zbox,
\]
their \textbf{concatenation} $\Sigma\rctxtplus\Theta$ is the resource context
$(\itrefinfj[x_i]{(\vec{\alpha}_i \seqplus \vec{\beta}_i)}{A_i})_i$, where
$\itype{\vec{\alpha}_i \seqplus \vec{\beta}_i}$ is the \textbf{concatenation} of
sequence types.

\subsubsection{Intersection type judgements}
We now introduce typing judgements for rigid intersection types. There
are two kinds of judgements, respectively for single intersection
types and for sequences:
\[
  \itfj{\Theta}{\Gamma}{M}{ \alpha }{A}
  \qand
  \itfj{\Theta}{\Gamma}{M}{ \vec{\alpha} }{A}
  \zbox.
\]

The rules appear in~\Cref{fig:it-rules} ignoring, for the moment, the $\cdots
\rtcolor{u/\rtmterm v} \refin \cdots$ parts in the middle.
In the variable rule, we only display variables with non-empty
sequences. The rules may appear heavy due to the multiple components of
jugdments as required for the simple type refinement. But ignoring
simple type refinements, what remains is the standard ruleset for
non-idempotent intersection types as appears \emph{e.g.} in
\cite{DBLP:journals/mscs/Carvalho18}, just without commutativity.

\begin{figure}
  \centering
  \columnwidth=\linewidth
  \begin{gather*}
    \inferrule{\incomment{(x_1:A_1,\ldots,x_n:A_n)\ctxtwf\quad} 
(\itrefinfj{\alpha}{A_i}) }{
 \rtfj{\ldots, \itrefinfj[x_i]{\seq{\alpha}}{A_i}, \ldots}{x_1:A_1,\ldots,x_n:A_n}{x^\alpha_i}{x_i}{\alpha}{A_i}}
    \\[1em]
    \inferrule{\rtfj{\Theta}{\Gamma}{m}{M}{\vec{\alpha} \synlinto \beta}{A \synto B}\quad%
      \rtfj{\Theta'}{\Gamma}{\rtmterm n}{N}{\vec{\alpha}}{A}}%
    {\rtfj{\Theta \ctxtplus\Theta'}{\Gamma}{m\;\rtmterm{n}}{M\,N}{\beta}{B}}
    \\[1em]
    \qquad
    \inferrule{%
      \text{$\forall i \in \set{1,\ldots,k}$,}\quad%
      \rtfj{\Theta_i}{\Gamma}{m_i}{M}{\alpha_i}{A}%
    }%
    {\rtfj{\Theta_1 \ctxtplus\;\cdots\;\ctxtplus
\Theta_n}{\Gamma}{\rtseq{m_1,\ldots,m_k}}{M}{
\seq{\alpha_1,\ldots,\alpha_k} }{A}}
    \\[1em]
    \inferrule{\rtfj{(\Theta,\itrefinfj[x]{\vec{\alpha}}{A})}{\Gamma,
        x:A}{m}{M}{\beta}{B}}%
    {\rtfj{\Theta}{\Gamma}{\lambda x.\,m}{\lambda x.\,M}{\vec{\alpha} \synlinto \beta}{A\synto B}}
  \end{gather*}
  \caption{Intersection types and approximation}
  \label{fig:it-rules}
\end{figure}

Given a derivation $\simplefj{\Gamma}{M}{A}$, $\gamma \in
\Ob(\ctxtinterp \Gamma)$ and $a \in \Ob(\typeinterp A)$, we write
$\terminterp M_{\gamma,a}$ for the \textbf{witnesses} of $\gamma, a$,
\emph{i.e.} the objects of $\terminterp M$
that project on $\gamma$ and $a$ through $\display^{\terminterp M}_l$
and $\display^{\terminterp M}_r$. As the definition of $\terminterp{M}$
directly follows the syntax, it is relatively direct that:
\begin{proposition}
  \label{prop:charact-terminterp-obj}
  Given a simply-typed $\simplefj{\Gamma}{M}{A}$,
  for every $\gamma \in \Ob(\typeinterp \Gamma)$, for every $a
  \in \Ob(\typeinterp A)$,
  we have a bijection
  \begin{gather*}
    \stsupp{\terminterp M}_{\gamma,a} \bij \set{\pi \mid \text{$\pi$ is a
        derivation of $\itfj{K_\Gamma(\gamma) }{\Gamma }{M }{K_A(a) }{A}$}}
    \zbox.
  \end{gather*}
\end{proposition}

Combined with Theorem \ref{th:direct_carac}, this shows that for any
simply-typed $\lambda$-term $\simplefj{\Gamma}{M}{A}$, for any $\gamma \in
\ctxtinterp{\Gamma}$ and $a \in \typeinterp{A}$, the set of $m \in
\Ob(\terminterpb{M})$ mapping to $\gamma, a$ may be regarded as the set of
derivations of $\itfj{K_\Gamma(\gamma)}{\Gamma}{M}{K_A(a)}{A}$ in our rigid
intersection type system. This result is to be compared with existing works
providing similar characterisations in generalized species of structure
\cite{DBLP:conf/lics/TsukadaAO17,ol:intdist}, where the rigid intersection type
systems considered are much more complex, in particular importing symmetries in
derivations -- and derivations must be quotiented by relations forgetting the
exact position of symmetries in the derivations. In contrast, our derivations
are the simple inductive structures they appear to be, no quotient is required
to obtain our characterisation.

\subsection{Extension to symmetries}

Proposition \ref{prop:charact-terminterp-obj} is analogous to earlier
results of Tsukada \emph{et al.} \cite{DBLP:conf/lics/TsukadaAO17} and
Olimpieri \cite{ol:intdist} set in generalized species of structures,
but here we go further and characterise the full \emph{groupoid} by
also giving an inductive, syntax-directed presentation of the
\emph{symmetries}.

\subsubsection{Intersection type morphisms}

The linear, sequence and multilinear intersection type morphisms are
defined by the grammar
\[
\begin{array}{rcll}
  \itmor{\phi,\psi,\ldots}
  &
  \qqp\syndef
  &
  \itmor{\itsstarid}
  \qp\mid
  \itmor{\widetilde{\phi} \synlinto \psi}
  \\
  \itmor{\vec{\phi}, \vec{\psi}, \ldots}
  &
  \qqp\syndef
  &
   \itmor{\seq{\phi_1, \dots, \phi_n}}
  & (n \in \N)\\
  \itmor{\widetilde{\phi}, \widetilde{\psi}, \ldots}
  &
  \qqp\syndef
  &
    \itmor{(\sigma, \vec{\phi})}
  & (\sigma \in \Sgroup_n, |\itmor{\vec{\phi}}| = n)
\end{array}
\]
where $\Sgroup_n$ is the symmetric group on $n$ elements.
Given two multilinear morphisms
$\itmor{\widetilde{\phi}_1}$ and $\itmor{\widetilde{\phi}_2}$ where
$\itmor{\widetilde{\phi}_i} =
\itmor{(\sigma_i,\seq{\phi_{i,1},\ldots,\phi_{i,n_i}})}$, we define their
\textbf{concatenation} $\itmor{\widetilde{\phi}_1 \seqplus
\widetilde{\phi}_2}$
as $\itmor{(\sigma_1 \permplus \sigma_2,\seq{\phi_{1,i}}_i \seqplus
\seq{\phi_{2,i'}}_{i'})}$. 

\subsubsection{Groupoids of refinements for types}
We extend our refinement relations to morphisms and introduce the
\textbf{linear and multilinear morphism refinement judgements}\simon{est-ce
  qu'on essaye de renommer au max les \eq{morphismes} en \eq{symétries} quand
  c'est possible comme ici?}
\pierre{On peut, j'ai pas d'avis très fort}, of the
form $\itmrefinfj{\phi}{\alpha }{\alpha' }{A}$ and
$\itmrefinfj{\widetilde{\phi} }{\vec{\alpha} }{\vec{\alpha}' }{A}$.
The former states that $\itmor{\phi}$ is a linear morphism from
$\itype{\alpha}$ to $\itype{\alpha'}$ within refinements of simple type
$A$, and likewise for the latter.  Those are defined inductively
through: 
\begin{gather*}
  \inferrule{\hbox{}}{\itmrefinfj{\itsstarid }{\itsstar }{\itsstar }{o}}
  \qquad
  \inferrule{\itmrefinfj{\widetilde{\phi} }{\vec{\alpha} }{\vec{\alpha}'}{A}
    \qquad
    \itmrefinfj{\psi }{\beta }{\beta' }{B}}%
  {(\itmrefinfj{\widetilde{\phi} \synlinto \psi) }{(\vec{\alpha} \synlinto \beta) }{(\vec{\alpha}' \synlinto \beta') }{A \synto B}}
  \\
  \inferrule{n\in \N\qquad \sigma \in \Sgroup_n\qquad \forall i \in
    \set{1,\ldots,n}\quad \itmrefinfj{\phi_i}{\alpha_i}{\alpha'_{\sigma(i)}}{A}}{\itmrefinfj{(\sigma,\seq{\phi_1,\ldots,\phi_n}) }{\seq{\alpha_1,\ldots,\alpha_n} }{\seq{\alpha'_1,\ldots,\alpha'_n} }{A}}
\end{gather*}

It is immediate that if $\itmrefinfj{\phi}{\alpha }{\alpha' }{A}$, then
$\itrefinfj{\alpha}{A}$ and $\itrefinfj{\alpha'}{A}$, and that likewise, if
$\itmrefinfj{\widetilde{\phi} }{\vec{\alpha} }{\vec{\alpha}' }{A}$, then
$\itrefinfj{\vec{\alpha}}{A}$ and $\itrefinfj{\vec{\alpha}'}{A}$.


As suggested by the syntax, the linear (\emph{resp.} multilinear)
intersection types and the associated morphisms that refine a common
simple type $A$ organize into a groupoid $\itgpd A$ (\emph{resp.}
$\itgpdoc{A}$). The composition operation is defined by induction on
derivations, with:
\[
\begin{array}{rcl}
\itmor{\itsstarid \circ \itsstarid}
  &=& \itmor{\itsstarid}\\
\itmor{(\widetilde{\phi}' \synlinto \psi') \circ (\widetilde{\phi}
\synlinto \psi)}
  &=&
\itmor{(\widetilde{\phi}' \circ \widetilde{\phi}) \synlinto (\psi'
\circ \psi)}\\
\itmor{(\sigma',\seq{\phi'_i}_{1 \le i \le n}) \circ
(\sigma,\seq{\phi_i}_{1 \le i \le n}) }
  &=&
\itmor{(\sigma'\circ\sigma,\seq{\phi'_{\sigma(i)} \circ \phi_{i}}_{1
\le i \le n})}
\end{array}
\]

The inverse of a morphism is defined by induction similarly. This allows us to
extend the correspondence of \Cref{prop:corres-typeinterp-it}:
\begin{proposition}
  \label{prop:isom-itgpd-typeinterp}
  For $A$ a simple type, there are groupoid isos:
  \[
    K_A \co \typeinterp A \iso \itgpd A
    \qand
    K_A^\oc \co \oc \typeinterp {A} \iso \itgpdoc {A}
    \zbox.
  \]
\end{proposition}

As $\intr{A}$ is a thin groupoid, it comes equipped with its two
polarized sub-groupoids $\intr{A}_-$ and $\intr{A}_+$ -- via the
proposition above, they transport to two sub-groupoids $\itgpdn{A}$ and
$\itgpdp{A}$ of $\itgpd{A}$.

\subsubsection{Groupoids of refinements for contexts} Consider $\Gamma$
a context and $\itrefinfj{\Theta, \Theta'}{\Gamma}$. A \textbf{context
morphism} from $\itype{\Theta}$ to $\itype{\Theta'}$ is a sequence
\[
  \itmor{\Xi} = (\itmrefinfj[x_1]{\widetilde{\phi}_1}{\vec{\alpha}_1}{\vec{\alpha}'_1}{},
       \ldots,
       \itmrefinfj[x_n]{\widetilde{\phi}_n}{\vec{\alpha}_n}{\vec{\alpha}'_n}{})
\]
where $\itype{\Theta} = (\itrefinfj[x_i]{\vec{\alpha}_i}{A_i})_{1\leq i \leq n}$ and
$\itype{\Theta'} = (\itrefinfj[x_i]{\vec{\alpha}'_i}{A_i})_{1\leq i \leq n}$ -- we also
write $\itmrefinfj{\Xi}{\Theta}{\Theta'}{\Gamma}$ to mean that
$\itmor{\Xi}$ is a morphism of refinements of $\Gamma$ from
$\itype{\Theta}$ to $\itype{\Theta'}$; in that case we write
$\itype{\Theta} = \dom(\itmor{\Xi})$ and $\itype{\Theta'} =
\cod(\itmor{\Xi})$. Given two such morphisms
$\itmrefinfj{\Xi_1}{\Theta_1}{\Theta'_1}{\Gamma}$ and
$\itmrefinfj{\Xi_2}{\Theta_2}{\Theta'_2}{\Gamma}$ for a common context $\Gamma$,
their \textbf{concatenation}
\[
  \itmrefinfj{\Xi_1 \rctxtplus \Xi_2}{\Theta_1 \rctxtplus \Theta_2}{\Theta'_1 \rctxtplus \Theta'_2}{\Gamma}
\]
is defined by componentwise concatenation. The resource contexts and
resource context morphisms form a groupoid $\itgpd\Gamma$ which can be
seen as the product of the $\itgpdoc{A_i}$, so we have a groupoid iso 
\[
  K_\Gamma : \ctxtinterp{\Gamma} \iso \itgpd{\Gamma}\,.
\]

\subsubsection{Morphisms between derivations} We finally set to
construct a \emph{groupoid} of derivations in our rigid intersection
type system. The morphisms will be given by two kinds of judgements,
of the form
\[
  \itmfj{\Xi}{\Gamma}{M}{\phi}{\alpha}{\alpha'}{A}
  \qand
  \itmfj{\Xi}{\Gamma}{M}{\widetilde{\phi}}{\vec{\alpha}}{\vec{\alpha}'}{A}
\]
read as stating that $\itmor{\phi}$ is a morphism from
$\itfj{\dom(\Xi)}{\Gamma}{M}{\alpha}{A}$ to
$\itfj{\cod(\Xi)}{\Gamma}{M}{\alpha'}{A}$, and likewise for multilinear
refinements.

The rules appear in Figure \ref{fig:itm-rules}. The most subtle case is
the last, corresponding to \emph{promotion} and introducing new
symmetries following an arbitrary permutation $\sigma$. In
particular, swapping derivations for $M$ by $\sigma$ requires swapping
accordingly the resource accesses in the context. This uses an
operation that to a family 
$(\itmrefinfj{\widetilde{\phi}_i }{\vec{\alpha}_i }{\vec{\alpha}'_i }{A})_{1\leq i \leq n}$ of morphisms of refinements of $A$ associates
\[
  \itmrefinfj{\rmcplact{\sigma} (\widetilde{\phi}_i)_{1\leq i \leq n}}%
  {\vec{\alpha}_1 \seqplus \ldots \seqplus \vec{\alpha}_n}%
  {\vec{\alpha}'_{\sigma^{-1}(1)} \seqplus \ldots \seqplus \vec{\alpha}'_{\sigma^{-1}(n)}}{A}
\]
a single morphism
defined in the obvious way. This generalizes to context refinement
morphisms transparently, variable by variable. 

\begin{figure*}
  \centering
  \columnwidth=\linewidth
  \begin{gather*}
\scalebox{.95}{$
    \inferrule{\incomment{(x_1:A_1,\ldots,x_n:A_n)\ctxtwf\quad} 
      (\itmrefinfj{\phi}{\alpha}{\alpha'}{A_i})
      }{
        \itmfj{\ldots, \itmrefinfj[x_i]{(\unit{\set{1}},\seq{\phi})}{\seq{\alpha}}{\seq{\alpha'}}{A_i}, \ldots }%
        {}{x_i }{\phi }{\alpha }{\alpha' }{A_i}}$}
    \quad
\scalebox{.95}{$
  \inferrule{\itmfj{\Xi}{}{M}{(\widetilde{\phi} \synlinto \psi)}{(\vec{\alpha} \synlinto \beta)}{(\vec{\alpha}' \synlinto \beta')}{A \synto B}%
      \quad%
      \itmfj{\Xi'}{}{N}{\widetilde{\phi}}{\vec{\alpha}}{\vec{\alpha}'}{A}}%
    {\itmfj{\Xi \rmctxtplus\Xi' }{}{M\,N}{\psi }{\beta }{\beta' }{B}}
$}
    \\[1em]
\scalebox{.95}{$
    \inferrule{\itmfj{\Xi,\itmrefinfj[x]{\widetilde{\phi}}{\vec{\alpha}}{\vec{\alpha}'}{A}}{}{M}{\psi}{\beta}{\beta'}{B}}%
    {\itmfj{\Xi}{}{\lambda x.\,M}{(\widetilde{\phi} \synlinto \psi)}{(\vec{\alpha} \synlinto \beta)}{(\vec{\alpha}' \synlinto \beta')}{A \synto B}}
$}
    \quad
\scalebox{.95}{$
    \inferrule{n\in \N\quad \sigma \in \Sgroup_n\quad \forall i \in
\set{1,\ldots,n},\;\itmfj{\Xi_i}{}{M}{\phi_i}{\alpha_i}{\alpha'_i}{A}}%
    {\itmfj{\plact{\sigma}{(\Xi_i)_{1 \le i\le n}} }{}{M}{(\sigma,\seq{\phi_1,\ldots,\phi_n}) }{\seq{\alpha_1,\ldots,\alpha_n} }{\seq{\alpha'_{\finv\sigma(1)},\ldots,\alpha'_{\finv\sigma(n)}}}{A}}$}
  \end{gather*}
  \caption{The rules for rigid intersection type morphisms}
  \label{fig:itm-rules}
\end{figure*}

Now, given a derivation $\simplefj{\Gamma}{M}{A}$, its associated intersection type
derivations $\itfj{\Theta}{\Gamma}{M}{\alpha}{A}$ and intersection type morphism
derivations $\itmfj{\Xi}{\Gamma}{M}{\phi}{\alpha}{\alpha'}{A}$ organize into a
groupoid $\itgpd M$, whose composition is directly derived from the ones of refinement
types and resource contexts. By considering the two projection functors defined
in the obvious way, we get a span
\[
\begin{tikzcd}
  \itgpd \Gamma
  &
  \itgpd M
  \ar[l,"\display^M_l"']
  \ar[r,"\display^M_r"]
  &
  \itgpd A
\end{tikzcd}
\]
which can be seen as a syntactic description of $\terminterp M$ by the
result: 
\begin{theorem}
  \label{thm:isom-terminterp-itgpd}
  For any simply-typed $\lambda$-term $\simplefj{\Gamma}{M}{A}$, there is
  an iso of groupoids $K_M \co \terminterp M \to \itgpd M$ making the 
  diagram commute:
  \begin{equation*}
    \begin{tikzcd}
      \ctxtinterp\Gamma
      \ar[d,"K_\Gamma"']
      &
      \terminterp M
      \ar[l,"{\stdisp[\altidisp{\ctxtinterp\Gamma}{l}]{\terminterp M}}"']
      \ar[r,"{\stdisp[\altidisp{\typeinterp A}{r}]{\terminterp M}}"]
      \ar[d,"K_M"]
      &
      \typeinterp A
      \ar[d,"K_A"]
      \\
      \itgpd\Gamma
      &
      \itgpd M
      \ar[l,"{\stdisp[l]{M}}"]
      \ar[r,"{\stdisp[r]{M}}"']
      &
      \itgpd A
    \end{tikzcd}
  \end{equation*}
\end{theorem}

By Theorem \ref{th:direct_carac}, this also applies to the
Kleisli interpretation. From this connection to the interpretation in
the cartesian closed bicategory $\Thinb$, we immediately get the
following corollary: 

\begin{corollary}
  \label{coro:interp-red}
  Consider $\simplefj{\Gamma}{M, M'}{A}$ simply-typed $\lambda$-term,
\emph{s.t.} $M \to_{\beta} M'$. Then, there is a weak iso
  of spans $\itgpd M \iso \itgpd{M'}$.
\end{corollary}

This shows that although rigid intersection types do not
enjoy subject reduction as observed in the
introduction, the interpretation in $\Thinb$ associates to every
$\beta$-reduction $M \to_\beta M'$ a bijective \emph{transport}
between derivations of $M$ and $M'$ ``correcting'' the error, up to
some residual symmetries in the groupoids for $\Gamma$ and $A$.

\subsection{Rigid Resource Calculus}

As derivations are somewhat heavy, it seems helpful to remark that they
can be equivalently presented as certain \emph{rigid resource terms}.

%

\subsubsection{Resource terms}

The grammar for \textbf{rigid resource terms} is:
\[
  \begin{array}{lcl}
    \rest{m,n, \ldots} &\qp\syndef& 
    \rest{x^\alpha\;}\mid 
    \rest{\lambda x. m \;}\mid 
    \rest{m\;\rtmterm n}\\
    \rest{\rtmterm m, \rtmterm n \ldots}
                 &\qp\syndef& 
    \rest{\rtseq{m_1, \ldots, m_k}\,,}
  \end{array}
\]
where $\rest{x^\alpha}$ is the data of a variable $x \in \Var$ and of a
\textbf{labelling} intersection type $\itmor{\alpha}$. 
Our resource terms depart from standard resource
terms~\cite{ehrhard2008uniformity} in two significant ways. Firstly, as
in \cite{DBLP:journals/lmcs/OlimpieriA22}
our calculus is \emph{rigid}: argument subterms are sequences rather
than finite multisets. Secondly, we label variable occurrences with
intersection types, so as to guarantee the correspondence with
derivations.


\subsubsection{Approximation relations}
Those resource terms are already implicitely present in our
derivations.
%
%
%
To formalize
that, we introduce the \textbf{linear and multilinear approximation judgements}
\[
  \rtfj{\Theta}{\Gamma}{m }{M }{ \alpha }{A}
  \qand
  \rtfj{\Theta}{\Gamma}{\rtmterm m }{M }{ \itmvec\alpha }{A}
\]
which are defined by the (full) rules of \Cref{fig:it-rules}.
We have a canonical forgetful
function $U$ mapping a derivation $\pi$ of $\rtfj{\Theta}{\Gamma}{m}{M}{
  \alpha }{A}$ to the corresponding derivation $U(\pi)$
of $\itfj{\Theta}{\Gamma}{M}{\alpha}{A}$ and similarly for
multilinear judgements. We easily check that: 
\begin{proposition}
  \label{prop:unicities-for-rterms}
  The following two properties hold:
  \begin{enumerate}[label=(\alph*),ref=(\alph*)]
  \item
    \label{prop:unicities-for-rterms:unique-derivation}
Given a term $\simplefj{\Gamma}{M}{A}$ and resource term $\rest{m}$, there
    is at most one $(\itype{\Theta},\itype{\alpha},\pi)$ with $\pi$ a derivation of
    $\rtfj{\Theta}{\Gamma}{m}{M}{\alpha}{A}$,
  \item
    \label{prop:unicities-for-rterms:unique-lifting}
    For a derivation $\pi$ of $\itfj{\Theta}{\Gamma}{M}{\alpha}{A}$, there is a unique $(u,\tilde\pi)$ \emph{s.t.} $\tilde\pi$ is a
    derivation of $\rtfj{\Theta}{\Gamma}{m}{M}{\alpha}{A}$
    and $U(\tilde\pi) = \pi$.
  \end{enumerate}
\end{proposition}

For a term $\simplefj{\Gamma}{M}{A}$, we write $\Res(M)$ for the
set of resource terms $\rest{m}$ such that $\rtfj{\Theta}{\Gamma}{m}{M}{\alpha }{A}$
is derivable, for some rigid intersection types / contexts $\itype{a},
\itype{\Theta}$. The proposition above gives 
\[
\Res(M) \bij \Ob(\itgpd{M})
\]
a bijection showing that up to isomorphism, $\Thinb$ interprets a
simply-typed $\lambda$-term as a set of rigid resource terms.


\subsubsection{Resource terms and reduction}
This representation lets us examine the action of the interpretation of
reduction steps given by \Cref{coro:interp-red}. Consider a $\beta$-redex
$\vdash (\lambda x.\,M)\,N$. There is an iso
\[
\terminterpb{(\lambda x.\,M)\,N \to_\beta M[N/x]} :
\terminterpb{(\lambda x.\,M)\,N} \iso \terminterpb{M[N/x]}
\]
obtained via the cartesian closed bicategorical structure of
$\Thinb$~\cite{DBLP:journals/mscs/FioreS21}, and through our results 
it yields a bijection
$\Omega : \Res((\lambda x.\,M)\,N) \bij \Res(M[N/x])$
which we can compute. Considering a resource term $\rest{u} =
\rest{(\lambda x.\,m)\,\seq{n_1, \dots, n_k}} \in \Res((\lambda
x.\,M)\,N)$ for $\rtrefinshort{m}{M}$, $\rtrefinshort{\vec{n}}{N}$, we
get 
\begin{eqnarray}
\Omega(\rest{(\lambda x.\,m)\,\seq{n_1, \dots, n_k}}) &=& 
\rest{m[n_1/x_1, \dots, n_k/x_k]}
\label{eq:toplevel_beta}
\end{eqnarray}
where $x_1, \dots, x_k$ are the occurrences of $x$ in $\rest{m}$, \emph{in
order from left to right} -- there must indeed be $k$ occurrences with
the right intersection types, because $\rest{u}$ matches an intersection
type derivation.

But this apparent simplicity for toplevel $\beta$-reductions is
misleading: $\Thinb$ interprets reduction as \emph{weak} span
isos. If we have
\[
  \rtfj{\Theta}{\Gamma}{m}{M}{\alpha}{A}\,,
\]
for $\simplefj{\Gamma}{M}{A}$ with $M \to_\beta M'$, then we do not have
$\itfj{\Theta}{\Gamma}{M'}{\alpha}{A}$ but only
$\itfj{\Theta'}{\Gamma}{M'}{\alpha'}{A}$ for $\Theta' \sym_\Gamma^- \Theta$ and
$\alpha' \sym_A^+ \alpha$; so we cannot directly perform
\eqref{eq:toplevel_beta} deep within $\rest{m}$ as the resulting resource term would
fail to typecheck in our rigid intersection type system.

$\Thinb$ \emph{does} provide some $\rest{m'} = \terminterpb{M \to_\beta
M'}(\rest{m})$, obtained through an interactive reindexing of all components
of $\rest{m}$, correcting the typing mismatches. But its
construction fully exploits the bicategorical structure of $\Thinb$,
and in particular the horizontal composition of $2$-cells (via the
uniqueness property of Lemma \ref{lem:comp_upto_sym}) and it does
not seem to have a simple syntactic presentation.

\subsubsection{Link with multiset resource terms}
\label{sec:res_sym}
To conclude this section, we show how our rigid resource terms do not
have a self-contained rewriting theory; however we show here how they
can be used as representatives for more standard (multiset-based)
resource terms.

We consider \textbf{multiset resource terms} generated by the grammar:
\[
  \begin{array}{lcl}
    \rest{\stdrt u, \stdrt v, \ldots} &\qp\syndef& 
                                                   \rest{x^{\stdit\alpha}}\;\mid \rest{\lambda x. \stdrt u\;}\mid
                                                   \rest{\stdrt u\;\stdmrt v}\\
    \mathrlap{\rest{\stdmrt u, \stdmrt v \ldots}}\hphantom{\itype{\stdmit{\alpha},\stdmit{\beta},\ldots}}
                                      &\qp\syndef& 
                                                   \mathrlap{\rest{\mset{\stdrt u_1, \ldots, \stdrt u_n}\,}}\hphantom{\itype{\mset{\stdit\alpha_1,\ldots,\stdit\alpha_n}}
                                                   \quad (n \in \N)}
  \end{array}
\]
using the (multiset) non-idempotent intersection types defined by
\[
  \begin{array}{lcl}
    \itype{\stdit\alpha,\stdit\beta,\ldots}
    &
      \qp\syndef
    &
      \itype{\itsstar}
      \qp\mid
      \itype{\stdmit{\alpha} \synlinto \stdit\beta}
    \\
    \itype{\stdmit{\alpha},\stdmit{\beta},\ldots}
    &
    \qp\syndef
    &
      \itype{\mset{\stdit\alpha_1,\ldots,\stdit\alpha_n}}
      \quad (n \in \N)
  \end{array}
\]
where, as expected, we use multisets $[\cdots]$ instead of sequences
$\seq{\cdots}$. 
%
Given a rigid intersection type $\itype{\alpha}$, one can obtain a multiset
intersection type $\itype{\itcl{\alpha}}$ by replacing inductively the
sequences $\itype{\seq{\cdots}}$ with multisets $\itype{[\cdots]}$. Similarly,
given a rigid resource term $\rest{m}$, one obtains a multiset resource
term $\rest{\rtcl{m}}$ with the same operation. Then:
\begin{proposition}
  Take $\beta$-normal $\simplefj{\Gamma}{M}{A}$, and $\rest{m},
\rest{n} \in \Res(M)$.

  Then, $\rest{m} \sym \rest{n}$ if and only if $\rest{\rtcl m} =
\rest{\rtcl n}$.
\end{proposition}

This is direct by induction --
here $\rest{m} \sym \rest{n}$ is defined via the correspondence with
derivations. This shows that standard resource terms
fit in the theory of thin spans of groupoids as symmetry classes in the
interpretation of terms, albeit for $\beta$-normal terms. For
non-normal terms this correspondence fails: we have
\begin{eqnarray*}
\rest{(\lambda y.\,x\,y\,y)\,\seq{z, w}}
&\not \sym&
\rest{(\lambda y.\,x\,y\,y)\,\seq{w, z}}
\end{eqnarray*}
though they both map to $\rest{(\lambda y.\,x\,y\,y)\,[w, z]}$ --
in rigid resource terms, $\beta$-redexes explicitly match
variable occurrences and resources in the
argument sequence, while usual resource terms do not.


\todo{fusionner le sous-tex une fois fini}%

\section{Thin Spans and Relational Models}
\label{sec:relational}

Now, we relate thin spans and other extensions of the relational model.
This shall let us re-interpret what these compute in
terms of rigid resource terms and symmetries of rigid intersection
types.

\subsection{The Relational Model}

First of all, we start by describing the relationship between thin
spans of groupoids and the relational model
\cite{DBLP:journals/apal/Girard88}.
It is fairly straightforward, but is hopefully helpful for the
generalizations to come.

\subsubsection{Introducing the relational model.} The relational model
builds on the category $\Rel$ of \emph{sets} and \emph{relations}.
$\Rel$ has a symmetric monoidal structure, obtained by defining the
tensor $A \tensor B = A \times B$ as the cartesian product of sets --
the unit is any singleton set. $\Rel$ is actually compact closed: the
\emph{dual} $A^*$ of a set $A$ is itself, and there are a unit $I \to A
\tensor A^*$ and co-unit $A^* \tensor A \to I$ given by the obvious diagonal
relations. This turns $\Rel$ into a symmetric monoidal closed category,
and as such a model of the linear $\lambda$-calculus -- in particular,
it supports a linear arrow defined as $A \lin B = A \times B$.

But $\Rel$ also has an \emph{exponential modality}, given by
$\oc A = \Mf(A)$
the set of finite multisets of elements of $A$. This 
extends to a comonad $\oc$ on $\Rel$ and for each $A, B$ there is an
isomorphism
$\oc (A\with B) \iso \oc A \tensor \oc B$,
the \emph{Seely isomorphism}. Together with additional coherence
conditions \cite{panorama}, this makes $\Rel$ a \emph{Seely category},
a model of intuitionistic linear logic, and the Kleisli category
$\Rel_{\oc}$ is cartesian closed.

\subsubsection{From $\Thin$ to $\Rel$.} \label{subsubsec:thin_rel}
It seems clear how to relate $\Thin$ and $\Rel$:
on objects, simply send a thin groupoid $A$ to $\syms{A} = A/\sym$ its
\textbf{symmetry classes} (or connected components) -- clearly,
$\syms{\Sym(A)} = \Mf(\syms{A})$. Likewise, given
a thin span $A \ot S \to B$, we can obtain 
\[
\syms{S} = \{(\class{s_A}, \class{s_B}) \mid s \in S\} \in \Rel[\syms{A},
\syms{B}]
\]
called its \textbf{relational collapse}, for $\class{(-)}$ the
equivalence class. Then:

\begin{restatable}{proposition}{functrel}
This yields a functor $\syms{-} : \Thin \to \Rel$.
\end{restatable}
\begin{proof}
This requires us to compose witnesses \emph{up to symmetry}, which we
do thanks to Lemma \ref{lem:comp_upto_sym} -- see Appendix
\ref{subsec:funct_rel}.
\end{proof}

\subsubsection{Preservation of further
structure.}\label{subsubsec:pres_fur}
From the definition,
it is straightforward that we have bijection yielding isos in $\Rel$:
\[
\begin{array}{lcccc}
t^\tensor_{A, B} &:& \syms{A} \tensor
\syms{B} &\iso &\syms{A\tensor B}\\
t^\with_{A, B} &:& \syms{A} \with \syms{B} &\iso& \syms{A \with B}\\
t^\oc_A &:& \oc \syms{A} &\iso & \syms{\oc A}
\end{array}
\]
for $A$ and $B$ thin groupoids; in particular the third amounts to
$\syms{\oc A} \bij \Mf(\syms{A})$ for $A$ any thin groupoid.  It is a
routine verification that these components satisfy the coherence
conditions required to make $\syms{-} : \Thin \to \Rel$ a Seely functor
(see Appendix \ref{app:seely_lifting}), so that:

\begin{theorem}
Setting, for any $\oc A \ot S \to B$ in $\Thin_{\oc}[A, B]$,
\[
\syms{S}_\oc = \syms{S} \circ t^\oc_A \in
\Rel_\oc[\syms{A}, \syms{B}]\,,
\]
yields $\syms{-}_{\oc} : \Thin_{\oc} \to \Rel_{\oc}$ a cartesian closed
functor.
\end{theorem}

It follows that this preserves the interpretation of the
simply-typed $\lambda$-calculus: for every
simple type $A$ there is a bijection $t_A : \intr{A}_{\Rel_\oc} \bij
\syms{\typeinterpb{A}}$ -- and likewise for contexts -- so that if
$\Gamma \vdash M : A$, $\gamma \in
\intr{\Gamma}_{\Rel_\oc}$, $a \in \intr{A}_{\Rel_\oc}$, $(\gamma, a)
\in \intr{M}_{\Rel_\oc}$ iff $(t_\Gamma\incomment{cette notation c'est quoi exactement?}\,\gamma, t_A\,a) \in
\syms{\terminterpb{M}}_\oc$.

\subsection{Weighted Relations}

The \emph{weighted relational model} is due to Larmarche
\cite{DBLP:journals/tcs/Lamarche92}, though its application to
semantics was fleshed out by Laird \emph{et al.}
\cite{DBLP:conf/lics/LairdMMP13}. In full generality, its construction
is parametrized by a complete semiring; but for the purposes of this
paper we will only work with the semiring $\Ni = \mathbb{N}
\cup \{+\infty\}$ of completed natural numbers.

\subsubsection{The weighted relational model} Rather than merely
collecting the completed executions, the weighted relational assigns a
\emph{weight} -- here, an element of $\Ni$ -- to any execution. In
other words, a weighted relation from set $A$ to set $B$ is a function
$A \times B \to \Ni$.

This lets us \emph{count} properties of execution: for instance, it is
shown in \cite{DBLP:conf/lics/LairdMMP13} how the relational model
weighted by $\Ni$ counts how many distinct executions may lead to a
given result at ground type, for a non-deterministic extension of
$\PCF$. But even for purely deterministic programs (in fact,
simply-typed $\lambda$-terms), the weighted relational model computes
non-trivial coefficients.

\begin{example}
Considering the simply-typed $\lambda$-term
\[
f : o \to o \to o, x : o, y : o \vdash f\,(f\,y\,x)\,(f\,x\,y) : o\,,
\]
then the point of the web written in intersection type notation as
\[
f : \itype{[[\itsstar] \lin [] \lin \itsstar, [] \lin [\itsstar] \lin
\itsstar]}, x : \itype{[\itsstar]}, y: \itype{[]} \vdash
\itype{\itsstar}
\]
has a weight of $2$ in the weighted relational model -- this reflects
the fact that this point can be realized in two distinct ways,
depending on which occurrence of $f$ calls which argument;
seemingly corresponding to two distinct normal resource terms:
\[
\rest{f\,[f\,[]\,[x^\itsstar]]\,[]}
\qquad
\rest{f\,[]\,[f\,[x^\itsstar]\,[]]\,,}
\]
or (via Section \ref{sec:res_sym}) to two symmetry classes of rigid terms.
\end{example}

This suggests that, maybe, the weighted relational model counts the
number of resource terms inhabiting a certain intersection type.
But that is not actually
the case, as illustrated by this next example.

\begin{example}\label{ex:non_res}
Considering now the simply-typed $\lambda$-term
\[
f : o \to o, g : o \to o, y:o \vdash f\,(g\,y) : o\,,
\]
then the point of the web written in intersection type notation as
\[
f : \itype{[[\itsstar,\itsstar] \lin \itsstar]}, g : \itype{[[] \lin
\itsstar, [\itsstar] \lin \itsstar]}, y :
\itype{[\itsstar]} \vdash \itype{\itsstar}
\]
is \emph{also} assigned a weight of $2$ by the weighted relational
model, even though the reader can check that there is only one resource
term inhabiting that type. Clearly here we are somehow accounting for
the \emph{symmetries} of this resource term -- but which symmetries?
\end{example}

\subsubsection{Categorical structure} The weighted relational model is
structured around the category $\WRel$: its objects are sets, and a
morphism from $A$ to $B$ is $\alpha \in \Ni^{A\times B}$ -- for
$a \in A$ and $b \in B$, we write $\alpha_{a, b} \in \Ni$ for
$\alpha(a, b)$. Identity is $(\id_A)_{a, a'} = \delta_{a, a'}$.
Composition is
\[
(\beta \circ \alpha)_{a, c} = \sum_{b\in B} \alpha_{a, b} \cdot
\beta_{b, c}
\]
for $\alpha \in \WRel[A, B]$, $\beta \in \WRel[B, C]$, $a\in A$ and
$c\in C$. This potentially infinite sum always ``converges'' because
our set of weights $\Ni$ includes the infinity. Just like $\Rel$,
$\WRel$ is a compact closed category with biproducts, see 
\cite{DBLP:conf/lics/LairdMMP13} for details.
%

Finally, there is an exponential modality $\oc A = \Mf(A)$ on
sets. On morphisms, the critical definition is that of
\emph{functorial promotion}:
\[
(\oc \alpha)_{\mu, [b_1, \dots, b_n]} = 
\sum_{\substack{(a_1, \dots, a_n)\\
\emph{s.t.}~\mu = [a_1, \dots, a_n]}}
\prod_{i=1}^n \alpha_{a_i, b_i}\,.
\]

Altogether, just like $\Rel$, $\WRel$ is a Seely category, and thus the
associated Kleisli category $\WRel_{\oc}$ is cartesian closed.

\subsubsection{Positive witnesses.} We must make the functor of
Section \ref{subsubsec:thin_rel} quantitative -- from a thin span $A \ot
S \to B$ and symmetry classes $\ca \in \syms{A}, \cb \in \syms{B}$, we
must assign a number $\syms{S}_{\ca, \cb} \in \Ni$. We naturally expect
this number to be the cardinal of a set of \emph{witnesses}
\[
\syms{S}_{\ca, \cb} = \#\,\wit_S(\ca, \cb)\,,
\]
thus our question boils down to the following: what is the adequate
notion of witnesses, in a thin span, for symmetry classes $\ca, \cb$?
It is tempting to count symmetry classes in $S$, however we have
seen in Section \ref{sec:res_sym} that (for normal terms) those correspond to
resource terms, and Example \ref{ex:non_res} shows that it is not
what the weighted relational model counts; in fact we shall see it
accounts for 
\begin{eqnarray}
f\,\tuple{\lambda x.\,g\,\tuple{y}, \lambda x\,g\,\tuple{}}\,,
\qquad
f\,\tuple{\lambda x.\,g\,\tuple{},\lambda x.\,g\,\tuple{y}}\,,
\label{eq:exnonres}
\end{eqnarray}
the \emph{two} rigid resource terms that intuitively inhabit the
intersection type of Example \ref{ex:non_res} -- even though the two
are symmetric. 
But it is not the case that we are simply counting rigid resource
terms! If we were to replace $y$ with $x$ in Example \ref{ex:non_res},
then the weight given by $\WRel$ becomes one and thus the two rigid
resource terms displayed in \eqref{eq:exnonres} with $x$ instead of $y$
should suddenly just account for one...  

$\Thin$ will help sort this out. Assume that all groupoids interpreting
types come equipped with a function $\rep{(-)}$ associating to each
symmetry class $\ca \in \syms{A}$ a representative $\rep{\ca} \in \ca$.
Then we set
\begin{eqnarray}
\pwit_S(\ca, \cb) &=& \{s \in S \mid \rep{\ca} \sym_A^- s_A \,\&\,s_B
\sym_B^+ \rep{\cb}\}
\label{eq:defpwit}
\end{eqnarray}
where $a \sym_A^+ a'$ means there is $\theta^+ \in A_+[a,
a']$ and likewise for $\sym_A^-$; we call those the \textbf{positive
witnesses} of $\ca$ and $\cb$ in $S$. This depends on a
choice of representatives for symmetry classes -- our development
will apply for thin groupoids equipped with representatives:

\begin{definition}\label{def:repr_tg}
A \textbf{representation} for a thin groupoid $A$ is a function
$\rep{(-)} : (\ca \in \syms{A}) \to \ca$
such that for all $\ca \in \class{A}$, $\rep{\ca}$ is
\textbf{canonical}, in the sense that for all $\theta \in A[\rep{\ca},
\rep{\ca}]$, the unique factorization $\theta = \theta^- \circ
\theta^+$ given by Lemma \ref{lem:factor} satisfies $\theta^- \in
A_-[\rep{\ca}, \rep{\ca}]$ and $\theta^+ \in A_+[\rep{\ca},
\rep{\ca}]$.
\end{definition}

If $A$ is a thin groupoid with a representation and $\ca \in A$, we
write $\m(\ca) = \# A(\rep{\ca}, \rep{\ca})$ the \textbf{symmetry
degree} of $\ca$. Likewise, we write $\m_+(\ca) = \#
A_+(\rep{\ca},\rep{\ca})$ (resp. $\# A_-(\rep{\ca}, \rep{\ca})$) the
\textbf{positive symmetry degree} (resp. negative) of $\ca$. From
Definition \ref{def:repr_tg}, we then have
\begin{eqnarray}
\m(\ca) &=& \m_+(\ca) \cdot \m_-(\ca)
\label{eq:symmetry_degrees}
\end{eqnarray}
reflecting quantitatively the factorization of Lemma \ref{lem:factor}.

One can build a representation for all constructions on thin groupoids
so far. The non-trivial case is the exponential: if we have canonical
$a_1, \dots, a_n \in A$, then so is $\tuple{a_1, \dots, a_n} \in \oc A$,
provided that whenever $a_i \sym_A a_j$ then $a_i = a_j$. Thus given
$\ca = [\ca_1, \dots, \ca_n] \in \syms{\oc A}$ we first consider
$[\rep{\ca_1}, \dots, \rep{\ca_n}]$, which we present in a sequential
ordering, following some total order on objects of $A$ that we assume
globally fixed in advance.
From now on, we consider all thin groupoids equipped with a
canonical representation.

Summing up, to any thin span $A \ot S \to B$ we associate 
$\syms{S}_{\ca, \cb} = \# \pwit_S(\ca, \cb)$, and we now aim to prove
that this extends to a functor. 

\subsubsection{Functoriality} \label{subsubsec:funct}
Preservation of the identity is obvious by the factorization property
of Lemma \ref{lem:factor}.  Composition is more subtle. Naturally, for
$A \ot S \to B$ and $B \ot T \to C$ we expect a bijection
\begin{eqnarray}
\pwit_{T\odot S}(\ca, \cc) &\bij& \sum_{\cb \in \syms{B}} \pwit_S(\ca,
\cb) \times \pwit_T(\cb, \cc)\,,\label{eq:exp_bij}
\end{eqnarray}
and while our results imply that such a bijection exists for
cardinality reasons, it is not actually what we shall build directly.
In fact, there appears to be no natural function from the right-hand
side to the left-hand side. We must assemble $s \in \pwit_S(\ca, \cb)$
and $t \in \pwit_T(\cb, \cc)$ into an element of $\pwit_{T\odot S}(\ca,
\cc)$ but we cannot do that directly, as we only have $s_B \sym_B t_B$
and not $s_B = t_B$. We can, as in the proof of Proposition
\ref{prop:compositor}, compose $s$ and $t$ via any symmetry $\theta_B :
s_B \sym_B t_B$ to obtain an element of $\pwit_{T\odot S}(\ca, \cc)$;
but this does not yield a function as the result depends on the choice
of $\theta_B$.

To address this dependency in the undetermined mediating symmetry, we
consider instead the composition of witnesses carrying explicit
symmetries: the \textbf{$\sim$-witnesses} from $\ca$ to $\cb$ are
triples
\[
\spwit_{S}(\ca, \cb) = 
\{(\theta_A^-, s, \theta_B^+) \mid \theta_A : \rep{\ca} \sym_A^-
s_A~\&~s_B \sym_B^+ \rep{\cb}\}\,;
\]
so $(\theta^-_A, s, \theta^+_B) \in \spwit_{S}(\ca, \cb)$ and
$(\vartheta^-_B, t, \vartheta^+_C) \in \spwit_T(\cb, \cc)$ providing 
$\vartheta_B^- \circ \theta_B^+$ used to compose $s$ and $t$
via Lemma \ref{lem:comp_upto_sym}.

While in a thin span $A \ot S \to B$ the display $S \to A
\times B$ is not a fibration, $\sim$-witnesses do enjoy a
fibration-like property:

\begin{proposition}\label{prop:act_strat}
Consider $A \ot S \to B$ a thin span, $s \in S$, and
\[
\theta_A^- : a \sym_A^- s_A\,
\qquad
\theta_B^+ : s_B \sym_B^+ b\,.
\]

For $\Omega_A : a' \sym_A a$ and $\Omega_B : b \sym_B b'$,
there are unique $\varphi^S : s \sym_S s'$ and
$\vartheta_A^- : a' \sym_A^- s'_A$, $\vartheta_B^+ : s'_B \sym_B^+ b'$
s.t. the diagrams commute: 
\[
\xymatrix@R=15pt@C=15pt{
a     \ar[r]^{\theta_A^-}
        \ar@{<-}[d]_{\Omega_A}&
s_A
        \ar[d]^{\varphi^S_A}\\
a'     \ar[r]_{\vartheta_A^-}&
s'_A
}
\qquad
\xymatrix@R=15pt@C=15pt{
s_B
        \ar[r]^{\theta_B^+}
        \ar[d]_{\varphi^S_B}&
b     \ar[d]^{\Omega_B}\\
s'_B
        \ar[r]_{\vartheta_B^+}&
b'
}
\]
\end{proposition}

This follows from Lemma \ref{lem:act_sym}. 
We can now establish the bijection patching
\eqref{eq:exp_bij}.
Consider $A \ot S \to B$ and $B \ot T \to C$, $\ca
\in \syms{A}, \cb \in \syms{B}$ and $\cc \in \syms{C}$, we write
$\spwit_{S,T}(\ca, \cb, \cc)$ for the 
\textbf{$\sim$-interaction witnesses}, \emph{i.e.} tuples
$(\theta_A^-, s, \Theta, t, \theta_C^+)$ where $\theta_A^- : \rep{\ca}
\sym_A^- s_A, s_B = t_B = b$ and $\theta_C^+ : t_C \sym_C^+ \rep{\cc}$
so that $(s, t) \in T\odot S$; and $\Theta : \rep{\cb} \sym_B b$.

\begin{restatable}{proposition}{witbij}
For $S, T, \ca, \cb, \cc$ as above, there is a bijection
\[
\Upsilon
~:~
\spwit_S(\ca, \cb) \times \spwit_T(\cb, \cc) 
~\bij~
\spwit_{S,T}(\ca, \cb, \cc)
\]
\emph{s.t.} for any $\Upsilon((\theta_A^-, s, \theta_B^+), (\Omega_B^-,
t, \Omega_B^+)) = (\psi_A^-, s', \Theta, t', \psi_C^+)$, there are
unique $\omega^S : s \sym_S s'$ and $\nu^T : t \sym_T t'$ making the
diagrams commute:
\[
\xymatrix@R=0pt@C=15pt{
&s_A
        \ar@{<-}[dl]_{\theta_A^-}
        \ar[dd]^{\omega_A^S}&
s_B
        \ar[r]^{\theta_B^+}
        \ar[dd]_{\omega^S_B}&
\rep{\cb}
        \ar[dd]^{\Theta}&
t_B
        \ar[dd]^{\nu^T_B}
        \ar@{<-}[l]_{\Omega_B^-}&
t_C
        \ar[dr]^{\Omega_C^+}
        \ar[dd]_{\nu^T_C}\\
        \rep{\ca}\ar[dr]_{\psi_A^-}&&&&&&\rep{\cc}\\
&s'_A
        &
s'_B   \ar@{=}[r]&
b&
t'_B   \ar@{=}[l]&
t'_C   \ar[ur]_{\psi_C^+}
}
\]
\end{restatable}

This is direct from Lemma \ref{lem:comp_upto_sym} and Proposition
\ref{prop:act_strat}, see App. \ref{app:bijquant}.

We now have a bijection that somewhat looks like \eqref{eq:exp_bij},
but we must sum over all symmetry classes in $B$ and check that the
cardinality of added symmetries cancels out. Indeed it is easy that
\begin{eqnarray*}
\#\,\spwit_S(\ca, \cb) &=& \m_-(\ca) \cdot \#\,\pwit_S(\ca, \cb) \cdot
\m_+(\cb)\,;
\end{eqnarray*}
from the definition, and since $\sim$-interaction witnesses carry a
symmetry class in $B$ and an endo-symmetry, it is also direct that 
\[
\#\,\spwit_{T\odot S}(\ca, \cc) = 
\sum_{\cb\in \syms{B}} \frac{1}{\m(\cb)} \cdot \#\,\spwit_{S, T}(\ca,
\cb, \cc)\,.
\]

From there and \eqref{eq:symmetry_degrees}, \eqref{eq:exp_bij} follows
from a simple computation. So:
\begin{corollary}
This yields a functor $\syms{-} : \Thin \to \WRel$.
\end{corollary}

\subsubsection{Exponential.} The crucial point remaining is that the
functorial action of ${\oc}$ is preserved. For this section, we adopt
notations inlining the bijections of Section \ref{subsubsec:pres_fur}: 
in particular, we write elements of $\syms{\oc A}$ as finite multisets
of elements of $\syms{A}$. We must give
\begin{eqnarray}
\pwit_{\oc S}(\bmu, [\cb_1, \dots, \cb_n]) 
&\bij& \sum_{\substack{\tuple{\ca_1, \dots, \ca_n}\\\emph{s.t.} [\ca_1,
\dots, \ca_n] = \bmu}} \prod_{i=1}^n \pwit_S(\ca_i, \cb_i)  
\end{eqnarray}
a bijection, for any thin span $A \ot S \to B$. 

From left to right, recall that writing $\bnu = [\cb_1,
\dots, \cb_n]$, $\pwit_{\oc S}(\bmu, \bnu)$ comprises those $\vec{s}$
such that $\rep{\bmu} \sym_A^- \vec{s}_{\oc A}$ and $\vec{s}_{\oc B}
\sym_B^+ \rep{\bnu}$. Let us write $\rep{\bnu} = \tuple{b_1, \dots,
b_n}$. On the right-hand side, as positive symmetries cannot exchange
elements of a sequence, we have $\vec{s} = \tuple{s^1, \dots, s^n}$
where $s^i_B \sym_B^+ b_i$. However on the left-hand side symmetries
\emph{can} exchange elements, so that there must exist an
(unspecified) permutation $\sigma \in \varsigma(n)$ such that
$\rep{\ca}_{\sigma(i)} \sym_A^- s^i_A$, informing
$\tuple{\ca_{\sigma(1)}, \dots, \ca_{\sigma(n)}}$ satisfying
$[\ca_{\sigma(1)}, \dots, \ca_{\sigma(n)}] = \bmu$ as needed.
Reciprocally, it is clear that data on the right-hand side can be
assembled into an element of $\pwit_{\oc S}(\bmu, \bnu)$ and that those
operations are inverse of one another.

This shows that modulo the bijection $t^{\oc}_A$ of Section
\ref{subsubsec:pres_fur}, the functorial action of $\oc$ is preserved.
The other bijections of Section \ref{subsubsec:pres_fur} still yield
isomorphisms in $\WRel$ -- for which, by a slight abuse, we keep the
same notation. All necessary coherence conditions are satisfied, so
that this operation lifts to the Kleisli (bi)categories.

\begin{theorem}
We have $\syms{-}_\oc :  \Thin_{\oc} \to \WRel_{\oc}$
cartesian closed.
\end{theorem}


\subsubsection{Consequences} Since a cartesian closed functor preserves
the interpretation of the simply-typed $\lambda$-calculus, this gives
us a combinatorial description of the coefficients computed by
$\WRel_{\oc}$:

\begin{corollary}\label{cor:count_wit}
Consider $\Gamma \vdash M : A$ a simply-typed $\lambda$-term.

For every $\bgamma \in \intr{\Gamma}_{\WRel_\oc}$ and $\ca \in
\intr{A}_{\WRel_\oc}$, we have
\[
(\intr{M}_{\WRel_{\oc}})_{\bgamma, a} =
\#\pwit_{\terminterpb{M}}(t_\Gamma\,\bgamma\incomment{c'est quoi cette notation?}, t_A\,\ca)\,.
\]
\end{corollary}

By the results in Section \ref{subsec:int_span}, this is also the number of
derivations
$\itfj{\Theta}{\Gamma}{M}{\alpha}{A}$
(or their representations as rigid resource terms) where $\itype{\Theta}$ is
negatively symmetric (\emph{resp.} $\itype{\alpha}$ is positively
symmetric) to the intersection type matching a chosen canonical rigid
representative for $\bgamma$ (\emph{resp.} for $\ca$). 
Note that we can also derive:
\begin{proposition}\label{prop:weight_symclass}
Consider $\Gamma \vdash M : A$ a simply-typed $\lambda$-term.

For every $\bgamma \in \intr{\Gamma}_{\WRel_\oc}$ and $\ca \in
\intr{A}_{\WRel_\oc}$, we have
\[
(\intr{M}_{\WRel_{\oc}})_{\bgamma, \ca} =
\sum_{\bt \in W}
\frac{\m_+(t_\Gamma\,\bgamma)\cdot \m_-(t_A\,\ca)}{\m(\bt)}
\]
where $W$ is the set of symmetry classes in $\terminterpb{M}$
mapping to $(t_\Gamma\,\bgamma, t_A\,\ca)$, and $\m(\bt)$ is the size of the
group of symmetries on $\bt$.
\end{proposition}

This is because to each symmetry class $\bt$ correspond a number of
positive witnesses equal to the negative symmetries of the matching
rigid intersection type, divided by the symmetries of $\bt$ -- the
proof appears in Appendix \ref{app:aux_thin}. Thus, one can obtain the
right coefficient from symmetry classes (and therefore for normal
standard resource terms following Section \ref{sec:res_sym}), but the
weight of each symmetry class must be corrected suitably accounting for
symmetries.

\subsection{Distributors and Generalized Species} 

We now establish a link between thin spans and the bicategory
of distributors (\emph{i.e.} profunctors). We keep this section
succinct; to a large extent, it is a simplification of the construction
in \cite{DBLP:conf/lics/ClairambaultOP23}.

\subsubsection{The bicategory of groupoids and distributors}
A \textbf{distributor} from groupoid $A$ to $B$
(\emph{a.k.a.} \emph{profunctor}) is a functor
$\alpha : A^{\op} \times B \to \Set$
giving, for all $a \in A, b \in B$, a set $\alpha(a, b)$ of
\emph{witnesses}, along with an action of symmetries: if
$x \in \alpha(a, b)$ and $\theta \in B(b,
b')$, we write $\theta \cdot x$ for the functorial action $\alpha(\id,
\theta)(x) \in \alpha(a, b')$. Similarly, if $\vartheta \in A(a', a)$,
we write $x\cdot \vartheta \in \alpha(a', b)$ for $\alpha(\vartheta,
\id)$. 

The bicategory $\Dist$ has groupoids as objects, distributors
as morphisms, and natural transformations as 2-cells. The
\textbf{identity distributor} on $A$ is the hom-set functor
$\id_A = A[-, -] : A^\op \times A \to \Set$.
The \textbf{composition} of two distributors
$\alpha : A^\op \times B \to \Set$ and $\beta : B^\op \times C \to
\Set$ is defined in terms of the coend formula:
\[
(\beta \bullet \alpha)(a, c) = \int^{b \in B} \alpha(a, b) \times
\beta(b,
c)\,,
\]
meaning that concretely, $(\beta \bullet \alpha)(a, c)$ consists in
pairs $(x, y)$, where $x \in \alpha(a, b)$ and $y \in \beta(b, c)$ for
some $b \in B$, quotiented by $(g \cdot x, y) \sim (x, y \cdot g)$ for
$x \in \alpha(a, b)$, $g \in B(b, b')$ and $y \in \beta(b', c)$. 
The bicategory $\Dist$ has 
%
cartesian products given by the disjoint union $A + B$.

\subsubsection{Extracting distributors from thin spans} 
On objects, we send a thin groupoid $(A, A_-, A_+, \U_A, \T_A)$ to its
underlying groupoid $A$. 

On morphisms, given a thin span $A \ot S \to B$, for all $a\in A$ and
$b \in B$ we must specify a set $\D{S}(a, b)$. It is tempting to set
simply the pre-image $(\display^S)^{-1}(a, b)$, but there is no
functorial action
\[
\D{S}(\theta_A, \theta_B) : \D{S}(a, b) \to \D{S}(a', b')
\]
for $\theta_A \in A(a', a)$ and $\theta_B \in B(b, b')$ as
$\display^S$ is not a fibration. We need
a finer symmetry lifting property of thin spans -- and we
have one, seen in Proposition \ref{prop:act_strat}.
Thus, we set instead $\D{S}(a, b)$ as the set $\spwit_S(a, b)$ of
\textbf{$\sim$-witnesses} of $(a, b)$ in $S$, \emph{i.e.} triples
$(\theta_A^-, s, \theta_B^+)$ \emph{s.t.} $s \in S$, $\theta_A^- \in
A_-(a, s_A)$ and $\theta_B^+ \in B_+(s_B, b)$. Though we keep the same
terminology and notation as in Section \ref{subsubsec:funct}, those are
$\sim$-witnesses of \emph{specific} objects of the groupoids $A$ and
$B$, not symmetry classes.

We get a functorial action by setting $\D{S}(\Omega_A,
\Omega_B)(\theta_A^-, s, \theta_B^+)$ as the positive witness
$(\vartheta_A^-, s', \vartheta_B^+)$ as in the statement of Proposition
\ref{prop:act_strat}, yielding a distributor for every thin span $A \ot
S \to B$: 

\begin{proposition}
\label{prop:strategy-to-distributor}
We have a distributor $\D{S} : A^{\op} \times B \to \Set$.
\end{proposition}

\subsubsection{Constructing natural transformations.} Consider $S, T$
thin spans from $A$ to $B$, and $(F, F^A,F^B) : S \to T$ a positive
morphism; consisting for each $s \in S$ of $F^A_s \in A_-(s_A, (F
t)_A)$ and $F^B_s \in B_+(s_B, (F s)_B)$.

To each $\w = (\theta_A^-, s, \theta_B^+) \in \D{S}(a, b)$, we
set $\D{S}(F, F^A,F^B)(\w)$ to
\[
(a \stackrel{\theta_A^-}{\to} s_A \stackrel{F^A_s}{\to} (F t)_A,
\qquad
F t,
\qquad
(F t)_B \stackrel{F^B_s}{\to} s_B \stackrel{\theta_B^+}{\to} b)
\]
which by the uniqueness property of Proposition \ref{prop:act_strat}
can be easily verified to give a natural transformation from $\D{S}$ to
$\D{T}$.

\subsubsection{Further components.} To complete the pseudofunctor, we
need two natural isomorphisms, the \emph{unitor} and the
\emph{compositor}.

\begin{proposition}
  Given a thin span $A$, there is a natural iso
\[
\pid^A : \D{\Id_A} \stackrel{\iso}{\Rightarrow} A[-, -] : 
A^{\op} \times A \to \Set\,.
\]
\end{proposition}

This is straightforward from the factorization result of Lemma
\ref{lem:factor}. Now, we focus on the preservation of composition. For
two thin spans $A \ot S \to B$ and $B \ot T \to C$,  we have the
\textbf{compositor}: 

\begin{proposition}\label{prop:compositor}
There is a natural isomorphism:
\[
\pcomp^{S, T} : \D{T \odot S} \Rightarrow
\D{T} \bullet \D{S} :
A^{\op} \times B \to \Set\,.
\]
\end{proposition}
\begin{proof}
The map $\pcomp^{S, T}_{a, c}$ sends $(\theta_A^-, (s, t), \theta_C^+)
\in \D{T \odot S}(a, c)$ (with $s_B = t_B = b$) to (the equivalence
class of) the pair 
\[
((\theta_A^-, s, \id_{b}), (\id_{b}, t, \theta_C^+))
\in (\D{T} \bullet \D{S})(a, c)\,.
\]

For each $a \in A$ and $c \in C$, this forms a bijection. Consider
indeed
\[
\w^S = (\theta_A^-, s, \theta_B^+)  
        \in \D{S}(a, b)
\qquad
\w^T = (\theta_B^-, t, \theta_C^+) 
        \in \D{T}(b, c)
\]
composable witnesses. By Lemma \ref{lem:comp_upto_sym} we compose
$s$ and $t$ through $\theta_B^- \circ \theta_B^+$, yielding
unique $\varphi^S \in S[s, s'], \varphi^T \in
T[t, t'], \vartheta_A^-, \vartheta_C^+$ s.t.: 
%
\[
\begin{tikzcd}[row sep=0.1em, column sep=1.5em]
  &  s_A \ar[dd, "\varphi^S_A"] & s_B \ar[dd,
"\varphi^S_B"] \ar[r, "\theta_B^+"] & b \ar[r, "\theta_B^-"] &
t_B \ar[dd, "\varphi^T_B"] & t_C \ar[dd,
"\varphi^T_C"'] \ar[dr, "\theta_C^+"] \\
  a \ar[ur, "\theta_A^-"] \ar[dr, "\vartheta_A^-"'] & & & & & & c
\\
  & s'_A & s'_B \ar[r, Rightarrow, no head]& b' \ar[r,
  Rightarrow, no head] & t'_B & t'_C \ar[ur,
  "\vartheta_C^+"'] &
\end{tikzcd}
\]
which, writing $\Theta_B = \varphi^S_B \circ {\theta_B^+}^{-1} =
\varphi^T_B \circ \theta_B^-$, entails
\[  
\begin{array}{rcrcl}
\v^{S} &=& (\vartheta_A^-, s', \id_{y_B}) &=& \Theta_B \cdot
(\theta_A^-, s, \theta_B^+)\\
\v^T &=& (\id_{y_B}, t', \vartheta_C^+) &=& (\theta_B^-, t,
\theta_C^+) \cdot \Theta_B
\end{array}
\]
so $(\v^S, \v^T) = (\Theta_B \cdot \w^S, \v^T)
\sim (\w^S, \v^T \cdot \Theta_B) = (\w^S, \w^T)$.
Now $(\v^S, \v^T) = \pcomp^{S, T}(\vartheta_A^-,
t'\odot s', \vartheta_C^+)$, showing surjectivity -- injectivity also
follows from the uniqueness clause in Lemma \ref{lem:comp_upto_sym}.
\end{proof}

The naturality and coherence requirements hold, and altogether:

\begin{theorem}\label{th:oplax_thin}
  This yields a peudofunctor $\D{-} : \Thin \to \Dist$.
\end{theorem}

\subsubsection{Lifting to Kleisli bicategories.} Recall that 
$\Esp$ is the Kleisli bicategory $\Dist_{\Sym}$.
%
Composition of $F : \Sym(A)^{\op} \times B \to \Set$ and $G :
\Sym(B)^{\op} \times C \to \Set$ is $G \bullet F^{\Sym}$, where the
\textbf{promotion} is 
\[
F^{\Sym}(\vec{a}, \tuple{b_1,\dots, b_n}) = 
\int^{\vec{a}'_1, \dots, \vec{a}'_n}
A[\vec{a}, \vec{a}'_1\dots\vec{a}'_n] \times \Pi_{i=1}^n F(\vec{a}'_i, b_i)
\]
comprising a morphism in $A[\vec{a}, \vec{a}'_1, \dots, \vec{a}'_n]$
along with a family in $\Pi_{i = 1}^n F(\vec{a}'_i, b_i)$, quotiented
by an equivalence relation.

Likewise, the promotion $S^{\Sym}$ of a thin span, constructed as
\[
\Sym(A) \ot \Sym(\Sym(A)) \ot \Sym(S) \to \Sym(B)\,,
\]
yields by $\D{-}$ the distributor associating to
$\vec{a}, \tuple{b_1, \dots, b_n}$ triples
\begin{eqnarray}
(\theta^-_{\Sym(A)}, \tuple{s_1, \dots, s_n}, \theta^+_{\Sym(B)}) 
\in \D{S^{\Sym}}(\vec{a}, \vec{b})\,,
\label{eq:triple}
\end{eqnarray}
but $\theta^+_{\Sym(B)}$ is positive, so cannot
reindex the $b_i$s and must be $(\id_{1\dots n}, (\theta^+_i)_{1 \leq i
\leq n})$ for $\theta^+_i$ is positive in $B$. 
So we map \eqref{eq:triple} to
\[
(\theta^-_{\Sym(A)}, \tuple{(\id, s_i, \theta^+_i)\mid 1 \leq i \leq
n}) \in \D{S}^{\Sym}(\vec{a}, \vec{b})
\]
inducing a natural bijection $\D{S^\Sym}_{\vec{a},
\vec{b}} \bij \D{S}^{\Sym}_{\vec{a}, \vec{b}}$.

Combined with $\pcomp^{S, T}$ this provides a natural
iso for preservation of Kleisli composition. Together with a
straightforward natural isomorphism for Kleisli identity laws
and lengthy verifications for coherence, we obtain
a pseudofunctor $\D{-} : \Thin_{\oc} \to \Esp$.

\subsubsection{A cartesian closed pseudofunctor.} We check that this
extends to a \emph{cc-pseudofunctor} \cite{DBLP:journals/mscs/FioreS21}. 
First, $\D{-}$ preserves constructions on objects strictly.
The notion of a \emph{fp-pseudofunctor}
\cite{DBLP:journals/mscs/FioreS21} requires that for each
$(A_i)_{1\leq i \leq n}$, $\tuple{\D{\pi_1}, \dots, \D{\pi_n}}$ is part
of an adjoint equivalence
\[
\xymatrix@R=-5pt{
\prod_{i=1}^n A_i\incomment{problème ici}
	\ar@/^1pc/|{\tuple{\D{\pi_1}, \dots, \D{\pi_n}}}[rr]&\bot&
\prod_{i=1}^n A_i
	\ar@/^1pc/|{\q^\times_{\vec{A}}}[ll]
}
\]
in $\Esp$: here $\q^\times$ can be taken to be the identity in
$\Esp$, completed to an adjoint equivalence in the obvious way.
On top of that, the definition of a \emph{cc-pseudofunctor}
\cite{DBLP:journals/mscs/FioreS21} then additionally requires
that $\e_{A, B} = \Lambda(\D{\evm_{A, B}} \bullet_\Sym \q^\times) :
A\tto B \to A\tto B$ is also part of
\[
\xymatrix@R=-5pt{
A\tto B
        \ar@/^1pc/|{\e_{A, B}}[rr]&\bot&
A\tto B
        \ar@/^1pc/|{\q^\tto_{A, B}}[ll]
}
\]
an adjoint equivalence. But $\e_{A, B}$ can be computed to be naturally
isomorphic to the identity on $A\tto B$ in $\Esp$; constructing the
adjoint equivalence is then straightforward. Altogether: 

\begin{theorem}\label{th:final_esp}
$\D{-} : \Thin_{\oc} \to \Esp$ is a cc-pseudofunctor.
\end{theorem}

\subsubsection{Consequences.}
Fix a simply-typed $\lambda$-term $\Gamma \vdash M : A$.

By Theorem \ref{th:final_esp}, we have a natural isomorphism
$I : \intr{M}_{\Esp} \iso \D{\terminterpb{M}}$ showing that up to iso,
generalized species of structure compute positive witnesses in the sense of thin
spans of groupoids.

By the results of Section \ref{sec:intersection-resource}, this
can be reformulated as:
\begin{corollary}
  For $\gamma \in \ctxtinterpb \Gamma$ and $a \in \typeinterpb A$, we have a
  bijection
\[
  \intr{M}_{\Esp}(\gamma, a) \iso 
  \left\{
    (\itmor{\theta_\Gamma^-}, \bt, \itmor{\theta_A^+})
    \quad\middle|\quad
    \begin{array}{l}
      \itmor{\theta_\Gamma^-} \in \itgpdn{\Gamma}[\itype{K_\Gamma
(s_\Gamma\,\gamma)}, \itype{\Theta}],\\
      \bt \in \itgpd{M}_{\itype{\Theta}, \itype{\alpha}},\\
      \itmor{\theta_A^+} \in \itgpdp{A}[\itype{\alpha}, \itype{K_A\,a}]
    \end{array}
  \right\}\,.
\]
\end{corollary}

This captures the interpretation of simply-typed $\lambda$-terms in
$\Esp$ syntactically.
This is analogous to results by Tsukada \emph{et
al.} \cite{DBLP:conf/lics/TsukadaAO17} and Olimpieri \cite{ol:intdist},
except our derivations are simpler, without quotient.

Finally, altogether, the isomorphism $I$ and Corollary
\ref{cor:count_wit} entail:
\begin{corollary}
For any $\bgamma \in \intr{\Gamma}_{\WRel_\oc}$ and $\ca \in
\intr{A}_{\WRel_\oc}$,
\[
(\intr{M}_{\WRel_\oc})_{\bgamma, \ca} = 
\frac{\# \intr{M}_{\Esp}(t_\Gamma\,\bgamma,
t_A\,\ca)}{\m_-(t_\Gamma\,\bgamma)\cdot \m_+(t_A\,\ca)}
\]
where $\# \intr{M}_{\Esp}(t_\Gamma\,\bgamma, t_A\,\ca)$ is defined for any
representative.
\end{corollary}

This is independent of $\Thinb$, though it does require the positive
and negative symmetries -- this shows that these are fundamental in
quantitative semantics, independently of their role in $\Thin$.

\section{Conclusion}

We have illustrated our results on the simply-typed
$\lambda$-calculus for the economy of presentation and since it already
features the phenomena of interest, but $\Thin$ readily supports
non-determinism and can be easily extended with quantitative
(probabilistic and quantum) primitives, for which we expect our results
still hold.

Our results show that the interpretation of the simply-typed
$\lambda$-calculus in $\Thin$ can be regarded as a rigid Taylor
expansion. Section \ref{sec:res_sym} then suggests a link with the
standard Taylor expansion of $\lambda$-terms which may illuminate the
coefficients appearing there; however we could not find an exposition
of the simply-typed Taylor expansion in the literature, so we had to
omit this by lack of space.  Detailing that, and the untyped case, is
left for future work.

\bibliographystyle{ACM-Reference-Format}
\bibliography{main}

\todo{transformer le corps du texte pour prendre en compte ifappendix}%
\ifappendix
  \clearpage
  \appendix
  
\onecolumn
\newgeometry{twoside=true,
  includeheadfoot, head=13pt, foot=2pc,
  paperwidth=6.75in, paperheight=10in,
  top=58pt, bottom=44pt, inner=4cm, outer=4cm,
  marginparwidth=2pc,heightrounded}%
\headwidth=\linewidth%

\begin{center}
  \LARGE \bf Appendix
\end{center}
\simon{c'est ok si on passe en une colonne pour l'annexe?}%
\todo{si une colonne, placer le onecolumn à la place du clearpage avant appendix}%

\section{Additional properties in $\Thin$}
\label{app:aux_thin}

Here, we provide the proofs of properties of thin spans of groupoids
that this paper need, which were not provided in
\cite{DBLP:conf/lics/ClairambaultF23}.

\subsection{Reindexing by a symmetry}
\label{app:reindexing}

Here, we show the detailed proof of Lemma \ref{lem:act_sym}, which
expresses how $\Thin$ lets us reindex witnesses by symmetries.

\makeatletter
\@ifundefined{actsym}{}{\actsym*}%
\makeatother
\begin{proof}
We show this ignoring the left-hand side, for $S \in \T_B$, $s \in S$,
$\theta_B : s_B \sym_B b$; the general case follows by applying this to
$A^\perp \parr B$.

\emph{Existence.} By hypothesis, we know that $S \in \T_B \subseteq
\U_B$. By definition of thin groupoids, we know that $(B_+, \id_B^+)
\in \T_B^\pperp \subseteq \U_B^\perp$, so that $S \perp (B_+,
\id_B^+)$. Hence, the pullback
\[
  \begin{tikzcd}[cramped,rsbo=1.7em,csbo=2.5em]
    &
    \cdot
    \ar[rd,"\pr",dashed] 
    \ar[ld,"\pl"',dashed] 
    \phar[dd,"\dcorner",very near start]
    &
    \\
    S
    \ar[rd,"\display^S"']
    &&
    B_+
    \ar[ld,"\id_B^+"]
    \\
    &
    B
  \end{tikzcd}
\]
is a bipullback. By our concrete characterisation of bipullbacks in
$\Gpd$, applying this to $s \in S$, $b \in B_+$ and $\theta_B : s_B
\sym_B b$, this gives us $\varphi : s \sym_S s'$ and $\vartheta_B^+ :
s'_B \sym_B^+ b$ such that $\theta_B = \vartheta_B^+ \circ \varphi_B$
as required.

\emph{Uniqueness.} Consider another solution, comprising $\psi : s
\sym_S s''$ and $\nu_B^+ : s''_B \sym_B^+ b$ such that $\nu_B^+ \circ
\psi_B = \theta_B$. Then, $\nu_B^+ \circ \psi_B = \vartheta_B^+ \circ
\varphi_B$, so
\[
(\psi \circ \varphi^{-1})_B = \psi_B \circ \varphi^{-1}_B =
(\vartheta_B^+)^{-1} \circ \nu_B^+
\]
a positive morphism. But by
\cite[Lem.~3]{DBLP:conf/lics/ClairambaultF23}, a morphism $\psi \circ
\varphi^{-1}$ in $S$ which maps to a positive morphism in $B$ must be
an identity; hence $s' = s''$ and $\psi \circ \varphi^{-1} = \id_{s'}$,
so that $\varphi = \psi$. Additionally, $\vartheta_B^+ = \theta_B \circ
(\varphi_B)^{-1} = \theta_B \circ (\psi_B)^{-1} = \nu_B^+$ as desired,
concluding the proof.
\end{proof}

\subsection{Counting symmetry classes}

Our aim here is to provide a characterisation of the number of positive
witnesses inhabiting a given symmetry class; providing the missing
brick for the proof of Proposition \ref{prop:weight_symclass}. For this
section, let us fix a thin groupoid $A$ and some $S \in \T_A$; we shall
derive the two-sided version of the result by simply applying it to
$A^\perp \parr B$.


First, we show that any symmetry class in $S$ has a representative that
is positively symmetric to (the chosen representative) of the
corresponding symmetry class in $A$:

\begin{lemma}
Consider $\cs \in S/\sym_S$, and consider $\ca$ its display.

Then, there is $s \in \cs$ such that $s_A \sym_A^+ {\rep{\ca}}$.
\end{lemma}
\begin{proof}
Consider first any $s \in \cs$. By hypothesis, there is $\theta_A : s_A
\sym_A \rep{\ca}$. It might not be positive, but by Lemma
\ref{lem:factor} (applied to $\theta_A^{-1}$) it factors uniquely as
$s_A \stackrel{\theta_A^-}{\sym_A^-} a \stackrel{\theta_A^+}{\sym_A^+}
\rep{\ca}$
and now, by Proposition \ref{prop:act_strat}, there are unique $\varphi
: s \sym_S s'$ and $\vartheta_A^+ : s'_A \sym_A^+ a$ such that
$\theta_A^- = \vartheta_A^+ \circ \varphi_A$. But then $s' \in \cs$ and
$s'_A \sym_A^+ a \sym_A^+ \rep{\ca}$.
\end{proof}

So, for each $\cs \in S/\sym_S$, we choose a representative $\rep{\cs}
\in \cs$ such that $(\rep{\cs})_A \sym_A^+ \rep{\ca}$; and we also
choose a ``reference'' positive symmetry $\theta_\cs^+ : (\rep{\cs})_A
\sym_A^+ \rep{\ca}$. Finally, for every $s \in \cs$ we choose some 
$\kappa_s : \rep{\cs} \sym_S s$.

Our aim is, for a fixed $\ca \in \syms{A}$ and for every symmetry class
$\cs$ such that $\cs_A = \ca$, to count the number of concrete positive
witnesses in $\cs$. We introduce some notations for this set -- let us
write 
\begin{eqnarray*}
\wit^+_S[\cs] &=& \{s \in \cs \mid s_A \sym_A^+ \rep{\ca}\}\\
\spwit_S[\cs] &=& 
\{(s, \theta_A^+) \mid s \in \cs ~\&~ \theta_A^+ : s_A \sym_A^+
\rep{\ca}\}
\end{eqnarray*}
for the concrete witnesses (resp. $\sim^+$-witnesses) within a
symmetry class $\cs$ for $\ca$.

Then, we prove the following bijection, for $\mathcal{S}(\cs) =
S[\rep{\cs}, \rep{\cs}]$ and $\mathcal{S}(\ca) = A[\rep{\ca},
\rep{\ca}]$.

\begin{proposition}\label{prop:bijwitc}
There is a bijection
$\spwit_\sigma[\cs] \times \mathcal{S}(\cs) \bij
\mathcal{S}(\ca)$.
\end{proposition}
\begin{proof}
First we show that for every $(s, \vartheta^+_A) \in
\spwit_S(\ca)$ and $\varphi \in \mathcal{S}(\cs)$, there is a unique
$\psi_A \in \mathcal{S}(\ca)$ such that the following diagram
commutes:
\[
\xymatrix@R=15pt@C=15pt{
(\rep{\cs})_A
        \ar[r]^{\theta_{\cs}^+}
        \ar[d]_{(\kappa_{s} \circ \varphi)_A}&
\rep{\ca}
        \ar[d]^{\psi_A}\\
s_A
        \ar[r]_{\vartheta^+_A}&
\rep{\ca}
}
\]
but this is obvious, as $\psi_A$ is determined by composition from the
other components.

Reciprocally, we show that for all $\psi_A \in \mathcal{S}(\ca)$, there
are unique $(s, \vartheta^+_A) \in \spwit_S(\ca)$ and $\varphi \in
\mathcal{S}(\cs)$ such that the same diagram above commutes. First, by
canonicity of $\rep{\ca}$, $\psi_A$ factors as $\psi_A = \psi^+_A \circ
\psi^-_A$ for $\psi^-_A \in \mathcal{S}(\ca)$ negative and $\psi^+_A
\in \mathcal{S}(\ca)$ positive. By Proposition \ref{prop:act_strat},
there are unique $(s', \omega^+_A) \in \spwit_S(\ca)$ and $\phi :
\rep{\cs} \sym_S s'$ such that the following diagram commutes:   
\[
\xymatrix@R=15pt@C=15pt{
(\rep{\cs})_A
        \ar[r]^{\theta^+_{s}}
        \ar[d]_{\phi_A}&
\rep{\ca}
        \ar[d]^{\psi^-_A}\\
s'_A
        \ar[r]_{\omega^+_A}&
\rep{\ca}
}
\]

We may then define
$s := s'$,
$\vartheta^+_A := \psi^+_A \circ \omega^+_A$, and
$\varphi := (\kappa_{s'})^{-1} \circ \phi$
and the diagram is obviously satisfied. It remains to prove uniqueness,
so assume we have $(t, \nu_A^+) \in \spwit_\sigma(\ca)$ and
$\xi \in \mathcal{S}(\cs)$ such that the following diagram commutes:
\[
\xymatrix@R=15pt@C=15pt{
(\rep{\cs})_A
        \ar[r]^{\theta^+_{s}}
        \ar[d]_{(\kappa_{t} \circ \xi)_A}&
\rep{\ca}
        \ar[d]^{\psi_A}\\
t_A
        \ar[r]_{\nu_A^+}&
\rep{\ca}
}
\]

But then $(\kappa_{t} \circ \xi)\circ (\kappa_{s} \circ \varphi)^{-1}$
is a symmetry in $S$ displaying to a positive symmetry in $A$, so
must be an identity by \cite[Lem.~3]{DBLP:conf/lics/ClairambaultF23}.
Thus $s = t$, $\xi = \varphi$, and so also $\vartheta^+_A = \nu_A^+$ as
it is uniquely determined from the other components by the diagram.
This gives constructions in both directions, and that they are inverses
follows directly from the uniqueness properties. 
\end{proof}

From that bijection, we may conclude the following result:

\begin{theorem}\label{th:count_symclass}
Consider $A$ a thin groupoid, $S \in \T_A$ and $\ca \in \syms{A}$,
$\cs$ displaying to $\ca$. Then,
\[
\sharp \wit_S^+[\cs] = \frac{\sharp \mathcal{S}_-(\ca)}{\sharp
\mathcal{S}(\cs)}
\]
where $\mathcal{S}_-(\ca)$ is the group of negative symmetries on
$\rep{\ca}$.
\end{theorem}
\begin{proof}
By Proposition \ref{prop:bijwitc}, we have
$\sharp \spwit_S[\cs] \times \sharp \mathcal{S}(\cs) = \sharp
\mathcal{S}(\ca)$,  so we have
\[
\sharp \mathcal{S}_+(\ca) \times \sharp \wit^+_S[\cs] \times
\sharp \mathcal{S}(\cs) = \sharp \mathcal{S}_+(\ca) \times \sharp
\mathcal{S}_-(\ca)
\]
via the easy fact that $\sharp \spwit_S[\cs] = \sharp
\mathcal{S}_+(\ca) \times \sharp \pwit_S[\cs]$ and canonicity of
$\rep{\ca}$. The identity follows. 
\end{proof}

And now, we can finally deduce:

\begin{corollary}
Consider $A, B$ thin groupoids, $A \ot S \to B$ a thin span, $\ca \in
\syms{A}$ and $\cb \in \syms{B}$. Then,
\[
\syms{S}_{\ca, \cb} = \sum_{\cs \in W} \frac{\m_+(\ca) \cdot
\m_-(\cb)}{\m(\cs)}
\]
for $W$ the set of symmetry classes in $S$ mapping to $\ca, \cb$.
\end{corollary}
\begin{proof}
We calculate:
\begin{eqnarray*}
\syms{S}_{\ca, \cb} 
&=& \sharp \pwit_S(\ca, \cb)\\
&=& \sum_{\cs \in W} \sharp \pwit_S[\cs]\\
&=& \sum_{\cs \in W} \frac{\m_+(\ca) \cdot \m_-(\cb)}{\m(\cs)}
\end{eqnarray*}
using the definition, then partitioning the positive witnesses by
symmetry class, and by Theorem \ref{th:count_symclass}.
\end{proof}

Proposition \ref{prop:weight_symclass} immediately follows from
Corollary \ref{cor:count_wit} in combination with this.

\section{The $\Sym$ pseudocomonad}
\label{app:sym}
\newcommand\defpmfletter{{\mathbf{H}}}%
\newcommand\otherpmfletter{{\mathbf{K}}}%

Here, we give additional details about the definition of the $\Sym$
pseudocomonad on $\Thin$, derived from the $\Sym$ monad on $\Gpd$. Notably, we
reuse the general results presented in
\cite[App.~G]{DBLP:conf/lics/ClairambaultF23}, which were developed to show that
the $\Fam$ pseudomonad on $\Gpd$ lifted to a pseudocomonad, and that we recall
below.

\subsection{General definitions and results}

A functor $F \co \Gpd \to \Gpd$ is called \emph{bicartesian} when it preserves
pullbacks and sends pullbacks that are bipullbacks to bipullbacks. A
\emph{\pmfunctor} is a tuple $(\defpmfletter,\defpmfletter^+,\iota)$ with
$\defpmfletter,\defpmfletter^+$ being functors $\Gpd \to \Gpd$ where
$\defpmfletter$ and $\defpmfletter^+$ are bicartesian and preserve functors
(between groupoids) that are bijective on objects (of the groupoids), and such
that $\defpmfletter^+$ preserves discrete groupoids, and $\iota \co
\defpmfletter^+ \To \defpmfletter$ being a natural transformation which is
pointwise monomorphic (that is, such that each $\iota_X$ is a monomorphism) and
surjective on objects of the groupoids, satisfying moreover that it is
\emph{bicartesian}, meaning that its naturality squares are both pullbacks and
bipullbacks.

Given two \pmfunctors $\defpmfletter = (\defpmfletter,\defpmfletter^+,\iota)$
and $\otherpmfletter = (\otherpmfletter,\otherpmfletter^+,\kappa)$, a
\emph{\pmtransformation} between $\defpmfletter$ and $\otherpmfletter$ is a pair
$(\alpha,\alpha^+)$ of natural transformations where $\alpha \co \defpmfletter
\To \otherpmfletter$ and $\alpha^+ \co \defpmfletter^+ \To \otherpmfletter^+$
are such that $\alpha$ is bicartesian and
\begin{equation}
  \label{eq:pmtransformation-eq}
  \kappa \circ \alpha^+ = \alpha \circ \iota \zbox.
\end{equation}
\begin{lemma}
  \label{lem:pmtransformation-one-pb-bpb}
  Given a \pmtransformation $(\alpha,\alpha^+)$, $\alpha^+$ is bicartesian.
\end{lemma}
\begin{proof}
  The bicartesianness of $\alpha^+$ can be deduced using standard properties of
  rectangles of pullbacks and their adaptation to rectangles of pullbacks that
  are bipullbacks~\cite[Lemma~5]{DBLP:conf/lics/ClairambaultF23}.
\end{proof}

Now, a \emph{\pmmodification} between two such \pmtransformations
$(\alpha,\alpha^+)$ and $(\beta,\beta^+)$ is the data of a modification $m \co
\alpha \TO \beta$ in the $3$\category of $2$\categories.

\begin{definition}
  We write $\pmFunct$ for the $3$\category with one object, \pmfunctors as
  $1$\morphisms, \pmtransformations as $2$\morphisms, and \pmmodifications as
  $3$\morphisms.
\end{definition}

Given a \pmfunctor $(\defpmfletter,\defpmfletter^+,\iota)$ and a thin groupoid
$\cA$, there is a canonical thin groupoid $\defpmfletter \cA$ whose class of
uniform strategies is $\U_{\defpmfletter\cA} = \set{\defpmfletter S \mid S \in
  \U_\cA}^{\perp\perp}$, whose class of thin prestrategies is
$\T_{\defpmfletter\cA} = \set{\defpmfletter S \mid S \in
  \T_\cA}^{\pperp\pperp}$, and whose negative and positive sub-groupoids are
$(\defpmfletter \cA)_- = \defpmfletter A_-$ and $(\defpmfletter\cA)_+ =
\defpmfletter^+A_+$ with embeddings given by the compositions
\[
  \defpmfletter A_-
  \xto{\defpmfletter(\id^-_A)}
  \defpmfletter A
  \qand
  \defpmfletter^+A_+
  \xto{\defpmfletter^+(\id^+_A)}
  \defpmfletter^+A
  \xto{\iota_A}
  \defpmfletter A
  \zbox.
\]
By the conditions of \pmfunctors, they can be shown to be elements of
$\T_{\defpmfletter\cA}$ and $\T_{\defpmfletter\cA}^{\pperp}$ as required.

The mapping $\cA \mapsto \defpmfletter\cA$ can be extended to a pseudofunctor
$\check \defpmfletter \co \Thin \to \Thin$ by mapping a thin span
$\tikzcdin[csbo=4em]{A \& S \ar[l,"\display^S_{\altidisp A l}"{description}]
  \ar[r,"\display^S_{\altidisp B r}"{description}] \& B}$ to the thin span
$\tikzcdin[csbo=6em]{\defpmfletter A \& \defpmfletter S
  \ar[l,"\defpmfletter(\display^S_{\altidisp A l})"{description}]
  \ar[r,"\defpmfletter(\display^S_{\altidisp B r})"{description}] \& \defpmfletter B}$, and by
mapping weak morphisms to their image by $\defpmfletter$.

Similarly, given a \pmtransformation $\alpha = (\alpha,\alpha^+)$ between two
\pmfunctors $\defpmfletter$ and $\otherpmfletter$, one can define a
pseudonatural transformation $\check\alpha$ between $\check{\otherpmfletter}$ and
$\check{\defpmfletter}$ by putting
\[
  \check\alpha_A
  \qqp=
  \tikzcdin[csbo=6em]{\otherpmfletter A \& \defpmfletter A
    \ar[l,"\alpha_A"{description}]
    \ar[r,"\id_{\defpmfletter A}"{description}] \& \defpmfletter A}
\]
and given a \pmmodification between two \pmtransformations $\alpha$ and $\beta$,
one can define a modification $\check m$ between $\check\alpha$ and
$\check\beta$ the expected way. By checking all the details, we get that
\begin{proposition}[{\cite[Proposition~20]{DBLP:conf/lics/ClairambaultF23}}]
  \label{prop:pmfunct-bifunctor}
  Considering $\pmFunct$ as a strict $2$\category by forgetting the dimension
  $0$, $\check{(-)}$ induces a pseudofunctor
  \[
    \check{(-)} : \pmFunct^\cocat \to \Bicat(\Thin,\Thin)
  \]
  between bicategories.
\end{proposition}
Now, one can define the notion of monad (or even pseudomonad) in $\pmFunct$ as
particular instances of the general notion of monad (or pseudomonad) expressed
in $\pmFunct$ seen as an abstract $3$\category.
While~\cite{DBLP:conf/lics/ClairambaultF23} considered $\Fam$ which was a
pseudomonad in $\Gpd$, this work is concerned with $\Sym$, which is a monad on
$\Gpd$, so that we only require results for the monadic case. In this regard,
we have
\begin{proposition}
  \label{prop:check-preserves-comonads}
  The functor $\check{(-)}$ send a monad on $\pmFunct$ to a pseudocomonad on $\Thin$.
\end{proposition}
\begin{proof}
  By a direct adaptation of the proof of
  \cite[Theorem~3]{DBLP:conf/lics/ClairambaultF23}, which shows that $\Fam$,
  seen as a pseudomonad on $\pmFunct$, is sent to a pseudocomonad on $\Thin$.
  Indeed, the proof is not specific to $\Fam$ and can be specialized to the case
  of a monad on $\pmFunct$.
\end{proof}

\subsection{Permutations and the multiplication of $\Sym$}

The full definition of the multiplication $\mu$ of the monad $\Sym$ on groupoids
relies on operations on permutations that we introduce below. In the following,
given $n\in \N$, we write $\nset n$ for the set $\set{1,\ldots,n}$.

The category $\permcat$ of \emph{permutations} is defined as the category whose
objects are the natural numbers $n \in \N$, and whose morphisms from $n$ to $n$
are the elements of the symmetric group $\Sgroup_n$, that is, the bijections
from $\nset{n}$ to itself, and with no morphisms from $m$ to $n$ for $m \neq n$;
composition of morphisms is given by the composition of the underlying
functions, and the identity morphism on $n$ is the identity function on
$\nset{n}$. We will often write $\permcat_n$ for $\permcat(n,n)$.

The category $\permcat$ is equipped with a tensor product $\permplus$ defined by
putting $m \permplus n = m + n$ for $m,n \in \N$ and, for $\lambda \in \permcat_m$
and $\rho \in \permcat_n$, by defining $\lambda \permplus \rho$ as the bijection
$\nu \co \nset{m+n}\to \nset{m+n}$ such that $\nu(i) = \lambda(i)$ for $i \in
\nset{m}$, and $\nu(m+i) = m+\rho(i)$ for $i \in \nset{n}$. More generally,
given $k \in \N$, $\vec n = (n_1,\ldots,n_k) \in \N^k$ and bijections $\nu_i \co
n_i \to n_i$, we write $\permplus_{1\le i \le k} \nu_i$ for the bijection
$(\cdots(\nu_1 \permplus \nu_2) \permplus \cdots ) \permplus \nu_k$. The natural
number $0$ is the unit object for this tensor product, making $\permplus$ a
strict monoidal category. It is even a strict symmetric monoidal category: for
$m,n \in \N$, one defines a bijection $\sigma_{m,n} \co \nset{m+n} \to
\nset{n+m}$ by putting $\sigma_{m,n}(i) = n+i$ for $i \in \nset{m}$ and
$\sigma_{m,n}(m+i) = i$ for $i \in \nset{n}$, and one readily verifies that it
gives an adequate symmetry for the monoidal structure. More generally, given $k
\in \N$, $\vec n = (n_1,\ldots,n_k) \in \N^k$ and $\rho \in \Sgroup_k$, one can
define a $k$-ary symmetry $\sigma_{\vec n,\rho} \co \nset{n_1+\cdots+n_k} \to
\nset{n_{\rho(1)} + \cdots + n_{\rho(k)}}$ by putting
\[
  \sigma_{\vec n,\rho}(n_1+ \cdots + n_{l-1} + i) = n_{\finv \rho(1)} + \cdots + n_{\finv \rho(\rho(l)-1)}+ i
\]
for $l \in \nset{k}$ and $i \in \nset{n_l}$. Given permutations $\tau_l \in
\permcat_{n_l}$ for every $l \in \nset{k}$, we will often write
$\plact{\rho}{(\tau_l)_{l\in \nset k}}$ for the composite $\sigma_{\vec n,\rho}
\circ (\permplus_{l\in\nset k}\tau_l)$. One easily verifies the following property:
\begin{lemma}
  \label{lem:plact-functorial}
  Given $l \in \N$, $(m_i)_{i \in \nset l} \in \N^l$ and $(n_{i,j})_{i\in\nset
    l,j\in \nset{m_i}}$, permutations $\rho_{i,j} \in \permcat_{n_{i,j}}$ for
  $i\in\nset l,j\in \nset{m_i}$ and $\sigma_i \in \Sgroup_{m_i}$ for $i \in
  \nset l$ and $\tau \in\Sgroup_l$, we have
  \[
    \plact{\rho}{(\plact{\sigma_i}{(\tau_{i,j})_j})_i}
    =
    \plact{(\plact{\tau}{(\sigma_i)_i})}{(\tau_{i,j})_{i,j}}
    \zbox.
  \]
\end{lemma}
We can now use the definition of $\plact{(-)}{(-)}$ to give a precise definition
of $\mu$: it is the natural transformation
\[
  \mu \co \Sym \circ \Sym \To \Sym
\]
defined on a groupoid $A$ as the functor $\mu_A$, defined as follows. Given
$\seq{\seq{a_{i,j}}_{j\in \nset{n_i}}}_{i\in\nset m}\in \Sym(\Sym(A))$, we have
\[
  \mu_A(\seq{{a_{i,j}}_{j\in \nset{n_i}}}_{i\in\nset
    m}) = \seq{a_{i,j}}_{i\in\nset m,j\in\nset{n_i}}
\]
and, given a morphism $ u =
(\sigma,\seq{(\tau_i,\seq{f_{i,j}}_{j\in\nset{n_i}})}_{i\in\nset m}) \in
\Sym(\Sym(A)) $, we have
\[
  \mu_A(u) = (\plact\sigma{(\tau_i)_{i\in\nset
      m}},(f_{i,j})_{i\in\nset m, j\in\nset{n_i}})
  \zbox.
\]
One can then use \Cref{lem:plact-functorial} to show that $\mu$ is associative
in the monadic sense.

\subsection{The $\Sym$ monad on $\pmFunct$}

In order to show that the $\Sym$ monad on $\Gpd$ induces an adequate comonad
$\oc$ on $\Thin$, we just need to lift the monadic structure on $\Gpd$ to a
monadic structure on $\pmFunct$, and then conclude by
\Cref{prop:check-preserves-comonads}. We first show that the $\Sym$ endofunctor
on $\Gpd$ can be lifted to a $1$\morphism of $\pmFunct$.

\begin{proposition}
  \label{prop:sym-pres-pb-bpb}
  $\Sym$ preserves pullbacks and sends pullbacks that are bipullbacks to
  bipullbacks.
\end{proposition}
\begin{proof}
  Consider a pullback
  \begin{equation}
    \label{eq:sym-pres-pullbacks:pb}
    \begin{tikzcd}
      P
      \ar[r,"{r}"]
      \ar[d,"{l}"']
      \phar[rd,very near start,"\drcorner"]
      &
      R
      \ar[d,"f^R"]
      \\
      L
      \ar[r,"f^L"']
      &
      M
    \end{tikzcd}
  \end{equation}
  in $\Gpd$. In order to show that this pullback is preserved by $\Sym$, we just
  need to show that a pair of morphisms $(\sigma^L,(u^L_i)_{1 \le i \le n^L}) \in
  \Sym L$ and $(\sigma^R,(u^R_i)_{1 \le i \le n^R}) \in
  \Sym R$ which are projected to the same morphism in $\Sym M$ lifts to a unique
  morphism of $\Sym P$. But it is quite immediate, since the common projection
  on $\Sym M$ implies that $n^L = n^R$, $\sigma^L = \sigma^R$ and that
  $f^L(u^L_i) = f^R(u^R_i)$ for every $i \in \set{1,\ldots,n^L}$. Thus, $\Sym$
  preserves pullbacks.

  Now, assuming that \eqref{eq:sym-pres-pullbacks:pb} is moreover a bipullback,
  we are required to show that its image by $\Sym$ is also a bipullback. For
  this, we use the criterion given by
  \cite[Proposition~9]{DBLP:conf/lics/ClairambaultF23}. Let $\vec
  a=\seq{a_1,\ldots,a_n} \in \Sym L$ and $\vec b = \seq{b_1,\ldots,b_m} \in \Sym
  R$, and a morphism $v = (\sigma,(v_i)_{1 \le i \le n})$ between
  $\Sym(f^L)(\vec a)$ and $\Sym(f^R)(\vec b)$. We need to show that $v =
  \Sym(f^R)(u^R) \circ \Sym(f^L)(u^L)$ for some $u^L \in \Sym L$ and $u^R \in
  \Sym R$. Since $v = (\sigma,(\id)_{1 \le i \le n}) \circ (\id,(v_i)_{1 \le i
    \le n})$ and that $(\sigma,(\id)_{1 \le i \le n})$ is in the image of $\Sym
  R$, we may assume that $\sigma = \id$. Since~\eqref{eq:sym-pres-pullbacks:pb}
  is a bipullback, we have that $v_i = f^R(u^R_i) \circ f^L(u^L_i)$ for some
  $u^L_i \co a_i \to a'_i \in L$ and $u^R_i \co b'_i \to b_i \in R$ for every
  $i$. By taking $u^L = (\id,\seq{u^L_i}_{i})$ and $u^R =
  (\id,\seq{u^R_i}_{i})$, we have $v = \Sym(f^R)(u^R) \circ \Sym(f^L)(u^L)$ as
  wanted. Thus, by \cite[Proposition~9]{DBLP:conf/lics/ClairambaultF23}, the
  image of~\eqref{eq:sym-pres-pullbacks:pb} by $\Sym$ is a bipullback.
\end{proof}
Moreover, it is immediate to check that $\Sym$ preserves functor $f \co A \to B
\in \Gpd$ that are bijective on objects.

We now define $\Symp$ as the functor $\Gpd \to \Gpd$ mapping a groupoid $X \in
\Gpd$ to the subgroupoid of $\Sym X$ with the same objects but whose morphisms
are restricted to be the ones of the form $(\id,\seq{u_i}_i)$ in $\Sym X$, and
with the evident image of functors $X \to Y$. This functor comes with a
canonical embedding natural transformation $\iota \co \Symp \To \Sym$. Just like
for $\Sym$, we have
\begin{proposition}
  \label{prop:symp-pres-pb-bpb}
  $\Symp$ preserves pullbacks and sends pullbacks that are bipullbacks to
  bipullbacks.
\end{proposition}
\begin{proof}
  The proof for $\Sym$ of \Cref{prop:sym-pres-pb-bpb} directly adapts to the
  case of $\Symp$.
\end{proof}
\begin{proposition}
  \label{prop:iota-bicartesian}
  The natural transformation $\iota \co \Symp \To \Sym$ is bicartesian.
\end{proposition}
\begin{proof}
  Given $F \co A \to B \in \Gpd$, consider the natural square
  \begin{equation}
    \label{eq:iota-bicartesian:pb-bpb}
    \begin{tikzcd}
      \Symp A
      \ar[r,"\Symp (F)"]
      \ar[d,"\iota_A"']
      &
      \Symp B
      \ar[d,"\iota_B"]
      \\
      \Sym A
      \ar[r,"\Sym(F)"']
      &
      \Sym B
    \end{tikzcd}
  \end{equation}
  and consider a pair of morphisms $u = (\rho,(u_i)_{1 \le i \le m}) \in
  \Sym A$ and $v = (\sigma,(v_i)_{1 \le i \le n}) \in \Symp B$ which are
  projected to the same morphism in $\Sym B$. For the
  square~\eqref{eq:iota-bicartesian:pb-bpb} to be a pullback, we need to show
  that this pair can be lifted to a unique morphism of $\Symp A$. From the
  common projection on $\Sym B$, we get that $\rho = \sigma = \id$ and that $v_i
  = F(u_i)$ for every $i$. Thus, $u$ actually lifts to $\Symp A$ through
  $\iota_A$, which is the required unique lifting.
  Thus,~\eqref{eq:iota-bicartesian:pb-bpb} is a pullback.

  We now show that~\eqref{eq:iota-bicartesian:pb-bpb} is a bipullback. Let $\vec
  a=\seq{a_1,\ldots,a_n} \in \Sym A$ and $\vec b = \seq{b_1,\ldots,b_m} \in
  \Symp B$, and a morphism $v = (\sigma,(v_i)_{1 \le i \le n})$ between
  $\Sym(F)(\vec a)$ and $\iota_B(\vec b) = \vec b$. We need to show that $v =
  \iota_B(u^R) \circ \Sym(F)(u^L)$ for some $u^L \in \Sym A$ and $u^R \in \Symp
  B$. Since $v = (\id,(v_{\finv\sigma(i)})_{1 \le i \le n}) \circ
  (\sigma,(\id)_{1 \le i \le n})$ and that $(\sigma,(\id)_{1 \le i \le n})$ is
  in the image of $\Sym (F)$, we may assume that $\sigma = \id$. But then, $v$
  is in the image of $\iota_B$, so that, by
  \cite[Proposition~9]{DBLP:conf/lics/ClairambaultF23},~\eqref{eq:iota-bicartesian:pb-bpb}
  is a bipullback.
\end{proof}
We thus get that
\begin{proposition}
  $(\Sym,\Symp,\iota)$ is a \pmfunctor.
\end{proposition}
\begin{proof}
  The other conditions for being a \pmfunctor are readily verified.
\end{proof}
We now provide liftings in $\pmFunct$ for the natural transformations $\eta \co
\unit{\Gpd} \To \Sym$ and $\mu \co \Sym \circ \Sym \To \Sym$. For this, we
provide $\eta^+ \co \unit{\Gpd} \To \Symp$ and $\mu^+ \co \Symp \circ \Symp \To
\Symp$ so that $(\eta,\eta^+)$ and $(\mu,\mu^+)$ define \pmtransformations.
Actually, this is easy: by the equation~\eqref{eq:pmtransformation-eq}, $\eta^+$
is essentially $\eta$ and
$\mu^+_X$ is the adequate restriction of $\mu_X$ to the subgroupoid $\Symp(\Symp
X)$ of $\Sym(\Sym X)$.
\begin{proposition}
  We have \pmtransformations
  \begin{gather*}
    (\eta,\eta^+)\co (\unit{\Gpd},\unit{\Gpd},\unit{\unit{\Gpd}}) \To (\Sym,\Symp,\iota)
    \shortintertext{and}
    (\mu,\mu^+) \co (\Sym\Sym,\Symp\Symp,\iota\iota)\To (\Sym,\Symp,\iota)
    \zbox.
  \end{gather*}
\end{proposition}
\begin{proof}
  By \Cref{lem:pmtransformation-one-pb-bpb}, we just need to show that
  $\eta,\mu$ are bicartesian natural transformations. We only give the proof for
  $\mu$, since the bicartesianness of $\eta$ is quite easy to show.

  Consider a functor $F \co A \to B \in \Gpd$ and the natural square
  \begin{equation}
    \label{eq:eta-mu-pmtransformations:nat-square}
    \begin{tikzcd}[cs=6em]
      \Sym(\Sym A)
      \ar[r,"\Sym(\Sym (F))"]
      \ar[d,"\mu_A"']
      &
      \Sym(\Sym B)
      \ar[d,"\mu_B"]
      \\
      \Sym A
      \ar[r,"\Sym(F)"']
      &
      \Sym B
      \zbox.
    \end{tikzcd}
  \end{equation}
  In order to show that it is a pullback, we consider a pair of morphisms $u =
  (\sigma,\seq{u_i}_{1 \le i \le n}) \in \Sym A$ and $v =
  (\rho,\seq{(\rho'_i,\seq{v_{i,j}}_{1 \le j \le m_i})}_{1 \le i \le l}) \in
  \Sym(\Sym B)$ that project on the same morphism of $\Sym B$, and show that it
  can be lifted to a unique morphism of $\Sym(\Sym A)$ adequately projecting on
  $u$ and $v$. By the common projection on $\Sym B$, we have that $\sigma =
  \plact{\rho}{(\rho'_i)_i}$, so that $u$ can be written $\mu_A(\tilde u)$ for
  $\tilde u = (\rho,\seq{(\rho'_i,\seq{\tilde u_{i,j}}_{1 \le j \le m_i})}_{1
    \le i \le l})$ for some adequate morphisms $\tilde u_{i,j}$. We moreover
  have that $F(\tilde u_{i,j}) = v_{i,j}$, so that $\tilde u$ is a lift for the
  pair $(u,v)$, and it can be easily proved to be unique, so that the
  square~\eqref{eq:eta-mu-pmtransformations:nat-square} is a pullback.

  We now show that it is a bipullback. So consider $\vec a = \seq{a_i}_{1 \le i
    \le n} \in \Sym A$ and $\vec b = \seq{\seq{b_{i,j}}_{1 \le j \le m_i}}_{1
    \le i \le l} \in \Sym(\Sym B)$ and a morphism $w = (\tau,\seq{w_i}_{1 \le i
    \le n}) \co \Sym(F)(\vec a) \to \mu_B(\vec b) \in \Sym B$. We have to show
  that $w = \mu_B(v) \circ \Sym(F)(u)$ for some $u \co \vec a \to \vec a' \in
  \Sym A$ and $v \co \vec b' \to \vec b \in \Sym(\Sym B)$. Since $w =
  (\id,\seq{w_{\finv\tau(i)}}_{1 \le i \le n}) \circ (\tau,\seq{\id}_{1 \le i
    \le n})$ and that $(\tau,\seq{\id}_{1 \le i \le n})$ is in the image of
  $\Sym(F)$, we may assume that $\tau = \id$. But then, it is clear that we may
  find $v \co \vec b' \to \vec b$ such that $w = \mu_B(v)$, so
  that~\eqref{eq:eta-mu-pmtransformations:nat-square} is a bipullback
  by~\cite[Proposition~9]{DBLP:conf/lics/ClairambaultF23}.
\end{proof}
\begin{proposition}
  \label{prop:sym-monad}
  The triple $((\Sym,\Symp,\iota),(\eta,\eta^+),(\mu,\mu^+))$ defines a monad on
  $\pmFunct$.
\end{proposition}
\begin{proof}
  One just need to check the monad axioms for this triple. But they directly
  follow from the monad axioms satisfied by the monad $(\Sym,\eta,\mu)$ on
  $\Gpd$.
\end{proof}
We may now conclude that
\begin{proposition}
  The $\Sym$ functor defines a pseudocomonad $\check\Sym$ on $\Thin$, with
  $\check\eta$ as counit and $\check\mu$ as comultiplication.
\end{proposition}
\begin{proof}
  This is a consequence of \Cref{prop:check-preserves-comonads,prop:sym-monad}.
\end{proof}

\section{Details on intersection types}

\subsection{The correspondence for the objects}

\begin{proof}[{Proof of \Cref{prop:corres-typeinterp-it}}] We prove the
  statement by induction on $A$. If $A = o$, then $\typeinterp A =
  (\termcat,\ldots)$ (the unique thin groupoid with the terminal groupoid as
  underlying groupoid) so that we can take $K_A = \unit{\termcat}$. If $A = B
  \synto C$, then $K_A$ is defined as the composite
  \[
    \begin{tikzcd}[cs=5em]
      \Ob(\typeinterp {B \synto C})
      \ar[r,equals]
      &[-3em]
      \Ob(\oc \typeinterp B) \times \Ob(\typeinterp C)
      \\
      \cdots
      \ar[r,"K^\oc_{B} \times K_C"]
      &
      \begin{multlined}[t]
        \set{ \seq{\beta_1,\ldots,\beta_n} \mid \itrefinfj{\beta_i }{ B }\text{ for $i \in
            \set{1,\ldots,n}$}}
        \times
        \set{\gamma \mid \itrefinfj{\gamma }{ C}}
      \end{multlined}
      \\[-1em]
      \cdots
      \ar[r,"\sim"]
      &
      \set{\delta \mid \itrefinfj{\delta }{ (B \synto C)}}\zbox.
    \end{tikzcd}
  \]
  Finally, we put $K^\oc_{A} = \oc K_A$, assuming the same encoding of sequences
  for the $\oc$ construction and the multilinear refinement types, for
  simplicity.
\end{proof}
\begin{proof}[{Proof of \Cref{prop:charact-terminterp-obj}}]
  Let us first give some precisions on the \eq{and similarly for $\terminterpoc
    M$} part. By that, we mean that, given a derivation $\simplefj\Gamma{M}{A}$,
  for every $\gamma \in \ctxtinterp \Gamma$ and $\vec a \in \oc \typeinterpb A$,
  we have a bijection
  \[
    \stsupp{\terminterpoc M}_{\gamma,\vec a} \simeq \set{\pi \mid \text{$\pi$ is a
        derivation of $\itfj{K_\Gamma(\gamma)}{\Gamma}{M}{K^\oc_A(\vec a)}{A}$}}
    \zbox.
  \]
  We then prove the property by induction on the derivations
  $\simplefj{\Gamma}{M}{A}$ and $\simplefj{\Gamma}{M}{A}$:
  \begin{itemize}
  \item in the case of the variable typing rule $\simplefj{x_1: A_1,\ldots,x_n:
      A_n}{x_i}{A_i}$, if $\gamma =
    (\seq{},\ldots,\seq{},\seq{a},\seq{},\ldots,\seq{})$ then
    \[
      \begin{tikzcd}
        \Ob(\stsupp{\terminterp {x_i}}_{\gamma,a})
        = \set{a}
        \ar[r,"\sim"]
        &
        \begin{multlined}[t]
          \{\pi \mid \text{$\pi$ is
            a derivation of}\\
          \hskip-3em \itfj{\itrefinfj[x_1]{\seq{}}{A_1},\ldots,\itrefinfj[x_i]{\seq{K_{A_i}(a)}}{A_i},\ldots,\itrefinfj[x_n]{\seq{}}{A_n}}{\Gamma}{x}{K_{A_i}(a)}{A} \}
        \end{multlined}
        \\[-1em]
        \hbox{}
        \ar[r,"\sim"]
        &
        \set{\pi \mid \text{$\pi$ is
            a derivation of $\itfj{K_\Gamma(\gamma)}{\Gamma}{x}{K_{A_i}(a)}{A_i}$}}
      \end{tikzcd}
    \]
    which is canonically isomorphic by $K_{A_i}$ to $\set{\alpha \mid
      \itrefinfj{\alpha}{A_i}}$. Otherwise, for other $\gamma$'s, we have
    $\Ob(\stsupp{\terminterp M}_{\gamma,a}) = \emptyset$ and no derivations for
    $\itfj{K_\Gamma(\gamma)}{\Gamma}{x}{K_{A_i}(a)}{A_i}$;
  \item in the case of an intersection typing $\simplefj{\Gamma}{M}{A}$ derived
    from a proof of $\simplefj{\Gamma}{M}{A}$, and given $\vec a =
    \seq{a_1,\ldots,a_l} \in \typeinterp A$, we have
    \begin{align*}
      \terminterpoc{M}_{\gamma,\vec a}
      &
        \simeq
        \bigsqcup_{\substack{\vec \gamma = \seq{\gamma_1,\ldots,\gamma_l} \in \oc\ctxtinterp\Gamma\\ \gamma_1\ctxtplus\cdots\ctxtplus\gamma_l  = \gamma}}
      \;\;
      \prod_{i=1}^l \terminterp {M}_{\gamma_i,a_i}
      \\
      &
        \simeq
        \bigsqcup_{\substack{\vec \gamma = \seq{\gamma_1,\ldots,\gamma_l} \in \oc\ctxtinterp\Gamma\\ \gamma_1\ctxtplus\cdots\ctxtplus\gamma_l  = \gamma}}
      \;\;
      \prod_{i=1}^l \set{\pi \mid \text{$\pi$ derivation of $\itfj{K_\Gamma(\gamma_i)}{\Gamma}{M}{K_{A}(a_i)}{A}$}}
      \\
      &\simeq
        \set{\pi \mid \text{$\pi$ derivation of $\itfj{K_\Gamma(\gamma)}{\Gamma}{M}{K^\oc_{A}(\vec a)}{A}$}}\zbox;
    \end{align*}

  \item in the case of an application $\Gamma \judge M\, N: B$ constructed from
    two derivations of $\Gamma \judge M : A \synto B$ and $\Gamma \judgem N : A$
    for some unique simple type~$A$,
    considering the definition $\terminterp{M\,N}$, we have that
    \begin{align*}
      \terminterp{M\,N}_{\gamma,b}
      &
        \simeq
        \bigsqcup_{\substack{\gamma_1,\gamma_2 \in \ctxtinterp \Gamma\\ \gamma_1\ctxtplus\gamma_2 = \gamma}}
        \;
        \;
        \bigsqcup_{\vec a \in \oc \typeinterp A}
      \;
      \;
      \terminterp {M}_{\gamma_1,(\vec a,b)}
      \times
      \terminterpoc {N}_{\gamma_2,\vec a}
      \\
      & 
        \simeq
        \bigsqcup_{\substack{\gamma_1,\gamma_2 \in \ctxtinterp \Gamma\\ \gamma_1\ctxtplus\gamma_2 = \gamma}}
        \;
        \;
        \bigsqcup_{\vec a \in \oc \typeinterp A}
      \begin{aligned}[t]
        &\set{\pi_1 \mid \text{$\pi_1$ derivation of $\itfj{K_\Gamma(\gamma_1) }{ \Gamma }{ M }{ K_{A \synto B}(\vec a,b) }{ A \synto B}$}}
        \\
        &\quad\quad\times
          \set{\pi_2 \mid \text{$\pi_2$ derivation of $\itfj{K_\Gamma(\gamma_2) }{ \Gamma }{ N }{ K^\oc_{A}(\vec a) }{ A}$}}
      \end{aligned}
      \\
      &\simeq
        \set{\pi \mid \text{$\pi$ derivation of $\itfj{K_\Gamma(\gamma) }{ \Gamma }{ M\,N }{ K_B(b) }{ B}$}}\zbox;
    \end{align*}
  \item in case of a lambda-abstraction $\Gamma \judge \lambda x.M : A \synto
    B$, given $\gamma \in \ctxtinterp \Gamma$ and $(\kappa, \beta) \in
    \typeinterp{A \synto B}$, we have
    \begin{align*}
      \terminterp{\lambda x.M}_{\gamma,(\kappa, \beta)}
      &\simeq
        \terminterp{M}_{\lstconsrp\gamma\kappa,\beta}
      \\
      &\simeq
        \set{\pi \mid \text{$\pi$ derivation of $\itfj{K_{(\Gamma,x:A)}(\lstconsrp\gamma\kappa) }{ (\Gamma,x:A) }{ M }{ K_{B}(\beta) }{ B}$}}
      \\
      &\simeq
        \set{\pi \mid \text{$\pi$ derivation of $\itfj{K_{\Gamma}(\gamma) }{ \Gamma }{ \lambda x.M }{ K_{A \synto B}((\kappa, \beta)) }{ A \synto B}$}}
    \end{align*}
    where, for $\gamma = (\gamma_1,\ldots,\gamma_n) \in \ctxtinterp \Gamma$, we write $\gamma\lstconsr \kappa$ for
    $(\kappa_1,\ldots,\kappa_n,\kappa)\in\ctxtinterp {\Gamma,x:A}$.
  \end{itemize}
\end{proof}

\subsection{The groupoids of intersection types}

Here, we give details about the structure of the groupoids $\itgpd A$ (and its
multilinear version) of intersection types and the associated morphisms that
refine a simple type~$A$.

Given a simple type $A$, we first define $\Ob(\itgpd A)$ as the set of linear
intersection types $\alpha$ such that $\itrefinfj{\alpha }{ A}$. We then define
the symmetries\todo{faire une passe sur tout le document pour appeler morphismes
  \eq{symétries} pour les groupoïdes?} $\alpha \to \alpha'$ of $\itgpd A$ as the
linear intersection type morphisms $\phi$ such that $\itmrefinfj{\phi }{ \alpha
}{ \alpha' }{ A}$. We have the convenient property that
\begin{lemma}
  \label{lem:unique-src-tgt-for-itmorphs}
  Given a linear intersection type morphism $\phi$, there is at most one pair
  $(\alpha,\alpha')$ of linear intersection types such that $\itmrefinfj{\phi }{\alpha}{\alpha' }{A}$, and similarly for multilinear intersection type
  morphisms $\theta$.
\end{lemma}
\begin{proof}
  By a simple induction on $\phi$ and $\theta$.
\end{proof}
Thus, given a morphism $\phi$ of $\itgpd A$, we may write $\src(\phi)$ and
$\tgt(\phi)$ for the unique $\alpha$ and $\alpha'$ such that $\itmrefinfj{\phi }{\alpha}{\alpha' }{A}$.

We may define similarly (the beginning of) a groupoid $\itgpdoc{A}$, whose
objects are the multilinear intersection types $\vec\alpha$ such that
$\itrefinfj{\vec\alpha }{ A}$, and whose morphisms $\vec\alpha \to \vec\alpha'$
are the multilinear intersection type morphisms are the $\widetilde\phi$'s such
that $\itmrefinfj{\widetilde\phi}{\vec\alpha}{\vec\alpha'}{A}$. By
\Cref{lem:unique-src-tgt-for-itmorphs}, we may write $\src(\widetilde\phi)$ and
$\tgt(\widetilde\phi)$ for these unique $\vec\alpha$ and $\vec\alpha'$.

Given $\alpha \in \Ob(\itgpd A)$ and $\vec\alpha \in\Ob(\itgpdoc{A})$, we define
$\itsid\alpha \in \Ar(\itgpd A)$ and $\itsid{\vec\alpha} \in \Ar(\itgpdoc {A})$
such that $\itmrefinfj{\itsid\alpha}{\alpha}{\alpha}{A}$ and
$\itmrefinfj{\itsid{\vec\alpha} }{ \vec\alpha }{ \vec\alpha }{ A}$ by mutual
induction on the derivations of $\itrefinfj{\alpha}{A}$ and $\itrefinfj{\vec\alpha
}{ A}$:
\begin{itemize}
\item in the case of the axiom $\itrefinfj{\itsstar }{ o}$, we put $\itsid\star =
  \itsstarid$. We then have $\itmrefinfj{\itsid\star }{ \itsstar }{ \itsstar }{
  A}$ by corresponding rule for morphisms;

\item in the case of the linear arrow $\itrefinfj{(\vec\beta \synlinto \gamma)
  }{ (B \synto C)}$, by induction hypothesis, we get $\itsid{\vec\beta} \in
  \Ar(\itgpdoc{B})$ and $\itsid\gamma \in \Ar(\itgpd{C})$ such that
  $\itmrefinfj{\itsid{\vec\beta} }{ {\vec\beta} }{ \vec\beta }{ B}$ and
  $\itmrefinfj{\itsid\gamma }{\gamma }{\gamma }{C}$. We then put $\itsid{\vec\beta
    \synlinto \gamma}$ to be $\itsid{\vec\beta} \synlinto \itsid\gamma$, for which we
  are able to derive $\itmrefinfj{\itsid{\vec\beta \synlinto \gamma} }{(\vec\beta
    \synlinto \gamma) }{(\vec\beta \synlinto \gamma) }{B \synto C}$ using the
  corresponding rule for morphisms;
\item in the case of the multilinear formation $\itrefinfj{\vec\alpha }{ A}$ for
  $\vec\alpha = \seq{\alpha_1,\ldots,\alpha_n}$, by induction hypothesis, we get
  $\itsid{\alpha_i} \in \Ar(\itgpd A)$ such that $\itmrefinfj{\itsid{\alpha_i}
  }{\alpha_i }{\alpha_i}{A}$ for $i\in \set{1,\ldots,n}$. We then put
  $\itsid{\vec\alpha} =
  (\unit{},\seq{\itsid{\alpha_1},\ldots,\itsid{\alpha_n}})$, for which we can
  easily derive that $\itmrefinfj{\itsid{\vec\alpha} }{\vec\alpha}{\vec\alpha}{
    A}$ using the corresponding rule for morphisms.
\end{itemize}

Given $\phi_1,\phi_2 \in \Ar(\itgpd A)$ (\resp
$\widetilde\phi_1,\widetilde\phi_2 \in \Ar(\itgpdoc {A})$), we say that they are
\emph{composable} when $\tgt(\phi_1) = \src(\phi_2)$ (\resp
$\tgt(\widetilde\phi_1) = \src(\widetilde\phi_2)$). It happens that two
composable intersection type morphisms are very \eq{similar} in their
construction:
\begin{lemma}
  \label{lem:composable-implications}
  Given $\phi_1,\phi_2 \in \Ar(\itgpd A)$ such that they are composable, we have
  that:
  \begin{itemize}
  \item if $\phi_1 = \itsstarid$, then $A = o$ and $\phi_2 = \itsstarid$;
  \item if $\phi_1 = \widetilde\phi_1 \synlinto \psi_1$, then
    \begin{itemize}
    \item $A = B \synto C$ for some unique simple types $B$ and $C$,
    \item $\widetilde\phi_1 \in \Ar(\itgpdoc B)$ and $\psi_1 \in \Ar(\itgpd C)$,
    \item $\phi_2 = \widetilde\phi_2 \synlinto \psi_2$ for some unique $\widetilde\phi_2 \in
      \Ar(\itgpd B)$ and $\psi_1 \in \Ar(\itgpd C)$,
    \item and $\widetilde\phi_1,\widetilde\phi_2$ (\resp $\psi_1,\psi_2$) are composable.
    \end{itemize}
  \end{itemize}
  Similarly, given $\widetilde\phi_1 = (\sigma_1,\seq{\phi_{1,1},\ldots,\phi_{1,n_1}})$ and
  $\widetilde\phi_2 = (\sigma_2,\seq{\phi_{2,1},\ldots,\phi_{2,n_2}})$ in $\Ar(\itgpdoc{A})$
  for some $n_i \in \N$, $\sigma_i \in \Sgroup_{n_i}$ and intersection type
  morphisms $\phi_{i,1},\ldots,\phi_{i,n_i}$ for $i \in \set{1,2}$, such that
  $\widetilde\phi_1$ and $\widetilde\phi_2$ are composable, we have that $n_1 = n_2$ and
  $\phi_{1,j}$ is composable with $\phi_{2,\sigma_1(j)}$ for $j\in
  \set{1,\ldots,n_1}$.
\end{lemma}
\begin{proof}
  By mutual induction on the derivations of $\phi_1$ and $\widetilde\phi_1$.
\end{proof}
Given a simple type $A$ and two composable linear (\resp multilinear)
intersection type morphisms $\phi_1,\phi_2 \in \Ar(\itgpd{A})$ (\resp
$\widetilde\phi_1,\widetilde\phi_2 \in \Ar(\itgpdoc{A})$), we now define their
\emph{composition} $\phi_2 \itscirc \phi_1$ (\resp $\widetilde\phi_2 \itscirc \widetilde\phi_1$) by
mutual induction. We use \Cref{lem:composable-implications} to give a complete
definition with a minimal case analysis:
\begin{itemize}
\item we put $\itsstarid \itscirc \itsstarid = \itsstarid \circ \itsstarid$;
\item given composable $\phi_1 = \zeta_1 \synlinto \psi_1$ and $\phi_2 = \zeta_2
  \synlinto \psi_2$, we put $\phi_2\itscirc\phi_1 = (\zeta_2 \itscirc \zeta_1)
  \synlinto (\psi_2 \itscirc \psi_1)$;
\item given composable $\widetilde\phi_1 = (\sigma_1,\seq{\psi_{1,1},\ldots,\psi_{1,n}})$ and
  $\widetilde\phi_2 = (\sigma_2,\seq{\psi_{2,1},\ldots,\psi_{2,n}})$, we put
  \[
    \widetilde\phi_2 \itscirc \widetilde\phi_1 = (\sigma_2 \circ \sigma_1,\seq{\psi_{2,\sigma(1)} \itscirc \psi_{1,1},\ldots,\psi_{2,\sigma(n)} \itscirc \psi_{1,n}})
    \zbox.
  \]
\end{itemize}
Given a simple type $A$ and a composable linear (\resp multilinear) intersection
type morphism $\phi \in \itgpd{A}$ (\resp $\widetilde\phi \in \Ar(\itgpdoc{A})$), we
now define its \emph{inverse} $\finv\phi$ (\resp $\finv{\widetilde\phi}$) by induction on
the derivations:
\begin{itemize}
\item we put $\finv {\itsstarid} = \itsstarid$;
\item given $\phi = \widetilde\phi \synlinto \psi$, we put $\finv\phi = \finv{\widetilde\phi}
  \synlinto \finv\psi$;
\item given $\widetilde\phi = (\sigma,\seq{\psi_{1},\ldots,\psi_{n}})$, we put $\finv{\widetilde\phi} =
  (\finv\sigma,\seq{\psi_{\finv\sigma(1)},\ldots,\psi_{\finv\sigma(n)}})$.
\end{itemize}
The above operations assemble into a groupoidal structure:
\begin{lemma}
  Given a simple type $A$, we have that
  \begin{itemize}
  \item $\itscirc$ is an associative composition operation on $\itgpd A$ (\resp
    $\itgpdoc {A}$) with $\itsid{}$ as unit, making $\itgpd {A}$ (\resp
    $\itgpdoc{A}$) a category;
  \item for every $\phi \in \Ar(\itgpd A)$ (\resp $\widetilde\phi \in \Ar(\itgpdoc
    {A})$), we have $\finv\phi \itscirc \phi = \itsid{\src(\phi)}$ and $\phi
    \itscirc \finv\phi = \itsid{\tgt(\phi)}$ (\resp $\finv{\widetilde\phi} \itscirc \widetilde\phi
    = \itsid{\src(\widetilde\phi)}$ and $\widetilde\phi \itscirc \finv{\widetilde\phi} =
    \itsid{\tgt(\widetilde\phi)}$), so that $\itgpd A$ (\resp $\itgpdoc {A}$) is a
    groupoid.
  \end{itemize}
\end{lemma}
\begin{proof}
  By simple inductions.
\end{proof}

\begin{proof}[Proof of \Cref{prop:isom-itgpd-typeinterp}]%
  The functors $K_A$ and $K^\oc_A$ are built as the direct extensions to
  symmetries of the ones built in the proof of \Cref{prop:corres-typeinterp-it},
  since the definition of multilinear intersection type morphisms closely
  follows the definition of the action of $\oc$ on groupoids and their
  symmetries.
\end{proof}
\subsection{Resource context groupoid}

We recall that we only consider contexts, resource contexts and resource
morphism contexts that are well-formed, so that we will often omit to precise
that such contexts are well-formed for conciseness.

Given a (well-formed) context $\Gamma = (x_1:A_1,\ldots,x_n:A_n)$, we give some
details about the definition of the groupoid $\itgpd\Gamma$: its objects are the
resource contexts $\Theta = (\itrefinfj[x_1 ]{\kappa_1 }{ A_1}, \ldots,\itrefinfj[ x_n ]{\kappa_n
}{ A_n})$, and its morphisms of type $\Theta \to \Theta'$, for another
resource context $\Theta' = (\itrefinfj[x_1 ]{\kappa'_1 }{ A_1}, \ldots,\itrefinfj[ x_n ]{\kappa'_n
}{ A_n})$, are the resource morphisms contexts $\Xi = (\itmrefinfj[x_1]{\theta_1 }{\kappa_1 }{\kappa'_1 }{A_1},\ldots,\itmrefinfj[x_n]{\theta_n }{\kappa_n }{\kappa'_n }{A_n})$. We then write $\src(\Xi)$ for $\Theta$ and $\tgt(\Xi)$
for $\Theta'$. Two resource morphism contexts $\Xi_1$ and $\Xi_2$ of $\itgpd
\Gamma$ defined by
\[
  \Xi_i = (\itmrefinfj[x_1]{\theta_{i,1} }{\kappa_{i,1} }{\kappa'_{i,1} }{A_1},\ldots,\itmrefinfj[x_n]{\theta_{i,n} }{\kappa_{i,n} }{\kappa'_{i,n} }{A_n})
\]
are \emph{composable} when $\tgt(\Xi) = \src(\Xi')$. In this case, we define
their composite as
\[
  \Xi_2 \ctxtcirc \Xi_1 = (\itmrefinfj[x_1]{\theta_{2,1} \itscirc \theta_{1,1} }{\kappa_{1,1} }{\kappa'_{2,1} }{A_1},\ldots,\itmrefinfj[x_n]{\theta_{2,n} \itscirc
  \theta_{1,n} }{\kappa_{1,n} }{\kappa'_{2,n} }{A_n})
  \zbox.
\]
Moreover, given $\Xi$ as above, there is a resource morphism context $\finv \Xi$
defined by
\[
  \finv \Xi = (\itmrefinfj[x_1]{\finv\theta_1 }{\kappa'_1 }{\kappa_1 }{A_1},\ldots,\itmrefinfj[x_n]{\finv\theta_n }{\kappa'_n }{\kappa_n }{A_n})
\]
and which is the \emph{inverse} of $\Xi$. Given a resource context $\Theta =
(\itrefinfj[x_1]{\kappa_1 }{ A_1},\ldots,\itrefinfj[x_n]{\kappa_n }{ A_n})$, there is an
\emph{identity resource morphism context} $\ctxtid\Theta$ defined by
\[
  \ctxtid\Theta
  =
  (\itmrefinfj[x_1]{\itsid{\kappa_1} }{\kappa_1 }{\kappa_1 }{A_1},\ldots,\itmrefinfj[x_n]{\itsid{\kappa_n} }{\kappa_n }{\kappa_n }{A_n})
  \zbox.
\]
Following what was done in the previous section, we readily have that
\begin{proposition}
  $\itgpd \Gamma$ has a structure of groupoid and, as such, it is isomorphic to
  the groupoid $\itgpdoc {A_1} \times \cdots \times \itgpdoc {A_n}$.
\end{proposition}

\subsection{Morphisms between derivations}

Here, we give more details about the definition of the $\plact{\rho}{(-)}$
operation on families of multilinear intersection type morphisms and resource
morphism contexts.

Let $m \in \N$, $\rho \in \Sgroup_m$. Given a family of morphisms
$(\itmrefinfj{\widetilde\phi_j}{\vec{\alpha}_j }{\vec{\alpha}'_j }{A})_{1 \le j
  \le m}$ where $\widetilde\phi_j = (\sigma_j,\seq{\phi_{j,k}}_{1 \le k \le
  l_j})$ with $l_j$ the length of $\vec{\alpha}_j$ (and $\vec{\alpha}'_j$) for
every $j \in \set{1,\ldots,m}$, we write $\plact{\rho}{(\widetilde\phi_j)_{1\le
    j \le m}}$ for the  multilinear intersection type morphism $\widetilde\phi$
defined by
\[
  \widetilde\phi
  \qp=
  (\plact{\rho}{(\sigma_j)_{1 \le j \le m}},\seq{\phi_{1,k}}_{1 \le k \le l_1} \seqplus \cdots
  \seqplus \seq{\phi_{m,k}}_{1 \le k \le l_m})
\]
where $\seq{\phi_{1,k}}_{1 \le k \le l_1} \seqplus \cdots \seqplus
\seq{\phi_{m,k}}_{1 \le k \le l_m}$ is the mere concatenation of the sequences
of morphisms. Note that we then have the refinement
\[
  \itmrefinfj{\widetilde\phi}%
  {\vec{\alpha}_1 \seqplus \cdots \seqplus \vec{\alpha}_m}%
  {\vec{\alpha}'_{\rho^{-1}(1)} \seqplus \cdots \seqplus \vec{\alpha}'_{\rho^{-1}(m)}}{A}
  \zbox.
\]
Now, given a family
\[
  (\Xi_j)_{1 \le j \le m} =
  ((\itmrefinfj[x_{i}]{\widetilde\phi_{i,j}}{\vec{\alpha}_{i,j}}{\vec{\alpha}'_{i,j}}{A_{i}})_{1
    \le i \le n})_{1 \le j \le m}
\]
of resource morphism contexts, all refining a common context $\Gamma = x_1 :
A_1,\ldots, x_n : A_n$, we define $\rmcplact{\rho}{(\Xi_j)_{1 \le j \le m}}$ as
the resource morphism context
\[
  \rmcplact{\rho}{(\Xi_j)_{1 \le j \le m}}
  \qp=
  (\itmrefinfj[x_{i}]{\plact{\rho}{(\widetilde\phi_{i,j})_{1 \le j \le
        m}}}{\vec\alpha_{i,1} \seqplus \cdots \seqplus\vec\alpha_{i,m}}{\vec\alpha_{i,\finv\rho(1)} \seqplus \cdots \seqplus\vec\alpha_{i,\finv\rho(m)}}{A_{i}})_{1 \le i \le n}
  \zbox.
\]

\subsection{The intersection type groupoid for a term}

We now give some details about the definition of the groupoid $\itgpd M$ for a
well-typed \lterm $\Gamma \judge M : A$. Given a derivation $\pi$ of
$\itmfj{\Xi}{\Gamma}{M}{\phi}{\alpha}{\alpha'}{A}$, one can define a derivation
$\src(\pi)$ of $\itfj{\dom(\Xi)}{\Gamma}{M}{\alpha}{A}$ and a derivation
$\tgt(\pi)$ of $\itfj{\cod(\Xi)}{\Gamma}{M}{\alpha'}{A}$ by induction on $\pi$
(and similarly for multilinear judgements). We only give the definition of
$\src(\pi)$ in \Cref{fig:def-src-pi} since the one of $\tgt(\pi)$ is similar.
\begin{figure}
  \centering
  \[
    \hss
    \begin{gathered}
      \src\left(
        \scalebox{0.7}{$
          \inferrule{\incomment{(x_1:A_1,\ldots,x_n:A_n)\ctxtwf\quad} 
            (\itmrefinfj{\phi}{\alpha}{\alpha'}{A_i})
          }{
            \itmfj{\ldots, \itmrefinfj[x_i]{(\unit{\set{1}},\seq{\phi})}{\seq{\alpha}}{\seq{\alpha'}}{A_i}, \ldots }%
            {}{x_i }{\phi }{\alpha }{\alpha' }{A_i}}
          $}
      \right) =
      \scalebox{0.7}{$
        \inferrule{\incomment{(x_1:A_1,\ldots,x_n:A_n)\ctxtwf\quad} 
          (\itrefinfj{\alpha}{A_i}) }{
          \itfj{\ldots, \itrefinfj[x_i]{\seq{\alpha}}{A_i}, \ldots}{x_1:A_1,\ldots,x_n:A_n}{x_i}{\alpha}{A_i}}
        $}
      \\
      \src\left(
        \scalebox{0.7}{$
          \inferrule{\inferrule{\pi_1}{\itmfj{\Xi}{}{M}{(\widetilde{\phi} \synlinto \psi)}{(\vec{\alpha} \synlinto \beta)}{(\vec{\alpha}' \synlinto \beta')}{A \synto B}}%
            \quad%
            \inferrule{\pi_2}{\itmfj{\Xi'}{}{N}{\widetilde{\phi}}{\vec{\alpha}}{\vec{\alpha}'}{A}}}%
          {\itmfj{\Xi \rmctxtplus\Xi' }{}{M\,N}{\psi }{\beta }{\beta' }{B}}
          $}
      \right) =
      \scalebox{0.7}{$
        \inferrule{\inferrule{\src(\pi_1)}{\itfj{\dom(\Xi)}{}{M}{(\vec{\alpha} \synlinto \beta)}{A \synto B}}%
          \quad%
          \inferrule{\src(\pi_2)}{\itfj{\dom(\Xi')}{\Gamma}{N}{\vec{\alpha}}{A}}}%
        {\itfj{\dom(\Xi \rmctxtplus\Xi' )}{\Gamma}{M\,N}{\beta }{B}}
        $}
      \\
      \src\left(
        \scalebox{.70}{$
          \inferrule{\inferrule{\pi'}{\itmfj{\Xi,\itmrefinfj[x]{\widetilde{\phi}}{\vec{\alpha}}{\vec{\alpha}'}{A}}{}{M}{\psi}{\beta}{\beta'}{B}}}%
          {\itmfj{\Xi}{}{\lambda x.\,M}{(\widetilde{\phi} \synlinto \psi)}{(\vec{\alpha} \synlinto \beta)}{(\vec{\alpha}' \synlinto \beta')}{A \synto B}}
          $}
      \right) =
      \scalebox{.70}{$
        \inferrule{\inferrule{\src(\pi')}{\itfj{\dom(\Xi),\itrefinfj[x]{\vec{\alpha}}{A}}{\Gamma,x:A}{M}{\beta}{B}}}%
        {\itfj{\dom(\Xi)}{\Gamma}{\lambda x.\,M}{(\vec{\alpha} \synlinto \beta)}{A \synto B}}
        $}
      \\
      \src\left(
        \scalebox{.70}{$
          \inferrule{n\in \N\quad \sigma \in \Sgroup_n\quad \forall i \in
            \set{1,\ldots,n},\;\inferrule*{\pi_i}{\itmfj{\Xi_i}{}{M}{\phi_i}{\alpha_i}{\alpha'_i}{A}}}%
          {\itmfj{\plact{\sigma}{(\Xi_i)_{1 \le i\le n}}
            }{}{M}{(\sigma,\seq{\phi_1,\ldots,\phi_n})
            }{\seq{\alpha_1,\ldots,\alpha_n}
            }{\seq{\alpha'_{\finv\sigma(1)},\ldots,\alpha'_{\finv\sigma(n)}}}{A}}
          $}
      \right)
      =
      \scalebox{.70}{$
        \inferrule{%
          \text{$\forall i \in \set{1,\ldots,k}$,}\quad%
          \inferrule*{\src(\pi_i)}{\itfj{\dom(\Xi_i)}{\Gamma}{M}{\alpha_i}{A}}%
        }%
        {\itfj{\dom(\plact{\sigma}{(\Xi_i)_{1 \le i\le n}})}{\Gamma}{M}{\seq{\alpha_1,\ldots,\alpha_k}}{A}}
        $}
    \end{gathered}
  \]
  \caption{The definition of $\src(\pi)$}
  \label{fig:def-src-pi}
\end{figure}
The correction of this definition relies on the following easy compatibility
property between resource contexts and the $\dom$ operation:
\begin{proposition}
  Let $\Gamma$ be a context. We have:
  \begin{enumerate}[label=(\alph*),ref=(\alph*)]
    \item given two resource morphism contexts $\itmor{\Xi},\itmor{\Xi'} \refin
      \Gamma$, we have $\dom(\Xi \ctxtplus \Xi') = \dom(\Xi)
      \ctxtplus \dom(\Xi')$;
      
    \item given resource morphism contexts $\itmor{\Xi_1},\ldots,\itmor{\Xi_n}
      \refin \Gamma$ and $\Sgroup_n$, we have
      \[
        \dom(\plact{\sigma}{(\Xi_i)_i}) = \dom(\Xi_1) \ctxtplus \cdots \ctxtplus \dom(\Xi_n)
      \]
  \end{enumerate}
  and similarly for $\cod$.
\end{proposition}
\begin{proof}
  By direct computation.
\end{proof}

We can now start the definition of $\itgpd M$. Its objects are the derivations
$\pi$ of $\itfj{\Theta}{\Gamma}{M}{\alpha}{A}$ and its morphisms between two
objects
\[
  \pi_s\co \itfj{\Theta}{\Gamma}{M}{\alpha}{A}
  \qqand
  \pi_t\co \itfj{\Theta'}{\Gamma}{M}{\alpha'}{A}
\]
are the derivations $\pi$ of $\itmfj{\Xi}{\Gamma}{M}{\phi}{\alpha}{\alpha'}{A}$
such that $\src(\pi) = \pi_s$ and $\tgt(\pi) = \pi_t$. Given two composable
morphisms $\pi_1 \co \itmfj{\Xi_1}{\Gamma}{M}{\phi_1}{\alpha}{\alpha'}{A}$ and 
$\pi_2 \co \itmfj{\Xi_2}{\Gamma}{M}{\phi_2}{\alpha'}{\alpha''}{A}$, their
composition $\pi_2 \circ \pi_1$ using the rules of
\Cref{fig:def-derivation-composition}.
\begin{figure}
  \centering
  \[
    \hss
    \begin{gathered}
      \scalebox{0.7}{$
        \inferrule{\incomment{(x_1:A_1,\ldots,x_n:A_n)\ctxtwf\quad} 
          (\itmrefinfj{\phi_2}{\alpha'}{\alpha''}{A_i})
        }{
          \itmfj{\ldots, \itmrefinfj[x_i]{(\unit{\set{1}},\seq{\phi_2})}{\seq{\alpha'}}{\seq{\alpha''}}{A_i}, \ldots }%
          {}{x_i }{\phi_2}{\alpha'}{\alpha''}{A_i}}
        $}
      \circ
      \scalebox{0.7}{$
        \inferrule{\incomment{(x_1:A_1,\ldots,x_n:A_n)\ctxtwf\quad} 
          (\itmrefinfj{\phi_1}{\alpha}{\alpha'}{A_i})
        }{
          \itmfj{\ldots, \itmrefinfj[x_i]{(\unit{\set{1}},\seq{\phi_1})}{\seq{\alpha}}{\seq{\alpha'}}{A_i}, \ldots }%
          {}{x_i }{\phi_1 }{\alpha }{\alpha' }{A_i}}
        $}
      \\
      =
      \scalebox{0.7}{$
        \inferrule{\incomment{(x_1:A_1,\ldots,x_n:A_n)\ctxtwf\quad} 
          (\itmrefinfj{\phi_2 \circ \phi_1}{\alpha}{\alpha''}{A_i})
        }{
          \itmfj{\ldots, \itmrefinfj[x_i]{(\unit{\set{1}},\seq{\phi_2 \circ \phi_1})}{\seq{\alpha}}{\seq{\alpha''}}{A_i}, \ldots }%
          {}{x_i}{\phi_2 \circ \phi_1}{\alpha}{\alpha''}{A_i}}
        $}
      \\
      \scalebox{0.7}{$
        \inferrule{\inferrule{\pi_{2,1}}{\itmfj{\Xi_2}{}{M}{(\widetilde{\phi}_2 \synlinto \psi_2)}{(\vec{\alpha}' \synlinto \beta')}{(\vec{\alpha}'' \synlinto \beta'')}{A \synto B}}%
          \quad%
          \inferrule{\pi_{2,2}}{\itmfj{\Xi'_2}{}{N}{\widetilde{\phi}_2}{\vec{\alpha}'}{\vec{\alpha}''}{A}}}%
        {\itmfj{\Xi_2 \rmctxtplus\Xi'_2}{}{M\,N}{\psi_2}{\beta'}{\beta''}{B}}
        $}
      \circ
      \scalebox{0.7}{$
        \inferrule{\inferrule{\pi_{1,1}}{\itmfj{\Xi_1}{}{M}{(\widetilde{\phi}_1 \synlinto \psi_1)}{(\vec{\alpha} \synlinto \beta)}{(\vec{\alpha}' \synlinto \beta')}{A \synto B}}%
          \quad%
          \inferrule{\pi_{1,2}}{\itmfj{\Xi'_1}{}{N}{\widetilde{\phi}_1}{\vec{\alpha}}{\vec{\alpha}'}{A}}}%
        {\itmfj{\Xi_1 \rmctxtplus\Xi'_1}{}{M\,N}{\psi_1}{\beta }{\beta' }{B}}
        $}
      \\
      =
      \scalebox{0.7}{$
        \inferrule{\inferrule{\pi_{2,1}\circ\pi_{1,1}}{\itmfj{\Xi_2 \circ \Xi_1}{}{M}{(\widetilde{\phi}_2 \synlinto \psi_2)\circ(\widetilde{\phi}_1 \synlinto \psi_1)}{(\vec{\alpha} \synlinto \beta)}{(\vec{\alpha}'' \synlinto \beta'')}{A \synto B}}%
          \quad%
          \inferrule{\pi_{2,2}\circ\pi_{1,2}}{\itmfj{\Xi'_2\circ\Xi'_1}{}{N}{\widetilde{\phi}_2\circ\widetilde{\phi}_1}{\vec{\alpha}}{\vec{\alpha}''}{A}}}%
        {\itmfj{(\Xi_2 \rmctxtplus\Xi'_2)\circ(\Xi_1 \rmctxtplus\Xi'_1)}{}{M\,N}{\psi_2\circ\psi_1}{\beta}{\beta''}{B}}
        $}
      \\
        \scalebox{.70}{$
          \inferrule{\sigma_2 \in \Sgroup_n\quad \forall i \in
            \set{1,\ldots,n},\;\inferrule*{\pi_{2,i}}{\itmfj{\Xi_{2,i}}{}{M}{\phi_{2,i}}{\alpha'_{\finv\sigma_1(i)}}{\alpha''_{\finv\sigma_1(i)}}{A}}}%
          {\itmfj{\plact{\sigma_2}{(\Xi_{2,i})_{1 \le i\le n}}
            }{}{M}{(\sigma_2,\seq{\phi_{2,1},\ldots,\phi_{2,n}})
            }{\seq{\alpha'_{\finv\sigma_1(1)},\ldots,\alpha'_{\finv\sigma_1(n)}}
            }{\seq{\alpha''_{\finv\sigma_1(\finv\sigma_2(1))},\ldots,\alpha''_{\finv\sigma_1(\finv\sigma_2(n))}}}{A}}
          $}
        \circ
        \scalebox{.70}{$
          \inferrule{n\in \N\quad \sigma_1 \in \Sgroup_n\quad \forall i \in
            \set{1,\ldots,n},\;\inferrule*{\pi_{1,i}}{\itmfj{\Xi_{1,i}}{}{M}{\phi_{1,i}}{\alpha_i}{\alpha'_i}{A}}}%
          {\itmfj{\plact{\sigma_1}{(\Xi_{1,i})_{1 \le i\le n}}
            }{}{M}{(\sigma_1,\seq{\phi_{1,1},\ldots,\phi_{1,n}})
            }{\seq{\alpha_1,\ldots,\alpha_n}
            }{\seq{\alpha'_{\finv\sigma_1(1)},\ldots,\alpha'_{\finv\sigma_1(n)}}}{A}}
          $}
      \\
      =
      \scalebox{.70}{$
        \inferrule{\forall i \in
          \set{1,\ldots,n},\;\inferrule*{\pi_{2,\sigma_1(i)}\circ \pi_{1,i}}{\itmfj{\Xi_{2,\sigma_1(i)}\circ\Xi_{1,i}}{}{M}{\phi_{2,\sigma_1(i)}\circ\phi_{1,i}}{\alpha_i}{\alpha''_i}{A}}}%
        {\itmfj{(\plact{\sigma_2}{(\Xi_{2,i})_{1 \le i\le n}}) \circ (\plact{\sigma_1}{(\Xi_{1,i})_{1 \le i\le n}})
          }{}{M}{(\sigma_2,\seq{\phi_{2,1},\ldots,\phi_{2,n}}) \circ (\sigma_1,\seq{\phi_{1,1},\ldots,\phi_{1,n}})
          }{\seq{\alpha_1,\ldots,\alpha_n}
          }{\seq{\alpha''_{\finv\sigma_1(\finv\sigma_2(1))},\ldots,\alpha''_{\finv\sigma_1(\finv\sigma_2(n))}}}{A}}
        $}
      \\
      \scalebox{.70}{$
        \inferrule{\inferrule{\pi'_2}{\itmfj{\Xi_2,\itmrefinfj[x]{\widetilde{\phi}_2}{\vec{\alpha}'}{\vec{\alpha}''}{A}}{}{M}{\psi_2}{\beta'}{\beta''}{B}}}%
        {\itmfj{\Xi_2}{}{\lambda x.\,M}{(\widetilde{\phi}_2 \synlinto \psi_2)}{(\vec{\alpha}' \synlinto \beta')}{(\vec{\alpha}'' \synlinto \beta'')}{A \synto B}}
        $}
      \circ
      \scalebox{.70}{$
        \inferrule{\inferrule{\pi'_1}{\itmfj{\Xi_1,\itmrefinfj[x]{\widetilde{\phi}_1}{\vec{\alpha}}{\vec{\alpha}'}{A}}{}{M}{\psi_1}{\beta}{\beta'}{B}}}%
        {\itmfj{\Xi_1}{}{\lambda x.\,M}{(\widetilde{\phi}_1 \synlinto \psi_1)}{(\vec{\alpha} \synlinto \beta)}{(\vec{\alpha}' \synlinto \beta')}{A \synto B}}
        $}
      \\
      =
      \scalebox{.70}{$
        \inferrule{\inferrule{\pi'_2 \circ \pi'_1}{\itmfj{\Xi_2 \circ
              \Xi_1,\itmrefinfj[x]{\widetilde{\phi}_2
                \circ\widetilde{\phi}_1}{\vec{\alpha}}{\vec{\alpha}''}{A}}{}{M}{\psi_2
            \circ \psi_1}{\beta}{\beta''}{B}}}%
        {\itmfj{\Xi_2 \circ \Xi_1}{}{\lambda x.\,M}{(\widetilde{\phi}_2
            \synlinto \psi_2) \circ (\widetilde{\phi}_1 \synlinto \psi_1)}{(\vec{\alpha} \synlinto \beta)}{(\vec{\alpha}'' \synlinto \beta'')}{A \synto B}}
        $}
    \end{gathered}
    \hss
  \]
  \caption{The definition of composition of intersection type morphism
    derivations}
  \label{fig:def-derivation-composition}
\end{figure}
Note that these are the only required rules, since, when $\pi_1$ and $\pi_2$ are
composable, they \eq{have the same shape}, because they are derivations for the
same term $M$, and the fact that $\tgt(\pi_1) = \src(\pi_2)$ allows one to infer
other constraints. The rules produce derivations of the adequate type since we
have:
\begin{proposition}
  The followings hold:
  \begin{enumerate}[label=(\alph*),ref=(\alph*)]
    \item given composable $\Xi_1,\Xi_2$ and composable $\Xi'_1,\Xi'_2$ that all
      refine a context $\Gamma$, we have
      \[
        (\Xi_2 \ctxtplus \Xi'_2) \circ (\Xi_1 \ctxtplus \Xi'_1) = (\Xi_2 \circ \Xi_1) \ctxtplus (\Xi'_2 \circ \Xi'_1);
      \]
    \item given $n \in \N$ and $\sigma_1,\sigma_2 \in \Sgroup_n$ and resource
      morphism contexts $\Xi_{1,1},\ldots,\Xi_{1,n}$ and
      $\Xi_{2,1},\ldots,\Xi_{2,n}$ such that they all refine a context $\Gamma$
      and $\Xi_{1,i}$ is composable with $\Xi_{2,\sigma_1(i)}$ for every $i \in
      \set{1,\ldots,n}$, we have
      \[
        (\plact{\sigma_2}{(\Xi_{2,i})_{1 \le i\le n}}) \circ (\plact{\sigma_1}{(\Xi_{1,i})_{1 \le i\le n}})
        =
        \plact{(\sigma_2 \circ \sigma_1)}{(\Xi_{2,\sigma(i)} \circ \Xi_{1,i})_{1 \le i\le n}}
        \zbox.
      \]
  \end{enumerate}
\end{proposition}
\begin{proof}
  By direct computation.
\end{proof}
By a similar inductive definition, we can define the identity $\id_\pi$ of a
derivation $\pi$ of a judgement $\itfj{\Theta}{\Gamma}{M}{\alpha}{A}$. Moreover,
following the definition of inverses for intersection type morphisms, we can
define the inverse $\finv\pi$ of a derivation $\pi$ of a judgement
$\itmfj{\Xi}{\Gamma}{M}{\phi}{\alpha}{\alpha'}{A}$. It is then routine to check that
\begin{proposition}
  The above operations equip $\itgpd M$ with a structure of groupoid.
\end{proposition}

\subsection{The correspondence interpretation/derivation correspondence}

\begin{proof}[Proof of \Cref{thm:isom-terminterp-itgpd}]%
  In fact, we prove the following stronger statement:

  {\it Let $\Gamma$ be a well-typed context. Given a derivation of $\Gamma
    \judge M : A$, there is a canonical morphism of groupoid $K_M \co
    \terminterp M \to \itgpd M$ (\resp $K^\oc_M \co \terminterpoc M \to \itgpdoc M$)
    making the squares of the following diagram commute:
    \begin{equation}
      \label{eq:thm:charact-terminterp:cd}
      \begin{tikzcd}
        \ctxtinterp\Gamma
        \ar[d,"K_\Gamma"']
        &
        \terminterp M
        \ar[l,"{\stdisp[\altidisp{\ctxtinterp\Gamma}{l}]{\terminterp M}}"']
        \ar[r,"{\stdisp[\altidisp{\typeinterp A}{r}]{\terminterp M}}"]
        \ar[d,"K_M"]
        &
        \typeinterp A
        \ar[d,"K_A"]
        \\
        \itgpd\Gamma
        &
        \itgpd M
        \ar[l,"{\stdisp[l]{M}}"]
        \ar[r,"{\stdisp[r]{M}}"']
        &
        \itgpd A
      \end{tikzcd}
      \qquad
      \left(
        \text{\resp}
        \begin{tikzcd}
          \ctxtinterp\Gamma
          \ar[d,"K_\Gamma"']
          &
          \terminterpoc M
          \ar[l,"{\stdisp[l]{\terminterpoc M}}"']
          \ar[r,"{\stdisp[r]{\terminterpoc M}}"]
          \ar[d,"K^\oc_M"]
          &
          \oc \typeinterp {A}
          \ar[d,"K^\oc_{A}"]
          \\
          \itgpd\Gamma
          &
          \itgpdoc M
          \ar[l,"{\stdispoc[l]{M}}"]
          \ar[r,"{\stdispoc[r]{M}}"']
          &
          \itgpdoc {A}
        \end{tikzcd}
      \right)
      \zbox.
    \end{equation}
  }

  We prove it by induction on a derivation of $\Gamma \judge M : A$:
  \begin{itemize}
  \item in the case of the variable typing rule $x_1: A_1,\ldots,x_n: A_n \judge
    x_i : A_i$, $K_M$ is defined as the functor sending $u \co a \to a'
    \in \typeinterp A = \terminterp {x_i}$ to the unique derivation of
    \[
      \begin{split}
        &\hskip-3em (\itmrefinfj[x_1]{\seq{} }{\seq{} }{\seq{} }{A_1},\ldots,
          \itmrefinfj[x_i]{(\unit{\set{1}},\seq{K_{A_i}(u)})}{\seq{K_{A_i}(a)}}{\seq{K_{A_i}(a')}}{A_i},\\
        &\ldots,\itmrefinfj[x_n]{\seq{} }{\seq{} }{\seq{}}{A_n}) \refin \Gamma
          \judgel x : \itmor{K_{A_i}(u)} \itsco \itype{K_{A_i}(a)} \itsto
          \itype{\seq{K_{A_i}(a')}} \refin A_i
          \zbox.
      \end{split}
    \]
    It is immediate that the squares of~\eqref{eq:thm:charact-terminterp:cd}
    commute for this definition;
  \item in the case of an application $\Gamma \judge M\, N: B$, we
    have that there exists a unique simple type~$A$ such that $\Gamma
    \judge N : A$, so that the typing of the application $M\, N$ is constructed
    from two derivations $\Gamma \judge M : A \synto B$ and $\Gamma \judge N :
    A$. We then have by the rules for intersection type morphism judgements that
    $\itgpd {M\,N}$ is the pullback
    \[
      \begin{tikzcd}
        \itgpd {M\,N}
        \ar[rr,"P_N",dashed]
        \ar[d,"P_M"',dashed]
        \phar[rd,very near start,"\drcorner"]
        &
        &
        \itgpdoc{N}
        \ar[d,"{\stdisp[r]{N}}"]
        \\
        \itgpd M
        \ar[r,"{\stdisp[r]M}"']
        &
        \itgpd{A \synto B}
        \ar[r,"P^{A\synto B}_A"']
        &
        \itgpdoc A
        \zbox.
      \end{tikzcd}
    \]
    where $P^{A\synto B}_A$ is the functor projecting a derivation of
    $\itmrefinfj{\widetilde{\phi} \synlinto
      \psi}{\vec\alpha\linto\beta}{\vec\alpha'\linto\beta}{A \synto B}$ to the
    associated derivation
    $\itmrefinfj{\widetilde{\phi}}{\vec\alpha}{\vec\alpha'}{A}$. Similarly, by
    considering again the definition of the groupoid $\terminterp{M\,N}$, we see
    that it can be alternatively expressed as the pullback
    \[
      \begin{tikzcd}
        \terminterp {M\,N}
        \ar[rr,"P'_N",dashed]
        \ar[d,"P'_M"',dashed]
        \phar[rd,very near start,"\drcorner"]
        &&
        \terminterp{N}^{\oc}
        \ar[d,"{\stdisp[r]{\terminterpoc{N}}}"]
        \\
        \terminterp M
        \ar[r,"{\stdisp[r]{\terminterp M}}"']
        &
        \oc\typeinterp{A} \times \typeinterp B
        \ar[r,"{\pl}"']
        &
        \oc \typeinterp A
        \zbox.
      \end{tikzcd}
    \]
    Then, using $K^\oc_A$ and the inductively defined $K_M$, $K^\oc_N$, we build
    an isomorphism between the underlying cospans of these pullbacks, so that we
    get a factorizing isomorphism $K_{M\,N} \co \terminterp {M\,N} \to
    \itgpd{M\,N}$. Concerning the commutativity condition we have the diagram
    \[
      \begin{tikzcd}
        \ctxtinterp \Gamma
        \ar[d,"K_{\Gamma}"']
        &
        \ctxtinterp \Gamma \times \ctxtinterp \Gamma
        \ar[l,"(\ctxtplus)"']
        \ar[d,"K_{\Gamma}\times K_{\Gamma}"{description}]
        &
        \terminterp M \times \terminterpoc N
        \ar[l,"{\stdisp[l]{\terminterp M} \times \stdisp[l]{\terminterpoc N}}"']
        \ar[d,"K_{M}\times K_N"{description}]
        &
        \terminterp {M\,N}
        \ar[l,"{(P'_M,P'_N)}"']
        \ar[d,"K_{M\,N}"]
        \\
        \itgpd \Gamma
        &
        \itgpd \Gamma \times \itgpd \Gamma
        \ar[l,"(\ctxtplus)"]
        &
        \itgpd M \times \itgpdoc N
        \ar[l,"{\stdisp[l]{M}\times\stdispoc[l]{N}}"]
        &
        \itgpd {M\,N}
        \ar[l,"{(P_M,P_N)}"]
      \end{tikzcd}
    \]
    where each rectangle commutes and the top row is precisely
    $\stdisp[l]{\terminterp{M\,N}}$ and the bottom row $\stdisp[l]{M\,N}$.

    On the side of $B$, we have the diagram
    \[
      \begin{tikzcd}
        \terminterp{M\,N}
        \ar[r,"P'_M"]
        \ar[d,"K_{M\,N}"']
        &
        \terminterp M
        \ar[r,"{\stdisp[r]{\terminterp{M}}}"]
        \ar[d,"K_{M}"{description}]
        &
        \oc \typeinterp {A} \times \typeinterp{B}
        \ar[r,"\pr"]
        \ar[d,"K_{A\synto B}"{description}]
        &
        \typeinterp {B}
        \ar[d,"K_{B}"]
        \\
        \itgpd{M\,N}
        \ar[r,"P_M"']
        &
        \itgpd{M}
        \ar[r,"{\stdisp[r]{M}}"']
        &
        \itgpd {A \synto B}
        \ar[r,"P^{A\synto B}_B"']
        &
        \itgpd B
      \end{tikzcd}
    \]
    where $P^{A\synto B}_B$ is defined like $P^{A\synto B}_A$, where every
    rectangle commutes, where the top row is $\stdisp[r]{\terminterp{M\,N}}$ and
    the bottom row is $\stdisp[r]{M\,N}$. Which concludes the proof of the
    commutativity conditions;
    
  \item in the case of a $\lambda$-abstraction $\Gamma\judge \lambda x.M: A
    \synto B$, we get by induction the commutative diagram
    \[
      \begin{tikzcd}
        \ctxtinterp{\Gamma,x:A}
        \ar[d,"K_{(\Gamma,x:A)}"']
        &
        \terminterp M
        \ar[l,"{\stdisp[l]{\terminterp M}}"']
        \ar[r,"{\stdisp[r]{\terminterp M}}"]
        \ar[d,"K_{M}"{description}]
        &
        \typeinterp B
        \ar[d,"K_{B}"]
        \\
        \itgpd{\Gamma,x:A}
        &
        \itgpd M
        \ar[l,"{\stdisp[l]{M}}"]
        \ar[r,"{\stdisp[r]{M}}"']
        &
        \itgpd B
      \end{tikzcd}
    \]
    so that, using the isomorphisms $\ctxtinterp{\Gamma,x:A} \cong
    \ctxtinterp\Gamma \times \oc \typeinterp A$, $\itgpd{\Gamma,x:A} \cong \itgpd
    \Gamma \times \itgpdoc A$, $\oc \typeinterp A \times \typeinterp B \cong
    \typeinterp{A \synto B}$ and $\itgpdoc A \times \itgpd B \cong \itgpd {A
      \synto B}$, we are able to get a similar commutative diagram for $\lambda
    x. M$;

  \item finally, we define the multilinear interpretation of a judgement $\Gamma
    \judge M : A$ from the above cases: we have a commutative diagram
    \[
      \begin{tikzcd}
        \ctxtinterp \Gamma
        \ar[d,equals]
        &
        &
        \terminterpoc M
        \ar[ll,"{\stdisp[l]{\terminterpoc M}}"']
        \ar[r,"{\stdisp[r]{\terminterpoc M}}"]
        \ar[d,equals]
        &
        \oc \typeinterp A
        \ar[d,equals]
        \\
        \ctxtinterp \Gamma
        \ar[d,"K_\Gamma"{description}]
        &
        \oc \ctxtinterp \Gamma
        \ar[l,"\ctxtmu_\Gamma"']
        \ar[d,"\oc K_\Gamma"{description}]
        &
        \oc \terminterp {M}
        \ar[l,"{\oc\stdisp[l]{\terminterp{M}}}"']
        \ar[r,"{\oc\stdisp[r]{\terminterp{M}}}"]
        \ar[d,"\oc K_{M}"{description}]
        &
        \oc \typeinterp A
        \ar[d,"\oc K_A"{description}]
        \\
        \itgpd \Gamma
        \ar[d,equals]
        &
        \oc \itgpd \Gamma
        \ar[l,"{\itsctxtmu_\Gamma}"']
        &
        \oc \itgpd {M}
        \ar[l,"{\oc\stdisp[l]{M}}"']
        \ar[r,"{\oc\stdisp[r]{M}}"]
        \ar[d,"\sim"{description}]
        &
        \oc \itgpd A
        \ar[d,"\sim"{description}]
        \\
        \itgpd \Gamma
        &
        &
        \itgpdoc {M}
        \ar[ll,"{\stdispoc[l]{M}}"']
        \ar[r,"{\stdispoc[r]{M}}"]
        &
        \itgpdoc {A}
      \end{tikzcd}
    \]
    where the morphism $\oc \itgpd {M} \xto\sim \itgpdoc {M}$ is basically the
    multilinear introduction rule---a sequence $\seq{\pi_i}_{1\le i \le n}$ of
    derivations $\pi_i$ of $\itfj{\Theta_i}{\Gamma}{M}{\alpha_i}{A}$ is mapped
    to the derivation of $\itfj{\Theta_1 \ctxtplus \cdots \ctxtplus
      \Theta_n}{\Gamma}{M}{\seq{\alpha_1,\ldots,\alpha_n}}{A}$ and similarly for
    sequences of morphism derivations---and the morphism $\oc \itgpd A \xto\sim
    \itgpdoc{A}$ is similarly the multilinear refinement introduction rule, and
    $\itsctxtmu$ is the functor mapping a sequence of resource contexts
    $\seq{\Theta_1,\ldots,\Theta_n} \in \Ob(\oc\itgpd\Gamma)$ to $\Theta_1
    \rctxtplus \ldots \rctxtplus \Theta_n$ and a morphism
    $(\sigma,\seq{\Xi_i}_{1 \le i \le n}) \in \Ar(\itgpd\Gamma)$ between two
    sequences $\seq{\Theta_1,\ldots,\Theta_n}$ and
    $\seq{\Theta'_1,\ldots,\Theta'_n}$ to $\rmcplact{\sigma}{(\Xi_i)_{1 \le i
        \le n}}$. It is quite immediate to check that every rectangle commutes.
    \qedhere
  \end{itemize}
\end{proof}

\section{Postponed proofs for relational collapses}

\label{app:aux_collapse}

\subsection{Functoriality of the collapse to $\Rel$}
\label{subsec:funct_rel}

We give the postponed proof of the following result:

\makeatletter
\@ifundefined{functrel}{}{\functrel*}%
\makeatother
\begin{proof}
It is obvious that the identity span $A \ot A \to A$ is sent to the
identity relation on $\syms{A}$. For functoriality, it is obvious by
definition that $\syms{T\odot S} \subseteq \syms{T} \circ \syms{S}$ --
but the other direction is not, since composition in $\Thin$ is more
constrained than in $\Rel$.

So consider $(\ca, \cb) \in \syms{S}$ and $(\cb, \cc) \in \syms{T}$.
By
definition, there are $s \in S$ such that $\ca = \class{s_A}$ and $\cb
= \class{s_B}$, and $t \in T$ such that $\cb = \class{t_B}$ and $\cc =
\class{t_C}$. Since $\class{s_B} = \class{t_B}$ those two are
symmetric, but they might not be equal, meaning that the pair $(s, t)$
may not be a valid element of $T \odot S$. However, by Lemma
\ref{lem:comp_upto_sym}
there must be $\varphi^S \in S[s, s']$ and $\varphi^T \in T[t, t']$
such that $s'_B = t'_B$ and we can now form $(s', t') \in T\odot S$
with $\class{(s', t')_A} = s'_A = \ca$ and $\class{(s', t')_C} =
\class{t'_C} = \cc$ as required. 
\end{proof}

\subsection{Bijection for the quantitative collapse}
\label{app:bijquant}

\makeatletter
\@ifundefined{witbij}{}{\witbij*}%
\makeatother
\begin{proof}
From $(\theta_A^-, s, \theta_B^+)$ and $(\Omega_B^-, t, \Omega_B^+)$,
we can apply Lemma \ref{lem:comp_upto_sym} and compose $s$ and $t$ via
$\Omega_B^- \circ \theta_B^+$, giving us unique $\omega^S, \nu^T$
such that the big rectangle commutes, $\omega^S_A$ negative and
$\nu^T_C$ positive. We get $\Theta$ as either path around the
rectangle, and $\psi^-_A, \psi^+_C$ by composition. Reciprocally, from
$(\psi_A^-, s', \Theta, t', \psi_C^+)$ we obtain uniquely the remaining
data by Proposition \ref{prop:act_strat}. 
\end{proof}

\section{Seely functors and their Kleisli lifting}
\label{app:seely_lifting}

Here we include a few folkore results that we required regarding the
adequate definition of morphisms between Seely categories, along with
the fact that they admit a lifting to cartesian closed functors between
the Kleisli categories.

\begin{definition}\label{def:seely_functor}
Consider $\catC$ and $\catD$ two Seely categories. 

A \textbf{Seely functor} $F : \catC \to \catD$ is a functor,
additionally equipped with isomorphisms 
\[
\begin{array}{rcrcl}
t^\oc_A &:& \oc F A &\to& F \oc A\\
t^\tensor_{A, B} &:& F A \tensor FB &\to& F(A\tensor B)\\
t^\with_{A, B} &:& FA \with FB &\to & F(A \with B)\\
t^\lin_{A, B} &:& FA \lin FB &\to& F(A \lin B) 
\end{array}
\]
such that $t^\oc_A$ is natural in $A$ and $t^\tensor_{A, B}$ is natural
in $A$ and $B$, and subject to the following coherence conditions:
\[
\xymatrix{
&\oc FA	\ar[rr]^{t^\oc_A}
	\ar[dl]_{\delta^\catD_{FA}}&&
F\oc A	\ar[dr]^{F \delta^\catC_A}\\
\oc \oc FA
	\ar[rr]_{\oc t^\oc_A}&&
\oc F \oc A
	\ar[rr]_{t^{\oc}_{\oc A}}&&
F\oc \oc A
}
\xymatrix{
\oc FA	\ar[rr]^{t^\oc_A}
	\ar[dr]_{\epsilon^\catD_{FA}}&&
F\oc A	\ar[dl]^{F \epsilon^\catC_A}\\
&FA
}
\]
\[
\xymatrix@R=20pt@C=20pt{
&\oc FA \tensor \oc F B
	\ar[dl]_{\see^\catD_{FA, FB}}
	\ar[dr]^{t^\oc_A \tensor t^\oc_B}\\
\oc (FA \with FB)
	\ar[d]_{\oc t^\with_{A, B}}&&
F\oc A \tensor F\oc B
	\ar[d]^{t^\tensor_{\oc A, \oc B}}\\
\oc F(A\with B)
	\ar[dr]_{t^\oc_{A \with B}}&&
F(\oc A \tensor \oc B)
	\ar[dl]^{F \see^\catC_{A, B}}\\
&F\oc (A \with B)
}
\xymatrix{
&FA \with FB
	\ar[dl]_{\pi^\catD_1}
	\ar[dd]|{t^\with_{A, B}}
	\ar[dr]^{\pi^\catD_2}\\
FA && FB\\
&F(A \with B)
	\ar[ul]^{F \pi^\catC_1}
	\ar[ur]_{F \pi^\catC_2}
}
\]
\[
\xymatrix{
(FA \lin FB) \tensor FA
	\ar[r]^{t^{\lin}_{A, B} \tensor FA}
	\ar[dr]_{\evm^\catD_{FA, FB}}&
F(A\lin B) \tensor FA
	\ar[r]^{t^\tensor_{A\lin B, A}}&
F((A \lin B) \tensor A)
	\ar[dl]^{F\evm^\catC_{A, B}}\\
&FB
}
\]
\end{definition}

The main interest of those is that they lift to cartesian closed
functors between the Kleisli categories:

\begin{theorem}
Consider $\catC$ and $\catD$ two Seely categories, and $F : \catC \to
\catD$ a Seely functor.

Then, defining $F_\oc(A) = F(A)$ on objects and $F_\oc(f) = Ff \circ
t^\oc_A$ for $f \in \catC[\oc A, B]$, we get
\[
F_\oc : \catC_\oc \to \catD_\oc
\]
a cartesian closed functor.
\end{theorem}
\begin{proof}
We must show that products and arrows are preserved up to (canonical)
isomorphism. For that, we construct the following morphisms in $\catD$:
\[
\begin{array}{rclcrcl}
k^\with_{A, B} &=& t^\with_{A, B} \circ \epsilon^\catD_{FA \with FB}
&:& \oc (FA \with FB) &\to& F(A\with B)\\
k^\tto_{A, B} &=& t^\lin_{\oc A, B} \circ ((t^\oc_A)^{-1} \lin B) \circ
\epsilon^\catD_{\oc FA \lin FB}
&:& \oc (\oc FA \lin FB)&\to& F(\oc A \lin B)\,,
\end{array}
\]
which we regard as $k^\with_{A, B} \in \catD_{\oc}[FA \with FB,
F(A\with B)]$ and $k^\tto_{A, B} \in \catD_{\oc}[FA \tto FB, F(A\tto
B)]$ where $A \tto B = \oc A \lin B$. By construction those are
isomorphisms, and to show canonicity we must prove (in $\catD_\oc$) the
diagrams corresponding to the last two diagrams of Definition
\ref{def:seely_functor}, which is a lengthy diagram chase.
\end{proof}


\else
\fi

\end{document}